\documentclass[a4paper,12pt]{article}
\usepackage[utf8]{inputenc}
\usepackage[
 top=1.25in,
  bottom=1.25in,
  left=1.3in,
  right=1.3in
]{geometry}
\usepackage{booktabs}
\usepackage{graphicx} 
\usepackage{amsmath,amssymb,amsthm}
\usepackage{url}
\usepackage[shortlabels]{enumitem}

\PassOptionsToPackage{table,dvipsnames}{xcolor} 
\usepackage{xcolor}

\usepackage{tikz}
\usetikzlibrary{math,calc,mindmap,backgrounds,patterns}
\usepackage{pgfplotstable}   
\usepackage{adjustbox}

\definecolor{HTML0}{HTML}{D62828} 
\definecolor{HTML4}{HTML}{F4A261} 
\definecolor{HTML1}{HTML}{3F5F8F} 
\definecolor{HTML6}{HTML}{4CB7B5} 
\definecolor{HTML3}{HTML}{5E1A5E} 

\usepackage{mathtools}
\usepackage{newpxtext,newpxmath}

\usepackage{bm}
\usepackage{bbm}
\usepackage{arydshln}
\usepackage{setspace}
\usepackage{array}
\usepackage{makecell}

\DeclareMathAlphabet{\mathpzc}{OT1}{pzc}{m}{it}

\usepackage{pgfplots}
\usepackage{subcaption}
\pgfplotsset{compat=1.18}

\usepackage{colortbl}
\usepackage{accents}
\usepackage{import}

\usepackage[authoryear]{natbib}
\usepackage[hidelinks]{hyperref}
\usepackage[noabbrev,capitalize,nameinlink]{cleveref}

\usepackage[ruled,vlined]{algorithm2e}
\crefname{algocf}{Algorithm}{Algorithms}

\theoremstyle{plain}
\newtheorem{theorem}{Theorem}
\newtheorem*{theorem*}{Theorem}
\newtheorem{proposition}{Proposition}
\newtheorem*{proposition*}{Proposition}
\newtheorem{lemma}{Lemma}
\newtheorem{corollary}{Corollary}
\newtheorem*{corollary*}{Corollary}
\newtheorem{claim}{Observation}
\newtheorem*{claim*}{Observation}

\theoremstyle{definition}
\newtheorem{definition}{Definition}
\newtheorem*{definition*}{Definition}
\newtheorem*{notation*}{Notation}

\theoremstyle{remark}
\newtheorem{remark}{Remark}
\newtheorem*{remark*}{Remark}

\newtheorem*{example*}{Example}

\newcommand\independent{\protect\mathpalette{\protect\independenT}{\perp}}
\def\independenT#1#2{\mathrel{\rlap{$#1#2$}\mkern2mu{#1#2}}}

\usepackage{manyfoot}
\SetFootnoteHook{\hspace*{-1.8em}}
\DeclareNewFootnote{bl}[gobble]
\newcommand{\Poisson}{\operatorname{Poisson}}
\newcommand{\Ggame}{G}
\newcommand{\anumvec}{\mathbf{m}}
\newcommand{\anum}{m}
\newcommand{\action}{a}
\newcommand{\actionvec}{\mathbf{a}}
\newcommand{\actionset}{{A}_i}

\newcommand{\preferencerelation}{\succsim_i}
\newcommand{\profileset}{\mathbf A}

\newcommand{\agentset}{N}
\newcommand{\V}[2]{ (#1,\dots, #2) }

\newcommand{\potpref}{\unrhd}
\newcommand{\potprefstrict}{\rhd}

\newcommand{\G}[2]{\mathbf G_{#1,#2}}
\newcommand{\g}{g_{n,\anumvec,\succsim}}
\newcommand{\gframe}{g_{n,\anumvec,\varnothing}}
\newcommand{\gf}{h}
\newcommand{\z}{d}
\newcommand{\Z}{D}

\newcommand{\rank}{\text{rank}}
\newcommand{\indicator}[1]{\mathbbm{1}_{#1}}

\title{\textsc{Satisficing Equilibrium}}

\author{
\begin{tabular}{c @{\hspace{4em}} c}
\large Bary S.R. Pradelski & \large Bassel Tarbush \vspace{-0.4cm} \\ 
{\normalsize CNRS} & {\normalsize Oxford}
\end{tabular}
}
\vspace{0.4cm}
\date{\today}

\begin{document}

\maketitle
\thispagestyle{empty}

\onehalfspacing
\begin{abstract} 
\noindent
In a satisficing equilibrium each agent $i$ plays one of her top $k_i$ actions in response to the actions of the other agents. Our concept unifies models of bounded rationality and yields predictions that differ from canonical solution concepts. We study its theoretical properties and show that it provides sharp predictions, exists in most games as well as in a broad new class of economic environments, admits standard epistemic and dynamic foundations, and is empirically falsifiable.
\end{abstract}
\footnotebl{{BP}: CNRS, Maison Fran\c{c}aise d'Oxford; Department of Economics, University of Oxford.\\
 \texttt{bary.pradelski@cnrs.fr}. \\
 {BT}: Merton College and Department of Economics, University of Oxford.\\
 \texttt{bassel.tarbush@economics.ox.ac.uk}}
\footnotebl{\emph{Keywords}: equilibrium, satisficing, bounded rationality.}
\footnotebl{\emph{Thanks}: We are grateful for
comments and suggestions by Jean Baccelli, Miguel Ballester, Jean-Paul Carvalho, Vincent Crawford, Bruno de Albuquerque Furtado, Francesc Dilm\'{e}, Peter Es\H{o}, Françoise Forges, Alkis Georgiadis-Harris, Olivier Gossner, En Hua Hu, Simon Jantschgi, Bernhard Kasberger, Davide Legacci, Annie Liang, Alistair Macaulay, Panayotis Mertikopoulos, Jonny Newton, John Quah, Philip Reny, Ariel Rubinstein, Anna Sanktjohanser,  Raimundo Saona, Ludvig Sinander, Alex Teytelboym, Jörgen Weibull, Peyton Young, and seminar and conference participants at Carg\`{e}se, Columbia, Durham, EC'25 (Stanford), Grenoble, Oxford, Parisian Game Theory Seminar, Princeton, and Royal Holloway.}

\newpage

\section{Introduction}\label{sec:introduction}

Much of economics identifies rationality with optimization, as exemplified by \citeauthor{nash1950equilibrium}'s equilibrium concept (\citeyear{nash1950equilibrium}). Yet, since \citeauthor{Simon1947}'s foundational work on bounded rationality (\citeyear{Simon1947}), substantial theoretical and empirical evidence has accumulated showing that subjects frequently do not optimize and do not play Nash equilibrium in games.\footnote{See, for example, \cite*{basu1994traveler,goeree2001ten,arad201211,crawford2013structural}. Note that this list is far from exhaustive.} This raises a fundamental question: what does equilibrium look like when agents do not fully optimize? 

We address this question by introducing the concept of \emph{satisficing equilibrium}.\footnote{The portmanteau \emph{satisfice}---to satisfy and to suffice---has been coined in its modern-day usage by \citeauthor{simon1997models} to describe a decision maker ``who chooses an alternative that meets or exceeds specified criteria, but that is not guaranteed to be either unique or in any sense the best'' (\citeyear{simon1997models}).} We associate each agent~$i\in\agentset$ with a positive integer $k_i$ and say that $i$ is \emph{$k_i$-satisficed} if her chosen action is among her top~$k_i$ actions in response to the others’ actions. The parameters $k_i$ measure the relaxation from optimization (where the latter corresponds to $k_i=1$), and may vary across agents and environments.\footnote{This is in line with empirical evidence; see, e.g., \citet*{benjamin2013behavioral,dardanoni2020inferring}.} 
We let $\mathbf k:=(k_i)_{i\in\agentset}$ and define:
\vspace{0.2cm}
\begin{quote} 
In a \emph{$\mathbf k$-satisficing equilibrium}, each agent~$i\in \agentset$ 
plays one of her top $k_i$ actions in response to the actions of the other agents.
\end{quote}
\vspace{0.2cm}
\noindent 
A $\V{1}{1}$-satisficing equilibrium (that is, $k_i = 1$ for all $i \in \agentset$) coincides with a pure Nash equilibrium. A $\V{2}{2}$-satisficing equilibrium is an action profile such that for each agent the action played is either a best response or a second-best response to the actions of the other agents.

Satisficing equilibrium preserves the unilateral stability logic of pure Nash equilibrium, while relaxing optimization in an ordinal way through the parameters $k_i$, which quantify heterogeneous departures from best-response behavior. Before outlining our theoretical results on the structure and foundations of satisficing equilibrium, we discuss its central features.

First, the concept relies only on ordinal comparisons. We thus eschew the question of preference intensity, which may be important in some settings.\footnote{Note that cardinal utilities are often natural in experiments as they are used to induce preferences \citep[cf.][]{goeree2001ten}, but many natural settings do not feature (easily measurable) preference intensity. Pure $\epsilon$-equilibrium may be a relevant notion when considering departures from optimization in a cardinal setting. Pure $\epsilon$-equilibrium and satisficing equilibrium can yield opposing predictions in some games (for example, in the 11--20 game of \citet{arad201211} discussed in \cref{sec:otherconcepts}).} Ordinal rankings are the only features recoverable from revealed-preference data and underpin several prominent models of bounded rationality, including consideration sets \citep{varian1980model} and choice from lists \citep*{rubinstein2006model,Caplin2011}.

Second, we focus on pure actions. This follows naturally from our restriction to ordinal preferences, avoids the interpretative difficulties of mixed equilibria,\footnote{See discussion in \citet{rubinstein1991comments} and \citet{osborne1998games}.} and allows for simple epistemic foundations that do not require belief hierarchies (see \cref{sec:emergence}). Our approach to capturing bounded rationality thus differs from solution concepts where randomization plays a central role.\footnote{For example,  quantal-response equilibrium  \citep{mckelvey1995quantal}, sampling equilibrium \citep{osborne1998games}, and $M$ equilibrium \citep{goeree2021m}.}

Third, we depart from optimization without committing to a specific model of (strategic) thinking.\footnote{This is in contrast to, for example, rationalizability \citep{bernheim1984rationalizable,pearce1984rationalizable} and level-$\mathpzc{k}$ thinking or cognitive-hierarchy models \citep*{stahl1993evolution,nagel1995unraveling,camerer2004cognitive}.}
We impose only a monotonicity restriction: if an agent is satisficed with an action, then she is also satisficed with any action that induces an outcome weakly preferred to it \citep[see][for a related assumption]{goeree2021m}. Notably, \cite*{Bar22} show that a broad class of bounded-rationality theories (including the aforementioned theories of consideration sets and choice from lists) are observationally equivalent---that is, empirically indistinguishable in a revealed-preference sense---to modeling an agent $i$ as choosing among her top~$k_i$ actions.\footnote{For recent surveys on bounded rationality, see \citet*{artinger2022satisficing,Clippel2024}. See Supplementary Appendix IV for an extended review of the related literature.} For instance, if an agent $i$ facing $m_i$ possible choices has  limited evaluation time, she may only evaluate $m_i-k_i+1$ of them; then maximizing within this restricted set results in her selecting one of her top~$k_i$ actions. Similarly, if an agent is presented with $m_i$ possible choices as a list, search costs may again result in her effectively choosing the maximum among the first $m_i-k_i+1$ entries.

Relatedly, observe that an agent choosing a top $k_i>1$ action  requires no more knowledge or sophistication than optimization ($k_i=1$) and, in fact, often less. In particular, she does not need to evaluate or even know all her available actions, does not need to know the rank of her chosen action, and may (possibly wrongly) believe that she is playing optimally. Having defined satisficing equilibrium and developed a preliminary understanding of its features, we next discuss our results.

\subsection{Overview of results}

We examine the theoretical properties of satisficing equilibrium. We show that satisficing equilibrium provides sharp predictions, exists in most games as well as in a broad new class of economic environments, admits standard epistemic and dynamic foundations, and is empirically falsifiable.

In \cref{sec:structural}, we study the structural properties of satisficing equilibrium. The set of satisficing equilibria is nested in $\mathbf{k}$: increasing $\mathbf{k}$ expands the set of equilibria and thus provides an increasingly permissive relaxation, starting from pure Nash equilibrium (\cref{observation:nested}). 
Two questions naturally arise for a given $\mathbf{k}$: (i) what is the precision of $\mathbf{k}$-satisficing equilibrium, and (ii) does a $\mathbf{k}$-satisficing equilibrium always exist? Using \citet{selten1991properties}'s definition of precision, namely how large the set of predicted outcomes is relative to the total set, we show that $\mathbf{k}$-satisficing equilibrium may have low precision in some games (\cref{prop:preditiveprecision}), but has high precision in all but a vanishing fraction of games (\cref{thm:predictiveprecision}). By ``all but a vanishing fraction'' we mean that the fraction of games in which the property fails is small and vanishes rapidly in the number of agents. These results show that, while satisficing equilibrium is more permissive than pure Nash equilibrium, it nevertheless yields sharp predictions in most games. In fact, for non-trivial $\mathbf k$, $\mathbf k$-satisficing equilibrium has the same asymptotic precision as pure Nash equilibrium.

Turning to existence, we show that satisficing equilibria may fail to exist in some games (\cref{prop:nonexistence}), but they exist in most games (\cref{thm:one_is_enough}). In fact, a satisficing equilibrium where one agent is 2-satisficed and all other agents are 1-satisficed (that is, best-respond) exists in all but a vanishing fraction of games. Consequently, most games admit a $\V{2}{2}$-satisficing equilibrium (\cref{cor:ksatisficing}). This is in contrast to pure Nash equilibrium, which exists only in approximately $1-1/e\approx 63\%$ of games \citep*{arratia1989two,Rin00}. Together, these results suggest that satisficing equilibrium may explain stable patterns of play in situations in which agents may not fully optimize.

Moreover, we provide a sufficient condition for existence. We define the class of \emph{$\mathbf{k}$-approximate ordinal potential games} in which $\mathbf k$-satisficing equilibria are guaranteed to exist (\cref{thm:approximatepotential}). This generalizes the class of potential games \citep{monderer1996potential}, and allows us to consider a wider class of economic environments. As an example of our new class of games, we consider a non-potential---but approximately potential---oligopoly production environment in the spirit of \cite*{rosenthal1973class}.

In \cref{sec:emergence}, we develop the foundations for our concept by considering how satisficing equilibria may emerge. In \cref{obs:static}, we provide epistemic conditions for satisficing equilibrium that are similar to those given by \citet{aumann1995epistemic} for pure Nash equilibrium. We show that only mutual knowledge of action choices is required, but no knowledge of others' choice behavior or preferences is needed (which is in contrast to the more demanding epistemic conditions for mixed Nash equilibrium).  Consequently, $\mathbf{k}$-satisficing equilibrium survives the announcement test, mediator recommendations, and implicit or explicit agreements; tests commonly used to support Nash equilibrium \citep[see][]{Hol04}.

Next, we provide dynamic foundations for equilibrium play \citep[see, e.g.,][]{weibull1995evolutionary,fudenberg1998theory,young2004strategic,samuelson2016game}. We study a simple dynamic in which each agent chooses one of her top $k_i \ge 2$ actions in response to the current action profile. We show that, in all but a vanishing fraction of games, the dynamic converges almost surely in finite time to a $\mathbf{k}$-satisficing equilibrium (\cref{prop:dynamic}).

In \cref{sec:characterization}, we turn to the empirical content of satisficing equilibrium. Building on \citet{sprumont2000testable} and \citet*{Bar22}, we construct a polynomial-time algorithm to test whether a given data set is consistent with $\mathbf{k}$-satisficing equilibrium for a given $\mathbf k$  (\cref{thm:testing}). Hence, $\mathbf{k}$-satisficing equilibrium is falsifiable, and our result provides the basis for future empirical work. In particular, our test can be adapted to obtain a polynomial-time test that finds the smallest possible $\mathbf k$ that is consistent with $\mathbf k$-satisficing equilibrium in a given data set (\cref{cor:testing}), and can therefore be used to explore comparative statics across games and agents.
 
Taken together, our results show that satisficing equilibrium unifies behavioral realism---allowing agents to behave in boundedly rational, heterogeneous, and context-dependent ways---with the discipline of equilibrium theory, and offers a tractable equilibrium concept for analyzing environments where agents satisfice rather than optimize.  

\subsection{Illustrative examples and comparison with other concepts}\label{sec:otherconcepts}

We next illustrate satisficing equilibrium by studying two example games. The games also allow us to compare satisficing equilibrium with  several of the aforementioned concepts. 
 
We first consider the 11--20 game by \citeauthor{arad201211}.\footnote{See Supplementary Appendix III for the bimatrix representations of the discussed games.}

\begin{quote}
\textbf{The 11--20 game \citep{arad201211}.} You and another player are playing a game in which each player requests an amount of money. The amount must be (an integer) between 11 and 20 shekels. Each player will receive the amount he requests. A player will receive an additional amount of 20 shekels if he asks for exactly one shekel less than the other player.
\end{quote}

It is (implicitly) assumed that each player's preferences are such that more shekels for oneself is considered strictly better. The game does not admit a $(1,1)$-satisficing equilibrium---that is, there is no pure Nash equilibrium. A player's best response is always one less than her opponent's action, unless the opponent plays their lowest action in which case the player's best response is her highest action. The $(2,2)$-satisficing equilibria are $(20,20),(19,20),(20,19)$ (see \cref{fig:1120_concepts_main}(a)). \cref{fig:1120_concepts_main}(b) shows the empirical frequency with which the action profiles were played by the subjects in \cite*{arad201211}'s experiment.\footnote{ \cite*{goeree2018noisy} corroborate these findings in further experiments.} Observe that subjects tended to play high actions, which is consistent with $(k,k)$-satisficing equilibrium for $k\leq 4$.

To compare satisficing equilibrium with other solution concepts, we must commit to a cardinal representation of the 11-20 game. To this end, suppose that the utility of each player is equal to the amount of shekels he receives.  \cref{fig:1120_concepts_main}(c)–(h) display the predictions of the other solution concepts.\footnote{See Supplementary Appendix IV.1 for definitions of the discussed  concepts.}  Pure $\epsilon$-equilibria exist for $\epsilon\geq 8$, and they consist of low actions for all $\epsilon<18$. All actions are rationalizable, but the unique symmetric Nash equilibrium, logit quantal-response equilibrium, and sampling equilibrium are consistent with satisficing equilibrium and the empirical data. 

\begin{figure}[h]
\caption{The 11--20 game: predictions.}
\centering
\adjustbox{width=1.0\linewidth}{
\begin{tabular}{cccc}
 \begin{tikzpicture}[scale=0.6]
  \draw[thin,black] (0.5,-0.5) rectangle (10.5,-10.5);

\tikzmath{\del = 0.05; \eps1 = 2; \eps2 = 1.0; \eps3 = 3.5;}

\draw[very thick,draw=none, fill=HTML1,rounded corners = 0.5mm,opacity=1]
(0.5 - \eps1*\del ,-4.5)
-- (3.5,-4.5)
-- (3.5,-3.5)
-- (4.5,-3.5)
-- (4.5,-0.5 + \eps1*\del)
-- (0.5 - \eps1*\del ,-0.5 + \eps1*\del)
-- cycle;  

\draw[very thick,draw=none, fill=HTML4,rounded corners = 0.8mm,opacity=1]
(0.5+ \eps2*\del,-3.5)
-- (2.5,-3.5)
-- (2.5,-2.5)
-- (3.5,-2.5) 
-- (3.5,-0.5- \eps2*\del)
-- (0.5 + \eps2*\del,-0.5- \eps2*\del)
-- cycle;

\draw[very thick,draw=none, fill=HTML0,rounded corners = 1mm,opacity=1]
(0.5+\eps3*\del,-2.5) 
-- (1.5,-2.5) 
-- (1.5,-1.5) 
-- (2.5,-1.5)
-- (2.5,-0.5 - \eps3*\del)
-- (0.5 + \eps3*\del,-0.5 - \eps3*\del)
-- cycle;

 \node[] at (5.5,-9) {\Large \textcolor{HTML0}{\textbf{2-SE}} $\subseteq$ \textcolor{HTML4}{\textbf{3-SE}} $\subseteq$ \textcolor{HTML1}{\textbf{4-SE}} };
  
  \foreach \i in {1,...,10} {
    \node[font=\small] at (0,-\i) {\the\numexpr21-\i\relax};
    \node[font=\small] at (\i,0) {\the\numexpr21-\i\relax};
  }
\end{tikzpicture} & \pgfplotstableread{
 0.0036  0.0072  0.018  0.0192  0.0036  0.0006  0.0036  0.0018  0.0  0.0024
 0.0072  0.0144  0.036  0.0384  0.0072  0.0012  0.0072  0.0036  0.0  0.0048
 0.018   0.036   0.09   0.096   0.018   0.003   0.018   0.009   0.0  0.012
 0.0192  0.0384  0.096  0.1024  0.0192  0.0032  0.0192  0.0096  0.0  0.0128
 0.0036  0.0072  0.018  0.0192  0.0036  0.0006  0.0036  0.0018  0.0  0.0024
 0.0006  0.0012  0.003  0.0032  0.0006  0.0001  0.0006  0.0003  0.0  0.0004
 0.0036  0.0072  0.018  0.0192  0.0036  0.0006  0.0036  0.0018  0.0  0.0024
 0.0018  0.0036  0.009  0.0096  0.0018  0.0003  0.0018  0.0009  0.0  0.0012
 0.0     0.0     0.0    0.0     0.0     0.0     0.0     0.0     0.0  0.0
 0.0024  0.0048  0.012  0.0128  0.0024  0.0004  0.0024  0.0012  0.0  0.0016
}\data

\begin{tikzpicture}[scale=0.6]
  \def\minval{0.001}
  \def\maxval{0.25}

  \foreach \r in {0,...,9} {
    \foreach \c in {0,...,9} {
      \pgfplotstablegetelem{\r}{[index]\c}\of{\data}%
      \pgfmathsetmacro{\val}{\pgfplotsretval}%

      \pgfmathsetmacro{\pct}{%
        ifthenelse(
          \val>\minval,
          (ln(\val + \minval)-ln(2*\minval)) / (ln(\maxval+\minval) -ln(2*\minval) ) * 100, 0
          )%
      }%

      \node[
        fill=HTML3!\pct!white,
        minimum size=6mm
      ] at ({\c + 1},{-\r -1}) {};
    }
  }

  \draw[thin,black] (0.5,-0.5) rectangle (10.5,-10.5);

  \foreach \i in {1,...,10} {
    \node[font=\small] at (0,-\i) {\the\numexpr21-\i\relax};
    \node[font=\small] at (\i,0) {\the\numexpr21-\i\relax};
  }
\end{tikzpicture} & 
 \begin{tikzpicture}[scale=0.6]
  \draw[thin,black] (0.5,-0.5) rectangle (10.5,-10.5);

 \draw[very thick,draw=none,fill=HTML4,rounded corners = 0.75mm,opacity=1, scale around={1.05:(9.5,-9.5)}, shift={(-0.05,0.05)}]
    (8.6,-9.5)
    -- (8.6,-10.55)
    -- (10.55,-10.55)
    -- (10.55,-8.6)  
    -- (9.5,-8.6)
    -- (9.5,-9.5)
    -- cycle;

  \draw[
  very thick,
  draw=none,
  fill=HTML0,
  rounded corners=0.8mm,
  opacity=1,
  scale around={0.92:(9.5,-9.5)},
  shift={(0.04,-0.04)}
]
(9.5,-8.5)
-- (10.5,-8.5)
-- (10.5,-9.5)
-- (9.5,-9.5)
-- (9.5,-10.5)
-- (8.5,-10.5)
-- (8.5,-9.5)
-- (9.5,-9.5)
-- cycle;

 \node[] at (4.3,-9) {\Large \textcolor{HTML0}{\textbf{$\epsilon=8$}} $\subseteq$ \textcolor{HTML4}{\textbf{$\epsilon =9$}}};

  \foreach \i in {1,...,10} {
    \node[font=\small] at (0,-\i) {\the\numexpr21-\i\relax};
    \node[font=\small] at (\i,0) {\the\numexpr21-\i\relax};
  }
\end{tikzpicture}
 &  \begin{tikzpicture}[scale=0.6]
  \draw[thin,black] (0.5,-0.5) rectangle (10.5,-10.5);

    \begin{scope}
    \fill[pattern=north east lines, pattern color=gray]
      (0.5,-0.5) rectangle (10.5,-10.5) -- cycle;
  \end{scope}

  \foreach \i in {1,...,10} {
    \node[font=\small] at (0,-\i) {\the\numexpr21-\i\relax};
    \node[font=\small] at (\i,0) {\the\numexpr21-\i\relax};
  }
\end{tikzpicture}  \\
 (a) Satisficing equilibria & (b) Data \citep{arad201211} & (c) Pure $\epsilon$-equilibria & (d) Rationalizable  \\
  \pgfplotstableread{
 0.0025  0.005  0.0075  0.01  0.0125  0.0125  0.0  0.0  0.0  0.0
 0.005   0.01   0.015   0.02  0.025   0.025   0.0  0.0  0.0  0.0
 0.0075  0.015  0.0225  0.03  0.0375  0.0375  0.0  0.0  0.0  0.0
 0.01    0.02   0.03    0.04  0.05    0.05    0.0  0.0  0.0  0.0
 0.0125  0.025  0.0375  0.05  0.0625  0.0625  0.0  0.0  0.0  0.0
 0.0125  0.025  0.0375  0.05  0.0625  0.0625  0.0  0.0  0.0  0.0
 0.0     0.0    0.0     0.0   0.0     0.0     0.0  0.0  0.0  0.0
 0.0     0.0    0.0     0.0   0.0     0.0     0.0  0.0  0.0  0.0
 0.0     0.0    0.0     0.0   0.0     0.0     0.0  0.0  0.0  0.0
 0.0     0.0    0.0     0.0   0.0     0.0     0.0  0.0  0.0  0.0
}\data

\begin{tikzpicture}[scale=0.6]
  \def\minval{0.001}
  \def\maxval{0.25} 

  \foreach \r in {0,...,9} {
    \foreach \c in {0,...,9} {
      \pgfplotstablegetelem{\r}{[index]\c}\of{\data}%
      \pgfmathsetmacro{\val}{\pgfplotsretval}%

      \pgfmathsetmacro{\pct}{%
        ifthenelse(
          \val>\minval,
          (ln(\val + \minval)-ln(2*\minval)) / (ln(\maxval+\minval) -ln(2*\minval) ) * 100, 0
          )%
      }%

      \node[
        fill=gray!\pct!white,
        minimum size=6mm
      ] at ({\c + 1},{-\r -1}) {};
    }
  }

  \draw[thin,black] (0.5,-0.5) rectangle (10.5,-10.5);

  \foreach \i in {1,...,10} {
    \node[font=\small] at (0,-\i) {\the\numexpr21-\i\relax};
    \node[font=\small] at (\i,0) {\the\numexpr21-\i\relax};
  }
\end{tikzpicture} &  \pgfplotstableread{
0.0153 0.0177 0.0167 0.0149 0.0131 0.0116 0.0102 0.0091 0.0081 0.0072
0.0177 0.0205 0.0193 0.0172 0.0152 0.0134 0.0118 0.0105 0.0093 0.0083
0.0167 0.0193 0.0182 0.0162 0.0143 0.0126 0.0112 0.0099 0.0088 0.0078
0.0149 0.0172 0.0162 0.0145 0.0127 0.0112 0.0099 0.0088 0.0078 0.0070
0.0131 0.0152 0.0143 0.0127 0.0112 0.0099 0.0087 0.0077 0.0069 0.0061
0.0116 0.0134 0.0126 0.0112 0.0099 0.0087 0.0077 0.0068 0.0061 0.0054
0.0102 0.0118 0.0112 0.0099 0.0087 0.0077 0.0068 0.0060 0.0054 0.0048
0.0091 0.0105 0.0099 0.0088 0.0077 0.0068 0.0060 0.0053 0.0048 0.0042
0.0081 0.0093 0.0088 0.0078 0.0069 0.0061 0.0054 0.0048 0.0042 0.0038
0.0072 0.0083 0.0078 0.0070 0.0061 0.0054 0.0048 0.0042 0.0038 0.0034
}\data

\begin{tikzpicture}[scale=0.6]
  \def\minval{0.001}
  \def\maxval{0.25} 

  \foreach \r in {0,...,9} {
    \foreach \c in {0,...,9} {
      \pgfplotstablegetelem{\r}{[index]\c}\of{\data}%
      \pgfmathsetmacro{\val}{\pgfplotsretval}%

      \pgfmathsetmacro{\pct}{%
        ifthenelse(
          \val>\minval,
          (ln(\val + \minval)-ln(2*\minval)) / (ln(\maxval+\minval) -ln(2*\minval) ) * 100, 0
          )%
      }%

      \node[
        fill=gray!\pct!white,
        minimum size=6mm
      ] at ({\c + 1},{-\r -1}) {};
    }
  }

  \draw[thin,black] (0.5,-0.5) rectangle (10.5,-10.5);

  \foreach \i in {1,...,10} {
    \node[font=\small] at (0,-\i) {\the\numexpr21-\i\relax};
    \node[font=\small] at (\i,0) {\the\numexpr21-\i\relax};
  }
\end{tikzpicture} & \pgfplotstableread{
0.0081 0.0146 0.0200 0.0233 0.0191 0.0047 0.0002 0.0001 0.0000 0.0000
0.0146 0.0262 0.0360 0.0419 0.0343 0.0084 0.0004 0.0001 0.0000 0.0000
0.0200 0.0360 0.0494 0.0575 0.0471 0.0115 0.0005 0.0001 0.0001 0.0000
0.0233 0.0419 0.0575 0.0669 0.0549 0.0134 0.0006 0.0001 0.0001 0.0000
0.0191 0.0343 0.0471 0.0549 0.0450 0.0110 0.0005 0.0001 0.0001 0.0000
0.0047 0.0084 0.0115 0.0134 0.0110 0.0027 0.0001 0.0000 0.0000 0.0000
0.0002 0.0004 0.0005 0.0006 0.0005 0.0001 0.0000 0.0000 0.0000 0.0000
0.0001 0.0001 0.0001 0.0001 0.0001 0.0000 0.0000 0.0000 0.0000 0.0000
0.0000 0.0000 0.0001 0.0001 0.0001 0.0000 0.0000 0.0000 0.0000 0.0000
0.0000 0.0000 0.0000 0.0000 0.0000 0.0000 0.0000 0.0000 0.0000 0.0000
}\data

\begin{tikzpicture}[scale=0.6]
  \def\minval{0.001}
  \def\maxval{0.25} 

  \foreach \r in {0,...,9} {
    \foreach \c in {0,...,9} {
      \pgfplotstablegetelem{\r}{[index]\c}\of{\data}%
      \pgfmathsetmacro{\val}{\pgfplotsretval}%

      \pgfmathsetmacro{\pct}{%
        ifthenelse(
          \val>\minval,
          (ln(\val + \minval)-ln(2*\minval)) / (ln(\maxval+\minval) -ln(2*\minval) ) * 100, 0
          )%
      }%

      \node[
        fill=gray!\pct!white,
        minimum size=6mm
      ] at ({\c + 1},{-\r -1}) {};
    }
  }

  \draw[thin,black] (0.5,-0.5) rectangle (10.5,-10.5);

  \foreach \i in {1,...,10} {
    \node[font=\small] at (0,-\i) {\the\numexpr21-\i\relax};
    \node[font=\small] at (\i,0) {\the\numexpr21-\i\relax};
  }
\end{tikzpicture} & \pgfplotstableread{
0.0985 0.0985 0.0676 0.0318 0.0118 0.0039 0.0013 0.0004 0.0001 0.0000
0.0985 0.0985 0.0676 0.0318 0.0118 0.0039 0.0013 0.0004 0.0001 0.0000
0.0676 0.0676 0.0464 0.0218 0.0081 0.0027 0.0009 0.0003 0.0001 0.0000
0.0318 0.0318 0.0218 0.0103 0.0038 0.0013 0.0004 0.0001 0.0000 0.0000
0.0118 0.0118 0.0081 0.0038 0.0014 0.0005 0.0001 0.0000 0.0000 0.0000
0.0039 0.0039 0.0027 0.0013 0.0005 0.0002 0.0000 0.0000 0.0000 0.0000
0.0013 0.0013 0.0009 0.0004 0.0001 0.0000 0.0000 0.0000 0.0000 0.0000
0.0004 0.0004 0.0003 0.0001 0.0000 0.0000 0.0000 0.0000 0.0000 0.0000
0.0001 0.0001 0.0001 0.0000 0.0000 0.0000 0.0000 0.0000 0.0000 0.0000
0.0000 0.0000 0.0000 0.0000 0.0000 0.0000 0.0000 0.0000 0.0000 0.0000
}\data

\begin{tikzpicture}[scale=0.6]
  \def\minval{0.001}
  \def\maxval{0.25}  

  \foreach \r in {0,...,9} {
    \foreach \c in {0,...,9} {
      \pgfplotstablegetelem{\r}{[index]\c}\of{\data}%
      \pgfmathsetmacro{\val}{\pgfplotsretval}%

      \pgfmathsetmacro{\pct}{%
        ifthenelse(
          \val>\minval,
          (ln(\val + \minval)-ln(2*\minval)) / (ln(\maxval+\minval) -ln(2*\minval) ) * 100, 0
          )%
      }%

      \node[
        fill=gray!\pct!white,
        minimum size=6mm
      ] at ({\c + 1},{-\r -1}) {};
    }
  }

  \draw[thin,black] (0.5,-0.5) rectangle (10.5,-10.5);

  \foreach \i in {1,...,10} {
    \node[font=\small] at (0,-\i) {\the\numexpr21-\i\relax};
    \node[font=\small] at (\i,0) {\the\numexpr21-\i\relax};
  }
\end{tikzpicture}
 
 \\
 (e) Symmetric Nash & (f) QRE ($\lambda = 0.1$) &(g) QRE ($\lambda = 0.73$) & (h) Sampling equilibrium  
\end{tabular}
}
\noindent\begin{minipage}{1\textwidth} 
\vspace{0.2cm}
\small{\textbf{Note.}
Each panel plots player 1's actions on the vertical axis and player 2's on the horizontal axis.  
Panel~(a) shows the $(k,k)$-satisficing equilibria, or $k$-SE, for $k\in\{2,3,4\}$ (note the nested structure).
 Panel~(b) shows the realized play in the experiments by \cite*{arad201211}.
 Panel~(c) shows the pure $\epsilon$-equilibria for $\epsilon< 18$. 
Panel~(d) shows the set of rationalizable action profiles.  
Panels~(e)–(h) display predicted probabilities under the unique symmetric Nash equilibrium, two logit quantal-response equilibria (QRE) with rationality parameters $\lambda\in\{0.1,0.73\}$, and the unique sampling equilibrium.
Darker shading indicates higher predicted probability or observed frequency; white denotes zero; purple vs. gray distinguishes frequencies based on experimental data from theoretical predictions.}
\end{minipage}
\label{fig:1120_concepts_main}
\end{figure}

To separate satisficing equilibrium from the latter solution concepts, we introduce the Double or Dozen game below, which has a similar structure to the 11--20 game.

\begin{quote}
    \textbf{The Double or Dozen game.} A mother asks her sons, Chacham and Rasha, to write down a number of dreidels between one and ten. If they choose different numbers, the son who picked the smaller number receives that amount plus twelve, and the other receives the number he himself chose. If the sons pick the same number, each receives double that number of dreidels.
\end{quote}

Suppose that each son strictly prefers to receive more dreidels. Then, the game admits no $(1,1)$-satisficing equilibrium---that is, there is no pure Nash equilibrium. In fact, the game has the same pure best-response correspondence as the 11--20 game. Here, there is a unique $(2, 2)$-satisficing equilibrium, which is at $(10,10)$ (see \cref{fig:DD_concepts}(a)).

To compare satisficing equilibrium with other solution concepts in the Double or Dozen game, we commit to the following cardinal representation: the utility of each son is equal to the number of dreidels he receives.  \cref{fig:DD_concepts}(b)–(h) display the predictions of the other solution concepts. All actions are again rationalizable. But here, in stark contrast to the 11--20 game, pure $\epsilon$-equilibrium agrees with satisficing equilibrium, while the other solution concepts disagree with $(k,k)$-satisficing equilibrium for low $k$, either predicting low actions or diffuse play. Maybe surprisingly, despite the structural similarity between the Double or Dozen and the 11--20 game, several solution concepts provide starkly varying predictions between the two games.

\begin{figure}[ht]
\caption{The Double or Dozen game: predictions.}
\centering
\adjustbox{width=1.0\linewidth}{
\begin{tabular}{cccc}
\begin{tikzpicture}[scale=0.6]
 \draw[thin,black] (0.5,-0.5) rectangle (10.5,-10.5);
  \tikzmath{\del = 0.05; \eps1 = 3; \eps2 = 1.2; \eps3 = 3;}
\draw[very thick, draw=none, fill=HTML1,rounded corners=0.5mm,opacity=1]
    (0.5 - \eps3*\del,-1.5 - \eps3*\del)
    -- (1.5 -\eps3*\del,-1.5 - \eps3*\del)
    -- (1.5 -\eps3*\del,-2.5 - \eps3*\del)
    -- (2.5 -\eps3*\del,-2.5 - \eps3*\del)
    -- (2.5 -\eps3*\del,-3.5 -\eps3*\del)
    -- (3.5 + \eps3*\del,-3.5 -\eps3*\del)
    -- (3.5 +\eps3*\del,-2.5 + \eps3*\del)
    -- (2.5 + \eps3*\del,-2.5 + \eps3*\del)
    -- (2.5+\eps3*\del,-1.5 + \eps3*\del)
    -- (1.5 +\eps3*\del,-1.5 + \eps3*\del)
    -- (1.5 + \eps3*\del,-0.5 +\eps3*\del)
    -- (0.5 - \eps3*\del,-0.5 + \eps3*\del)
    -- cycle;

\draw[very thick, draw=none, fill=HTML4, rounded corners = 0.8mm,opacity=1]
(0.5,-1.5)
-- (1.5,-1.5)
-- (1.5,-2.5)
-- (2.5,-2.5)
-- (2.5,-1.5)
-- (1.5,-1.5)
-- (1.5,-0.5)
-- (0.5,-0.5)
-- cycle;

\draw[very thick, draw=none, fill=HTML0, rounded corners = 1mm,opacity=1]
  (0.5 + \eps1*\del,-1.5+\eps1*\del) rectangle (1.5-\eps1*\del,-0.5-\eps1*\del);

 \node[] at (5.5,-9) {\Large \textcolor{HTML0}{\textbf{2-SE}} $\subseteq$ \textcolor{HTML4}{\textbf{3-SE}} $\subseteq$ \textcolor{HTML1}{\textbf{4-SE}} };

  \foreach \i in {1,...,10} {
    \node[font=\small] at (0,-\i) {\the\numexpr11-\i\relax};
    \node[font=\small] at (\i,0) {\the\numexpr11-\i\relax};
  }
\end{tikzpicture}& \begin{tikzpicture}[scale=0.6]
 \draw[thin,black] (0.5,-0.5) rectangle (10.5,-10.5);
  \tikzmath{\del = 0.05; \eps1 = 3; \eps2 = 1.2; \eps3 = 3;}
\draw[very thick, draw=none, fill=HTML1,rounded corners=0.5mm,opacity=1]
    (0.5 - \eps3*\del,-1.5 - \eps3*\del)
    -- (1.5 -\eps3*\del,-1.5 - \eps3*\del)
    -- (1.5 -\eps3*\del,-2.5 - \eps3*\del)
    -- (2.5 -\eps3*\del,-2.5 - \eps3*\del)
    -- (2.5 -\eps3*\del,-3.5 -\eps3*\del)
    -- (3.5 + \eps3*\del,-3.5 -\eps3*\del)
    -- (3.5 +\eps3*\del,-2.5 + \eps3*\del)
    -- (2.5 + \eps3*\del,-2.5 + \eps3*\del)
    -- (2.5+\eps3*\del,-1.5 + \eps3*\del)
    -- (1.5 +\eps3*\del,-1.5 + \eps3*\del)
    -- (1.5 + \eps3*\del,-0.5 +\eps3*\del)
    -- (0.5 - \eps3*\del,-0.5 + \eps3*\del)
    -- cycle;

\draw[very thick, draw=none, fill=HTML4, rounded corners = 0.8mm,opacity=1]
(0.5,-1.5)
-- (1.5,-1.5)
-- (1.5,-2.5)
-- (2.5,-2.5)
-- (2.5,-1.5)
-- (1.5,-1.5)
-- (1.5,-0.5)
-- (0.5,-0.5)
-- cycle;

\draw[very thick, draw=none, fill=HTML0, rounded corners = 1mm,opacity=1]
  (0.5 + \eps1*\del,-1.5+\eps1*\del) rectangle (1.5-\eps1*\del,-0.5-\eps1*\del);

 \node[] at (5.5,-9) {\Large \textcolor{HTML0}{\textbf{$\epsilon=1$}} $\subseteq$ \textcolor{HTML4}{$\epsilon=2$} $\subseteq$ \textcolor{HTML1}{\textbf{$\epsilon=3$}} };

  \foreach \i in {1,...,10} {
    \node[font=\small] at (0,-\i) {\the\numexpr11-\i\relax};
    \node[font=\small] at (\i,0) {\the\numexpr11-\i\relax};
  }
\end{tikzpicture} &
\begin{tikzpicture}[scale=0.6]
  \draw[thin,black] (0.5,-0.5) rectangle (10.5,-10.5);

    \begin{scope}
    \fill[pattern=north east lines, pattern color=gray]
      (0.5,-0.5) rectangle (10.5,-10.5) -- cycle;
  \end{scope}

  \foreach \i in {1,...,10} {
    \node[font=\small] at (0,-\i) {\the\numexpr11-\i\relax};
    \node[font=\small] at (\i,0) {\the\numexpr11-\i\relax};
  }
\end{tikzpicture} & \pgfplotstableread{
0.0064  0.0075  0.0072  0.0073  0.0072  0.0073  0.0072  0.0075  0.0064  0.0160
0.0075  0.0087  0.0084  0.0085  0.0084  0.0085  0.0084  0.0087  0.0075  0.0187
0.0072  0.0084  0.0081  0.0082  0.0081  0.0082  0.0081  0.0084  0.0072  0.0180
0.0073  0.0085  0.0082  0.0084  0.0083  0.0084  0.0082  0.0085  0.0073  0.0183
0.0072  0.0084  0.0081  0.0083  0.0082  0.0083  0.0081  0.0084  0.0072  0.0181
0.0073  0.0085  0.0082  0.0084  0.0083  0.0084  0.0082  0.0085  0.0073  0.0183
0.0072  0.0084  0.0081  0.0082  0.0081  0.0082  0.0081  0.0084  0.0072  0.0180
0.0075  0.0087  0.0084  0.0085  0.0084  0.0085  0.0084  0.0087  0.0075  0.0187
0.0064  0.0075  0.0072  0.0073  0.0072  0.0073  0.0072  0.0075  0.0064  0.0160
0.0160  0.0187  0.0180  0.0183  0.0181  0.0183  0.0180  0.0187  0.0160  0.0400
}\data

\begin{tikzpicture}[scale=0.6]
  \def\minval{0.001}
  \def\maxval{0.25} 

  \foreach \r in {0,...,9} {
    \foreach \c in {0,...,9} {
      \pgfplotstablegetelem{\r}{[index]\c}\of{\data}%
      \pgfmathsetmacro{\val}{\pgfplotsretval}%
      \pgfmathsetmacro{\pct}{%
        ifthenelse(
          \val>\minval,
          (ln(\val + \minval)-ln(2*\minval)) / (ln(\maxval+\minval) -ln(2*\minval) ) * 100, 0
          )%
      }%

      \node[
        fill=gray!\pct!white,
        minimum size=6mm
      ] at ({\c + 1},{-\r -1}) {};
    }
  }

  \draw[thin,black] (0.5,-0.5) rectangle (10.5,-10.5);

  \foreach \i in {1,...,10} {
    \node[font=\small] at (0,-\i) {\the\numexpr11-\i\relax};
    \node[font=\small] at (\i,0) {\the\numexpr11-\i\relax};
  }
\end{tikzpicture}\\
(a)  Satisficing equilibria & Pure $\epsilon$-equilibria & (c) Rationalizable & (d) First symmetric Nash   \\
\pgfplotstableread{
0.2500  0.0000  0.0625  0.0357  0.0536  0.0357  0.0625  0.0000  0.0000  0.0000
0.0000  0.0000  0.0000  0.0000  0.0000  0.0000  0.0000  0.0000  0.0000  0.0000
0.0625  0.0000  0.0156  0.0089  0.0134  0.0089  0.0156  0.0000  0.0000  0.0000
0.0357  0.0000  0.0089  0.0051  0.0076  0.0051  0.0089  0.0000  0.0000  0.0000
0.0536  0.0000  0.0134  0.0076  0.0115  0.0076  0.0134  0.0000  0.0000  0.0000
0.0357  0.0000  0.0089  0.0051  0.0076  0.0051  0.0089  0.0000  0.0000  0.0000
0.0625  0.0000  0.0156  0.0089  0.0134  0.0089  0.0156  0.0000  0.0000  0.0000
0.0000  0.0000  0.0000  0.0000  0.0000  0.0000  0.0000  0.0000  0.0000  0.0000
0.0000  0.0000  0.0000  0.0000  0.0000  0.0000  0.0000  0.0000  0.0000  0.0000
0.0000  0.0000  0.0000  0.0000  0.0000  0.0000  0.0000  0.0000  0.0000  0.0000
}\data

\begin{tikzpicture}[scale=0.6]
  \def\minval{0.001}
  \def\maxval{0.25}  

  \foreach \r in {0,...,9} {
    \foreach \c in {0,...,9} {
      \pgfplotstablegetelem{\r}{[index]\c}\of{\data}%
      \pgfmathsetmacro{\val}{\pgfplotsretval}%

      \pgfmathsetmacro{\pct}{%
        ifthenelse(
          \val>\minval,
          (ln(\val + \minval)-ln(2*\minval)) / (ln(\maxval+\minval) -ln(2*\minval) ) * 100, 0
          )%
      }%

      \node[
        fill=gray!\pct!white,
        minimum size=6mm
      ] at ({\c + 1},{-\r -1}) {};
    }
  }

  \draw[thin,black] (0.5,-0.5) rectangle (10.5,-10.5);

  \foreach \i in {1,...,10} {
    \node[font=\small] at (0,-\i) {\the\numexpr11-\i\relax};
    \node[font=\small] at (\i,0) {\the\numexpr11-\i\relax};
  }
\end{tikzpicture}
&  
\pgfplotstableread{
 0.00826488  0.0082658   0.00826936  0.00828086  0.0083135   0.00839772  0.00860141  0.009079    0.0102181  0.0132208
 0.0082658   0.00826672  0.00827028  0.00828178  0.00831443  0.00839866  0.00860236  0.00908001  0.0102192  0.0132223
 0.00826936  0.00827028  0.00827385  0.00828536  0.00831801  0.00840228  0.00860607  0.00908393  0.0102236  0.013228
 0.00828086  0.00828178  0.00828536  0.00829688  0.00832958  0.00841396  0.00861804  0.00909656  0.0102378  0.0132464
 0.0083135   0.00831443  0.00831801  0.00832958  0.00836241  0.00844712  0.00865201  0.00913241  0.0102782  0.0132986
 0.00839772  0.00839866  0.00840228  0.00841396  0.00844712  0.0085327   0.00873966  0.00922493  0.0103823  0.0134333
 0.00860141  0.00860236  0.00860607  0.00861804  0.00865201  0.00873966  0.00895164  0.00944868  0.0106341  0.0137591
 0.009079    0.00908001  0.00908393  0.00909656  0.00913241  0.00922493  0.00944868  0.00997332  0.0112246  0.0145231
 0.0102181   0.0102192   0.0102236   0.0102378   0.0102782   0.0103823   0.0106341   0.0112246   0.0126329  0.0163452
 0.0132208   0.0132223   0.013228    0.0132464   0.0132986   0.0134333   0.0137591   0.0145231   0.0163452  0.0211484
}\data

\begin{tikzpicture}[scale=0.6]
  \def\minval{0.001}
  \def\maxval{0.25} 

  \foreach \r in {0,...,9} {
    \foreach \c in {0,...,9} {
      \pgfplotstablegetelem{\r}{[index]\c}\of{\data}%
      \pgfmathsetmacro{\val}{\pgfplotsretval}%

      \pgfmathsetmacro{\pct}{%
        ifthenelse(
          \val>\minval,
          (ln(\val + \minval)-ln(2*\minval)) / (ln(\maxval+\minval) -ln(2*\minval) ) * 100, 0
          )%
      }%

      \node[
        fill=gray!\pct!white,
        minimum size=6mm
      ] at ({\c + 1},{-\r -1}) {};
    }
  }

  \draw[thin,black] (0.5,-0.5) rectangle (10.5,-10.5);

  \foreach \i in {1,...,10} {
    \node[font=\small] at (0,-\i) {\the\numexpr11-\i\relax};
    \node[font=\small] at (\i,0) {\the\numexpr11-\i\relax};
  }
\end{tikzpicture} & \pgfplotstableread{
 0.00820155  0.00824313  0.00822779  0.00823679  0.00822883  0.00823935  0.00821803  0.00828913  0.00780981  0.016868
 0.00824313  0.00828493  0.00826951  0.00827855  0.00827055  0.00828113  0.0082597   0.00833116  0.00784941  0.0169535
 0.00822779  0.00826951  0.00825412  0.00826314  0.00825516  0.00826571  0.00824433  0.00831566  0.0078348   0.0169219
 0.00823679  0.00827855  0.00826314  0.00827218  0.00826418  0.00827475  0.00825334  0.00832475  0.00784337  0.0169404
 0.00822883  0.00827055  0.00825516  0.00826418  0.00825619  0.00826675  0.00824536  0.0083167   0.00783579  0.0169241
 0.00823935  0.00828113  0.00826571  0.00827475  0.00826675  0.00827732  0.00825591  0.00832734  0.00784581  0.0169457
 0.00821803  0.0082597   0.00824433  0.00825334  0.00824536  0.00825591  0.00823455  0.00830579  0.00782551  0.0169019
 0.00828913  0.00833116  0.00831566  0.00832475  0.0083167   0.00832734  0.00830579  0.00837765  0.00789321  0.0170481
 0.00780981  0.00784941  0.0078348   0.00784337  0.00783579  0.00784581  0.00782551  0.00789321  0.00743679  0.0160623
 0.016868    0.0169535   0.0169219   0.0169404   0.0169241   0.0169457   0.0169019   0.0170481   0.0160623   0.034692
}\data

\begin{tikzpicture}[scale=0.6]
  \def\minval{0.001}
  \def\maxval{0.25}  

  \foreach \r in {0,...,9} {
    \foreach \c in {0,...,9} {
      \pgfplotstablegetelem{\r}{[index]\c}\of{\data}%
      \pgfmathsetmacro{\val}{\pgfplotsretval}%

      \pgfmathsetmacro{\pct}{%
        ifthenelse(
          \val>\minval,
          (ln(\val + \minval)-ln(2*\minval)) / (ln(\maxval+\minval) -ln(2*\minval) ) * 100, 0
          )%
      }%

      \node[
        fill=gray!\pct!white,
        minimum size=6mm
      ] at ({\c + 1},{-\r -1}) {};
    }
  }

  \draw[thin,black] (0.5,-0.5) rectangle (10.5,-10.5);

  \foreach \i in {1,...,10} {
    \node[font=\small] at (0,-\i) {\the\numexpr11-\i\relax};
    \node[font=\small] at (\i,0) {\the\numexpr11-\i\relax};
  }
\end{tikzpicture}
 & \pgfplotstableread{
0.0000  0.0001  0.0001  0.0002  0.0004  0.0006  0.0008  0.0006  0.0002  0.0000
0.0001  0.0004  0.0008  0.0013  0.0024  0.0039  0.0050  0.0043  0.0016  0.0002
0.0001  0.0008  0.0018  0.0028  0.0050  0.0082  0.0107  0.0091  0.0034  0.0003
0.0002  0.0013  0.0028  0.0045  0.0080  0.0130  0.0169  0.0143  0.0055  0.0005
0.0004  0.0024  0.0050  0.0080  0.0142  0.0231  0.0300  0.0254  0.0097  0.0009
0.0006  0.0039  0.0082  0.0130  0.0231  0.0377  0.0489  0.0414  0.0158  0.0015
0.0008  0.0050  0.0107  0.0169  0.0300  0.0489  0.0634  0.0538  0.0205  0.0019
0.0006  0.0043  0.0091  0.0143  0.0254  0.0414  0.0538  0.0456  0.0174  0.0016
0.0002  0.0016  0.0034  0.0055  0.0097  0.0158  0.0205  0.0174  0.0066  0.0006
0.0000  0.0002  0.0003  0.0005  0.0009  0.0015  0.0019  0.0016  0.0006  0.0001
}\data

\begin{tikzpicture}[scale=0.6]
  \def\minval{0.001}
  \def\maxval{0.25}

  \foreach \r in {0,...,9} {
    \foreach \c in {0,...,9} {
      \pgfplotstablegetelem{\r}{[index]\c}\of{\data}%
      \pgfmathsetmacro{\val}{\pgfplotsretval}%

      \pgfmathsetmacro{\pct}{%
        ifthenelse(
          \val>\minval,
          (ln(\val + \minval)-ln(2*\minval)) / (ln(\maxval+\minval) -ln(2*\minval) ) * 100, 0
          )%
      }%

      \node[
        fill=gray!\pct!white,
        minimum size=6mm
      ] at ({\c + 1},{-\r -1}) {};
    }
  }

  \draw[thin,black] (0.5,-0.5) rectangle (10.5,-10.5);

  \foreach \i in {1,...,10} {
    \node[font=\small] at (0,-\i) {\the\numexpr11-\i\relax};
    \node[font=\small] at (\i,0) {\the\numexpr11-\i\relax};
  }
\end{tikzpicture}
\\
(e) Second symmetric Nash
&(f) QRE ($\lambda = 1$) & (g) QRE ($\lambda = 16$) & (h) Sampling equilibrium 
\end{tabular}
}
\noindent\begin{minipage}{1\textwidth} 
\vspace{0.2cm}
\small{\textbf{Note.} 
Each panel plots Chacham's actions on the vertical axis and Rasha's on the horizontal axis.  
Panel~(a) shows the $(k,k)$-satisficing equilibria, or $k$-SE, for $k\in\{2,3,4\}$ (note the nested structure). 
 Panel~(b) shows the pure $\epsilon$-equilibria for $\epsilon\leq 3$. 
Panel~(c) shows the set of rationalizable action profiles.  
Panels~(d)–(h) display predicted probabilities under the two symmetric Nash equilibria, two logit quantal-response equilibria (QRE) with rationality parameters $\lambda\in\{1,16\}$, and the unique sampling equilibrium.  
Darker shading indicates higher predicted probability; white denotes zero.
}
\end{minipage}
\label{fig:DD_concepts}
\end{figure}

In  \cref{app:gamesmain}, we introduce an additional game that separates satisficing equilibrium from level-$\mathpzc{k}$ thinking (\cref{app:110game}) and discuss the Traveler’s Dilemma \citep[][\cref{sec:TD}]{basu1994traveler}, where satisficing equilibrium with values of~$k \le 4$ again  aligns well with existing empirical evidence. However, given our focus on theory, we  do not provide a full empirical study of satisficing equilibrium (see related discussion in \cref{sec:characterization}).

%
%

\section{Setup}\label{sec:model}

A \emph{game} is a tuple
\[
\g := \left( \;N, \, \{A_i\}_{i \in N}, \, \{\succsim_i\}_{i \in N} \; \right)
\]
consisting of a set of agents $N:=\{1,\ldots,n\}$, a set of actions $A_i:=\{1,\ldots,m_i\}$ for each agent $i\in\agentset$ (with $\anumvec=(m_1,\ldots,m_n)$), and a preference relation $\succsim_i$ for each agent $i\in\agentset$ over the set of action profiles $\profileset:=\times_{i \in N} A_i$. 

Write $\succ_i$ for the strict part of $\succsim_i$ and, in a standard abuse of notation, let $\actionvec=(a_i,\actionvec_{-i}) \in \profileset$ denote an action profile in which $a_i$ is the action of $i$ and $\actionvec_{-i} \in \profileset_{-i } := \times_{j \in N \setminus \{i\}} A_j$ is the action profile of agents in $N \setminus \{i\}$. For a game $\g$, agent $i\in\agentset$, and positive integer $k_i$, we say that the action $\action_i $ is a \emph{top}-$k_i$ \emph{response} to $\actionvec_{-i} $, and that  agent $i$ is \emph{$k_i$-satisficed} at $\actionvec$, if 
\[
\rank_i(\actionvec \,|\preferencerelation):= 1 + |\{x \in A_i : (x,\actionvec_{-i}) \succ_i \actionvec \}| \leq k_i.
\]
We next define the central concept of our article.
\begin{definition*}[Satisficing equilibrium] 
In any  game $\g$ and for any vector of positive integers $ \mathbf k := \V{k_1}{k_n}$, an action profile $\actionvec\in\profileset$ is a \emph{$\mathbf k$-satisficing equilibrium} if each agent $i\in\agentset$ is $k_i$-satisficed at $\actionvec$. 
\end{definition*}

While the definition allows for indifferences, we will assume throughout that each agent's preferences over their unilateral deviations are strict; that is, for any $i$, and any distinct action profiles $(x,\actionvec_{-i})$ and $(x',\actionvec_{-i})$, either $(x,\actionvec_{-i}) \succ_i (x',\actionvec_{-i})$ or $ (x,\actionvec_{-i})\prec_i (x',\actionvec_{-i})$. This assumption is not required for most of our results, but it simplifies the presentation and comparison with some related results from the literature.

Lastly, for any positive integer $n$ and any vector of positive integers $\anumvec=(m_1,\ldots,m_n)$, let  $\G{n}{\anumvec}$ denote the set of all games with agent set $N=\{1,\ldots,n\}$ and action set $\{1,\ldots,m_i\}$ for each agent $i \in N$. Any distinct games $\g$ and $g_{n,\anumvec,\succsim'}$ in $\G{n}{\anumvec}$ therefore have the same set of agents and action profiles and differ only in the collections of agents' preferences, $\succsim$ versus $\succsim'$. When subscripts are implied, we denote a representative element of $\G{n}{\anumvec}$ simply by $g$.
%
%

\section{Structural properties}\label{sec:structural}

In this section, we study the structure of $\mathbf{k}$-satisficing equilibria. We start with a simple observation on their nestedness.

\begin{claim}\label{observation:nested}
 Fix a game $\g$. For vectors $\mathbf{k},\mathbf{k}'\in\{1,2,\ldots\}^n$ with $\mathbf{k}\le \mathbf{k}'$, every $\mathbf{k}$-satisficing equilibrium of $\g$ is also a $\mathbf{k}'$-satisficing equilibrium of $\g$.
\end{claim}

In particular, \cref{observation:nested} shows that every pure Nash equilibrium is also a $\mathbf k$-satisficing equilibrium. Increasing $\mathbf k$ provides a structured departure from optimization that gradually accommodates more choice patterns.

\subsection{Precision}\label{sec:precision}

How sharp are the predictions of satisficing equilibrium? We formalize \emph{precision} in accordance with \cite*{selten1991properties}.

\begin{definition}
    For any game $g \in \G{n}{\anumvec}$ and $\mathbf k\in\{1,2,\ldots\}^n$, the \emph{precision} of $\mathbf k$-satisficing equilibrium in $g$ is given by
\begin{equation*}
    1-\frac{|\, \{\actionvec \in \profileset: \actionvec \text{ is a }\mathbf{k}\text{-satisficing equilibrium in }g\}\,|}{|\profileset|}.
\end{equation*}
The \emph{minimal precision} of satisficing equilibrium in $g$ is given by the above expression with $\mathbf k=\V{m_1-1}{m_n-1}$.
\end{definition}

Precision measures how sharply $\mathbf k$-satisficing equilibrium restricts possible outcomes. If there are few $\mathbf k$-satisficing equilibria compared to the total number of action profiles, then the precision is close to $1$. If most action profiles are $\mathbf k$-satisficing equilibria, then the precision is close to $0$. Given the nested structure of satisficing equilibrium (\cref{observation:nested}), the precision for any $\mathbf k'\leq \V{m_1-1}{m_n-1}$ is at least as high as the minimal  precision. The latter is the precision of satisficing equilibria in which no agent plays their worst action in response to the actions of the other agents.

We first provide a lower bound on the minimal precision of satisficing equilibrium. All proofs are provided in \cref{app:proofs}.

\begin{proposition}\label{prop:preditiveprecision}
For any set of games $\G{n}{\anumvec}$, the minimal precision of satisficing equilibrium in any $g\in\G{n}{\anumvec}$ is bounded below by $\tfrac{1}{\min_i m_i}$. Moreover, there exists some $g\in\G{n}{\anumvec}$ that achieves this bound.
\end{proposition}

Since the lower bound can be small, \cref{prop:preditiveprecision} shows that satisficing equilibrium may be fairly permissive in some games, thus not offering sharp selection. However, the precision of satisficing equilibrium is typically high:

\begin{theorem}\label{thm:predictiveprecision} 
Fix $\epsilon>0$. The minimal precision of satisficing equilibrium is  at least $1-\epsilon$ in all games, except for a fraction that is exponentially small in the number of agents. Formally, there exist positive constants $c$ and $\bar{n}$ such that for any set of games $ \G{n}{\anumvec}$ with $n \geq \bar{n}$ and $\anumvec \in \{2,3,\ldots\}^n$, 
\begin{equation*}\label{eq:natural}
\bigg| \left\{ g \in \G{n}{\anumvec} : \;
  \vcenter{\hbox{\text{\shortstack[l]{\emph{the minimal precision of satisficing}\\ 
     \emph{equilibrium in $g$ is at least $1-\epsilon$}}}}}
    \; \right\} \bigg|
\;\geq\; (1 - e^{-cn}) \, |\, \G{n}{\anumvec} \,|.
\end{equation*}
\end{theorem}

\cref{thm:predictiveprecision} shows that, in most games, satisficing equilibrium provides a sharp selection among all possible action profiles. In fact, the $\mathbf k$-satisficing equilibria in which no agent plays her worst response (that is, $k_i<m_i$ for all $i\in\agentset$) constitute at most an $\epsilon$-fraction of all action profiles in most games.  This implies that the precision of $\mathbf k$-satisficing equilibrium is typically high, and, in fact, asymptotically as high as that of pure Nash equilibrium.\footnote{\cite{mclennan2005expected} shows that, in an appropriately defined way, the expected number of mixed Nash equilibria is exponential in the number of agents and, in particular, is larger than the number of pure action profiles. In this sense, $\mathbf k$-satisficing equilibrium is sharper than mixed Nash equilibrium.} Observe, however, that increasing $\mathbf k$ does lead to a loss in precision and one may be interested in quantifying this loss for different values of $\mathbf k$. For example, for small $k=k_i$ in the Double or Dozen game (see \cref{sec:otherconcepts}), the loss in precision is $1/100$ when moving from $k$ to $k+1$.

\subsection{Existence}\label{sec:existence}

We next turn to study the existence of satisficing equilibrium. As noted, pure Nash equilibrium is not guaranteed to exist (e.g., in rock-paper-scissors or Bertrand competition with a discrete action space). Satisficing equilibrium is similarly not guaranteed to exist.

\begin{proposition}\label{prop:nonexistence}
For any set of games $\G{n}{\anumvec}$, there exists a game $g\in \G{n}{\anumvec}$ that does not have a $\mathbf k$-satisficing equilibrium for any $\mathbf k \leq (m_1-\lceil \tfrac{m_1}{n}\rceil,\dots,m_n-\lceil \tfrac{m_n}{n}\rceil)$. 
\end{proposition}

For intuition, consider a set of games $\G{n}{\anumvec}$ in which $m_i \leq n$ for all $i$. \cref{prop:nonexistence} implies that there is a game $g$ in this set for which, at every action profile, some agent is playing her worst action in response to the others' actions. In other words, the game $g$ does not have a $\mathbf k$-satisficing equilibrium for any $\mathbf k$ such that $k_i<m_i$ for all $i\in \agentset$. 

Our proof of \cref{prop:nonexistence} recasts the problem of finding a game with the desired property as a matching problem and then invokes \citeauthor{Hall35}'s Theorem (\citeyear{Hall35}). Additionally, in \cref{sec:nonexistencetight}, we show that the statement of \cref{prop:nonexistence} is tight in the following sense: every game $g\in \mathbf G_{n,\anumvec}$ has a $\mathbf k$-satisficing equilibrium for any $\mathbf k\geq (m_1-\lceil \tfrac{m_1}{n}\rceil+1,\dots,m_n-\lceil \tfrac{m_n}{n}\rceil+1)$.

While $\mathbf k$-satisficing equilibria are not guaranteed to exist in every game for every $\mathbf k$, we next show that such equilibria generally have good existence properties, even for small values of $\mathbf k$.

\begin{theorem}\label{thm:one_is_enough}
All but a fraction of games that is exponentially small in the number of agents admit a satisficing equilibrium where one agent is $2$-satisficed and all other agents are $1$-satisficed. Formally, there exist positive constants $c$ and $\bar{n}$ such that for any set of games $\G{n}{\anumvec}$ with $n \geq \bar{n}$ and any $\anumvec \in \{2,3,\ldots\}^n$,
    \[
\bigg| \left\{ g \in \G{n}{\anumvec} : \;
  \vcenter{\hbox{\text{\shortstack[l]{\emph{$g$ admits a $\mathbf k$-satisficing equilibrium with}\\
\emph{{$k_i=2$ for some $i$ and $k_j=1$ for all $j \neq i$}}}}}}
  \; \right\} \bigg|
\;\geq\; (1 - e^{-cn}) \, |\, \G{n}{\anumvec} \,|.
\]
\end{theorem}

By the nestedness of satisficing equilibrium (\cref{observation:nested}) we obtain the following corollary of \cref{thm:one_is_enough}.\footnote{In fact, \cref{cor:ksatisficing} extends trivially to $\mathbf k$ for any $\mathbf k \geq \V{2}{2}$.} 

\begin{corollary}\label{cor:ksatisficing}
All but a fraction of games that is exponentially small in the number of agents admit a $\V{2}{2}$-satisficing equilibrium.
\end{corollary}

\cref{thm:one_is_enough,cor:ksatisficing} contrast with the fact that pure Nash equilibria exist only in approximately $1-1/e\approx 63\%$ of games \citep*{arratia1989two,Rin00}. Our results therefore show that a very small deviation from optimizing behavior (as captured by a single agent being $2$-satisficed) suffices to ensure the general existence of stable outcomes.\footnote{Note that in \cref{thm:one_is_enough} the agent who is $2$-satisficed is not fixed. One might instead be interested in the fraction of games that admit a $(k_1,1,\ldots,1)$-satisficing equilibrium where agent 1 specifically is fixed to have $k_1>1$. \cref{cor:cor1} in \cref{appendixB} implies that the fraction of games that do not possess such an equilibrium is non-negligible and does not vanish with the number of agents.} This may explain the emergence of stable outcomes even in games that do not possess a pure Nash equilibrium (for  more on this point, see \cref{sec:potgame} below).

\cref{thm:one_is_enough} follows from a more general theorem, which approximates the distribution of $\mathbf k$-satisficing equilibria and allows for more granular results (see \cref{lem:general} in \cref{appendixB}). Its proof employs the probabilistic method and the Chen-Stein method \citep*[see][]{arratia1989two}.

\subsection{Approximate ordinal potential games}\label{sec:potgame}

The previous results establish that satisfying equilibria typically exist, but they may fail to exist in some cases. We next introduce a sufficient condition that guarantees existence within a broad and economically meaningful class of environments. Our approach parallels the classical theory of potential games: when agents’ incentives are aligned with a common potential function, a pure Nash equilibrium is guaranteed to exist \citep{monderer1996potential}. Here, we generalize this logic by allowing the alignment to be approximate rather than exact and define the class of \emph{$\mathbf{k}$-approximate ordinal potential games}. Intuitively, an agent's ranking of an outcome may only be worse than the one given by a common order by a limited amount---captured by the vector $\mathbf{k}$. 

Formally, in a game $\g$, for any $i \in \agentset$, any action profile $\actionvec \in \profileset$, and any weak order $\potpref$ on $\profileset$, let $\rank_i(\actionvec \,|\,\potpref):= 1 + |\{x \in \actionset : (x,\actionvec_{-i}) \potprefstrict \actionvec \}|$.

\begin{definition*}[Approximate ordinal potential games]\label{def:approximatepotential}
Fix any game $\g$ and $\mathbf k\in\{1,2,\ldots\}^n$. We say that $\g$ is a \emph{$\mathbf k$-approximate ordinal potential game} if there exists a weak order $\potpref$ on $\profileset$ such that for all $i\in \agentset$ and all $\actionvec \in \profileset$,
      \[
     \, \rank_i(\actionvec \,|\succsim_i)<\rank_i(\actionvec \,|\,\potpref) \,  + k_i .
     \]
\end{definition*}
This definition nests ordinal potential games because a game is a $\V{1}{1}$-approximate ordinal potential game if and only if it is an ordinal potential game \citep{monderer1996potential}.\footnote{If $\mathbf{k}=\V{1}{1}$, the defining inequality implies $ \rank_i(\actionvec \mid \succsim_i)\le \rank_i(\actionvec \mid \potpref) $ for all $i\in\agentset$ and all $\actionvec\in\profileset$. Suppose, for a contradiction, that the inequality is strict for some $i$ at $\actionvec=(a_i,\actionvec_{-i})$. Then there are strictly more deviations ranked better than $\actionvec$ under $\potpref$ than under $\succsim_i$, so for some other deviation $\actionvec'=(a_i',\actionvec_{-i})$ we must have $ \rank_i(\actionvec' \mid \succsim_i)>\rank_i(\actionvec' \mid \potpref)$, contradicting the defining inequality. Hence the two rank functions coincide on every unilateral deviation set, so $\succsim_i$ and $\potpref$ induce the same weak order there. Therefore $\g$ is a $\V{1}{1}$-approximate potential game if and only if it an is ordinal potential game.} In a $\V{2}{2}$-approximate ordinal potential game, there exists a weak order such that for any agent $i\in\agentset$ and any action profile $(a_i,\actionvec_{-i})$, $a_i$ is a top-$(1+\{\rank_i(\actionvec \,|\,\potpref)\})$ response to $\actionvec_{-i}$.

\begin{theorem}\label{thm:approximatepotential}
Fix any game $\g$ and any $\mathbf k\in\{1,2,\ldots\}^n$. If $\g$ is a $\mathbf k$-approximate ordinal potential game, then it admits a $\mathbf k$-satisficing equilibrium.
\end{theorem}

\cref{thm:approximatepotential} provides a simple sufficient condition for equilibrium existence. For instance, if a game is a $\V{2}{2}$-approximate ordinal potential game, then it is guaranteed to admit a $\V{2}{2}$-satisficing equilibrium.  Observe, however, that many games that are not $\V{2}{2}$-approximate ordinal potential games may still admit a $\V{2}{2}$-satisficing equilibrium; this is the case, for example, in the Double or Dozen game (see \cref{sec:otherconcepts}). 

Many familiar economic environments can be viewed as potential games under restrictive assumptions---for instance, public-goods provision with homogeneous preferences, Cournot competition with linear demand, or congestion games with identical cost functions. When these assumptions are relaxed, the agents' incentives  often remain approximately aligned, but the potential structure may be lost. Hence, $\mathbf{k}$-approximate ordinal potential games are an economically meaningful class of games. 

For illustration, consider an oligopoly production environment in the spirit of \citet[Example 2]{rosenthal1973class}. Each firm chooses one of several available production processes, each of which uses a subset of primary factors. Satisficing behavior may arise naturally in such environments, as firms may, for example, exhibit decisional inescapability or inertia. When all firms share identical ordinal preferences---with each preferring to employ processes that rely on less congested factors---the game is an ordinal potential game and thus admits a pure Nash equilibrium \citep{rosenthal1973class,monderer1996potential}.  Introducing mild preference heterogeneity breaks the exact alignment of incentives and can destroy the potential structure. For example, suppose that each firm prefers to share the same process as a specific competitor when factors are otherwise equally congested. This could be due to reputational benefits or standardization of tooling and workforce skills among downstream customers. The resulting game is a $\V{2}{2}$-approximate ordinal potential game and thus, by \cref{thm:approximatepotential}, admits a $\V{2}{2}$-satisficing equilibrium, but it is not a potential game and a pure Nash equilibrium may fail to exist. A formal treatment of this game is provided in Supplementary Appendix II. 

Similar arguments apply to other nearly potential environments, such as public-goods games with idiosyncratic components that depend on all agents’ actions, or Cournot competition with heterogeneous cost functions. The more preferences depart from being perfectly aligned, the larger $\mathbf k$ must be for a game to remain a $\mathbf k$-approximate ordinal potential game.

In sum, in \cref{sec:structural}, we established key structural and existence properties of satisficing equilibria. We showed that the concept gradually extends the classical equilibrium framework---retaining tractability while accommodating bounded rationality.

\section{Foundations}\label{sec:emergence}

We now consider why or how satisficing equilibria may emerge.  We assume that agents exhibit \emph{satisficing behavior} in accordance with  decision-theoretic foundations. 

As discussed in \cref{sec:introduction}, a large decision-theoretic and empirical literature documents satisficing behavior in individual choice problems. Theories include consideration sets, complexity, search,  attention to differences of first-order importance,  sparsity, and status-quo bias (see Supplementary Appendix IV.2 for a detailed review). Recently, \citet*{Bar22} show that modeling an agent as choosing among her top~$k_i$ actions is observationally equivalent---that is, empirically indistinguishable in a revealed preference sense---to a broad class of such theories. Moreover, beyond bounded rationality, there are also rational foundations for choosing one of one's top~$k_i$ actions. \cite*{Sen1993} shows that an agent may choose her second best action  due to internalized norms such as not wanting to appear ``greedy'' or wanting to show ``self-control''. Together, these decision-theoretic foundations motivate our assumption that agents exhibit $k_i$-satisficing behavior. Formally:

\begin{quote}
Agent $i$ exhibits $k_i$-\emph{satisficing behavior} in a game $\g$ at an action profile $\actionvec\in\profileset$ if she chooses or remains at an action $a_i$ such that she is  $k_i$-satisficed at $\actionvec=(a_i,\actionvec_{-i})$. 
\end{quote}

Below, we show how satisficing behavior leads to satisficing equilibrium, both epistemically---requiring only weak knowledge and rationality assumptions---and dynamically---arising as the limit of a simple decentralized adjustment process. In addition, Supplementary Appendix I provides an axiomatization of satisficing equilibrium in line with the classic axiomatization of Nash equilibrium by \cite*{peleg1996consistency}.

\subsection{Epistemic conditions}

The following observation provides epistemic conditions for satisficing equilibrium.

\begin{claim}\label{obs:static}   Fix a game $\g$, an action profile $\actionvec \in \profileset$, and $\mathbf k\in\{1,2,\ldots\}^n$. Suppose that each agent $i\in\agentset$ exhibits $k_i$-satisficing behavior at $\actionvec$ and that each agent knows that the action choices of the other agents are $\actionvec_{-i}$. Then $\actionvec$ is a $\mathbf k$-satisficing equilibrium of $\g$. 
\end{claim}

\cref{obs:static} mirrors the epistemic conditions for pure Nash equilibrium given in \citet[][preliminary observation]{aumann1995epistemic}. It shows that mutual knowledge of action choices is sufficient for $\mathbf{k}$-satisficing equilibrium, while common knowledge is not required. No knowledge of others’ preferences or behavior is needed.\footnote{Note that, in \citet{aumann1995epistemic}, stronger epistemic conditions---mutual knowledge of rationality and common beliefs---become necessary only when mixed strategies and strategic uncertainty are introduced.}  Hence, $\mathbf{k}$-satisficing equilibrium passes the ``announcement test’’ \citep[see, e.g.,][]{Hol04}: if all agents announced their actions simultaneously, then no $k_i$-satisficing agent would want to reconsider if the announced action profile is a $\mathbf k$-satisficing equilibrium. The same logic applies to mediator recommendations and implicit or explicit agreements.

\subsection{Dynamic emergence}

Next, we turn to the dynamic emergence of satisficing equilibria. Dynamic foundations for equilibrium behavior go back to Cournot and Walras and have since been proposed to overcome strong assumptions such as forward-looking optimizing behavior or common knowledge of rationality \citep[see][]{weibull1995evolutionary,samuelson2016game}. This approach has also led to fruitful applications in the study of the emergence of social norms, macro-economic market fluctuations, and industrial organization \citep{young2004strategic,newton2018evolutionary}. We show that satisficing equilibrium emerges as a stable point of a dynamic process in which agents exhibit satisficing behavior and myopically adjust their actions over time in response to the changing actions of the other agents. 

Consider the \emph{top-$\mathbf k$ response dynamic} in which each agent becomes active at random times and chooses one of her top-$k_i$ actions in response to the current action profile. Formally, denote the initial, arbitrarily selected, action profile by $\actionvec(0)$. At each time $t\in \{1,2,3,\ldots\}$, each agent is active with probability $\lambda_i \in (0,1)$, independently of all other agents. An agent $i$ who is active at time $t$ for the first time chooses one of her top-$k_i$  actions in response to $\actionvec_{-i}(t-1)$, each with positive probability. An agent $i$ who is active at time $t$ and was last active at $t' < t$ considers the sequence of action profiles $\actionvec(t'),\dots,\actionvec(t-1)$. If the sequence has at least two distinct profiles, $i$ chooses one of her top-$k_i$  actions in response to $\actionvec_{-i}(t-1)$, each with positive probability; otherwise, $i$ continues to play the action $\action_i(t')$. Finally, an agent $i$ who is not active at time $t$ plays $\action_i(t)=\action_i(t-1)$. Observe that if $k_i=1$ for all $i\in \agentset$, the above dynamic reduces to the standard best-response dynamic with memory and inertia \citep{young2004strategic}.

\begin{theorem}\label{prop:dynamic}
The top-$\mathbf k$ response dynamic converges to a $\mathbf{k}$-satisficing equilibrium almost surely in all games with a bounded number of actions, except for a fraction that is exponentially small in the number of agents. Formally, for any integer $M \geq 2$, there exist positive constants $c$ and $\bar{n}$ such that for any set of games $\G{n}{\anumvec}$ with $n \geq \bar{n}$ and $\anumvec \in \{2,\ldots,M\}^n$, and for any $\mathbf k \in \{2,3,\ldots\}^n$,
    \[
\bigg| \left\{ g \in \G{n}{\anumvec} : \;
  \vcenter{\hbox{\text{\shortstack[l]{\emph{the top-$\mathbf k$ response dynamic converges}\\
    \emph{a.s. to a $\mathbf k$-satisficing equilibrium of $g$}}}}}
    \; \right\} \bigg|
\;\geq\; (1 - e^{-cn}) \, |\, \G{n}{\anumvec} \,|.
\]
\end{theorem}

This result shows that a simple dynamic converges to satisficing equilibrium in most games. More sophisticated learning dynamics, either motivated by experimentally observed human behavior or optimized to be played by machines, may therefore also be expected to converge to (selections of) satisficing equilibria. Together, our results demonstrate that the emergence of satisficing equilibria is supported by simple epistemic conditions and natural adaptive behavior.

%
%

\section{Empirical content}\label{sec:characterization}

In this section, we characterize the empirical content of satisficing equilibrium by deriving a tractable revealed-preference test. We leverage the fact that $\mathbf k$-satisficing equilibrium, like Nash equilibrium, is a unilateral concept. This allows us to consider the decision-problem of each agent separately by fixing the actions of all other agents, thus reducing $\mathbf k$-satisficing equilibrium to each agent exhibiting $k_i$-satisficing behavior in the resulting decision problem. This feature allows us to show that $\mathbf k$-satisficing equilibrium is falsifiable and, in fact, that there exists a computationally efficient testing procedure.

We adapt the framework of \cite*{sprumont2000testable} for testing pure Nash equilibrium play and combine it with the test for single-agent $k_i$-satisficing behavior proposed by \citet*{Bar22}. Consider a finite set of agents $N$, assume that there is a ``universal'' set of actions $A_i$ for each agent $i \in N$, and let $\profileset := \times_{i \in N} A_i$. A \emph{game form} $\gf$ consists of the set of agents $N$ and an action set $\varnothing \subset B_i \subseteq A_i$ for each $i \in N$. A \emph{dataset} $\{(\gf^\z,\mathbf{b}^\z)\}_{\z=1}^\Z$ with $\Z$ observations consists of a set of pairs $(\gf^\z,\mathbf{b}^\z)$ where $\gf^\z$ is a game form  and $\mathbf{b}^\z \in \times_{i\in N} B_i^\z$ is the chosen action profile in that game form. A dataset is \emph{$\mathbf{k}$-rationalizable} if there exists a set of preference relations $\{\succsim_i\}_{i \in N}$ where $\succsim_i \subseteq \profileset \times \profileset$ for each $i$ such that for each $d$, $\mathbf{b}^\z$ is a $\mathbf{k}$-satisficing equilibrium of $\gf^\z$ under the preferences $\{\succsim_i\}_{i \in N}$.\footnote{This can straightforwardly be generalized to $\{\mathbf{k}^\z\}_{\z=1}^\Z$-rationalizability which, for each $i$, would allow the threshold $k_i^\z$ to vary across the game forms $\gf^\z$ in the dataset. This variation could, for example, capture the varying complexity of the game forms.} We can now state our result that $\mathbf k$-satisficing equilibrium is falsifiable.
 
\begin{theorem}\label{thm:testing} For any agent set $N$ and  $\mathbf k\in\{1,2,\ldots\}^{|N|}$, $\mathbf k$-satisficing equilibrium is falsifiable. In fact, \cref{alg:main} (presented in \cref{app:algoproof}) determines whether a dataset $\{(h^\z,\mathbf{b}^\z)\}_{\z=1}^\Z$ is $\mathbf{k}$-rationalizable, polynomially in $|N|$, $\Z$, and $\max_i |A_i|$.
\end{theorem}
\cref{thm:testing} establishes whether a given dataset is consistent with $\mathbf k$-satisficing equilibrium, for a given $\mathbf k$. The specific $\mathbf k$ could arise from a particular theory the analyst has in mind. Alternatively, the analyst may be interested in finding the smallest possible $\mathbf{k}$ for which a dataset is $\mathbf{k}$-rationalizable. The following corollary states that there also exists an efficient testing procedure for this case.

\begin{corollary}\label{cor:testing}
    \cref{alg:main} (presented in \cref{app:algoproof}) can be adapted to find the smallest $\mathbf{k}$ for which a dataset $\{(\gf^\z,\mathbf{b}^\z)\}_{\z=1}^\Z$ is $\mathbf{k}$-rationalizable, polynomially in $|N|$, $\Z$, and $\max_i |A_i|$.
\end{corollary}

Together, \cref{thm:testing,cor:testing} establish the empirical content of satisficing equilibrium. As discussed in \cref{sec:introduction}, satisficing behavior is empirically indistinguishable from a broad class of theories of boundedly-rational decision-making. In our extension to games, one similarly cannot establish for what reason---e.g., consideration sets versus choice from lists---a data set is $\mathbf k$-rationalizable, which reflects the fact that our reduced-form definition of satisficing equilibrium captures everything that is falsifiable in a revealed-preference sense.

Implementing the revealed-preference test in practice would require data that are, to our knowledge, not currently available. It would  be of interest to explore comparative statics, such as whether the minimal $\mathbf{k}$ that explains the observed play increases with the complexity of the game or decreases with the agents' sophistication.

An alternative empirical route would be to induce preferences directly, for instance, through monetary rewards in laboratory experiments. Existing experimental datasets, such as those of \citet{wright2017predicting} or \citet{fudenberg2019predicting}, have action spaces that are too small to systematically evaluate satisficing equilibrium. This and related limitations also hold for other datasets from the literature (see the discussion in  Supplementary Appendix IV.3). Thus, a new experimental study would be required to evaluate the theory quantitatively.

\section{Conclusion}\label{sec:conclusion}

We have introduced satisficing equilibrium, an ordinal relaxation of Nash equilibrium, in which each agent chooses one of her top $k_i$  actions in response to the other agents' actions. We have shown that this relaxation preserves much of the discipline that makes equilibrium analysis useful: in most games, satisficing equilibrium provides sharp predictions, exists under substantially weaker conditions than pure Nash equilibrium, admits natural epistemic and dynamic foundations, and yields empirically testable restrictions on data. In this sense, our article identifies a tractable middle ground between full optimization and unconstrained behavioral flexibility in games.

Several questions follow naturally. On the positive side, empirical work is needed to measure the $k_i$s that rationalize observed behavior and to study how they vary with game complexity and agent sophistication. When multiple satisficing equilibria exist, equilibrium selection also becomes a central concern, both for prediction and for welfare analysis.  On the normative side, relaxing optimization enlarges the equilibrium set and can improve or worsen welfare relative to Nash, suggesting that the relevant policy benchmark may change. We also view an endogenous theory of the parameters $k_i$ as an important avenue for future research. Similarly, an analysis of the mixed extension and resulting questions on beliefs and belief hierarchies may be of interest. Finally, satisficing equilibrium may offer a useful tool for mechanism and market design, potentially enabling the implementation of outcomes not attainable under other solution concepts.

\appendix

\section*{Appendix}

\cref{app:gamesmain} contains the omitted games mentioned in \cref{sec:otherconcepts} and \cref{app:proofs} contains all the proofs.

The Supplementary Appendix contains additional results, definitions, and further discussion of related literature (see \url{https://arxiv.org/abs/2409.00832}).

\section{Additional games}\label{app:gamesmain}

The bimatrix representations of all games discussed here and in the main text are in Supplementary Appendix III.

\subsection{The Halving game}\label{app:110game}

We introduce the Halving game, which is explicitly designed to contrast satisficing equilibrium and level-$\mathpzc k$ thinking.

\begin{quote}
\textbf{The Halving game.}
You and another player each request between 1 and 10 marbles, and normally each receives the amount she requests. 
There are two exceptions: 
if one player requests the largest integer weakly below half of the other player’s request, then the half-requesting player receives 11 marbles and the other receives 0; 
and if one player requests 1 while the other requests strictly more than 5, then the former receives 1  and the latter receives 4 marbles.
\end{quote}

We assume that the preferences of each agent are such that more marbles for oneself are strictly better. The game does not admit a $(1,1)$-satisficing equilibrium---that is, there is no pure Nash equilibrium. There is a unique $(2,2)$-satisficing equilibrium at action profile $(10,10)$, the $(3,3)$-satisficing equilibria, in addition, include $(10,9),(9,10),(9,9)$, and the $(4,4)$-satisficing equilibria, in addition, include $(10,8),(8,10),(9,8),(8,9),(8,8)$. See \cref{fig:Halving_concepts}(a).

\begin{figure}[ht]
\caption{The Halving game: predictions.}
\centering
\adjustbox{width=1.0\linewidth}{
\begin{tabular}{cccc}
 \begin{tikzpicture}[scale=0.6]
  \draw[thin,black] (0.5,-0.5) rectangle (10.5,-10.5);

\tikzmath{\del = 0.05; \eps1 = 2.5; \eps2 = 0; \eps3 = 2;}

\draw[very thick,draw=none, fill=HTML1,rounded corners = 0.5mm,opacity=1]
  (0.5 - \eps1*\del,  -0.5 + \eps1*\del)
    rectangle
  (3.5 + \eps1*\del,  -3.5 - \eps1*\del);

\draw[very thick,draw=none, fill=HTML4,rounded corners = 0.8mm,opacity=1]
  (0.5 + \eps2*\del,  -0.5 - \eps2*\del)
    rectangle
  (2.5 - \eps2*\del,  -2.5 + \eps2*\del);

\draw[very thick,draw=none, fill=HTML0,rounded corners = 1mm,opacity=1]
  (0.5 + \eps3*\del,  -0.5 - \eps3*\del)
    rectangle
  (1.5 - \eps3*\del,  -1.5 + \eps3*\del);

 \node[] at (5.5,-9) {\Large \textcolor{HTML0}{\textbf{2-SE}} $\subseteq$ \textcolor{HTML4}{\textbf{3-SE}} $\subseteq$ \textcolor{HTML1}{\textbf{4-SE}} };
  
  \foreach \i in {1,...,10} {
    \node[font=\small] at (0,-\i) {\the\numexpr11-\i\relax};
    \node[font=\small] at (\i,0) {\the\numexpr11-\i\relax};
  }
\end{tikzpicture} &  \begin{tikzpicture}[scale=0.6]

  \draw[thin,black] (0.5,-0.5) rectangle (10.5,-10.5);

  \draw[thin,black] (5.5,-5.5) rectangle (10.5,-10.5);

\begin{scope}
    \fill[pattern=north east lines, pattern color=gray]
      (5.5,-5.5) rectangle (10.5,-10.5) -- cycle;
  \end{scope}

  \foreach \i in {1,...,10} {
    \node[font=\small] at (0,-\i) {\the\numexpr11-\i\relax};
    \node[font=\small] at (\i,0) {\the\numexpr11-\i\relax};
  }

\node[] at (8,-7.5) {\Large Level-$\mathpzc{k}$};

\node[] at (8,-9) {\Large  {$\mathpzc{k}\geq \mathpzc{2}$ }};
  
\end{tikzpicture} 
 & \pgfplotstableread{
 1.0  1.0  0.0  1.0  0.0  1.0  1.0  1.0  1.0  1.0
 1.0  1.0  0.0  1.0  0.0  1.0  1.0  1.0  1.0  1.0
 0.0  0.0  0.0  0.0  0.0  0.0  0.0  0.0  0.0  0.0
 1.0  1.0  0.0  1.0  0.0  1.0  1.0  1.0  1.0  1.0
 0.0  0.0  0.0  0.0  0.0  0.0  0.0  0.0  0.0  0.0
 1.0  1.0  0.0  1.0  0.0  1.0  1.0  1.0  1.0  1.0
 1.0  1.0  0.0  1.0  0.0  1.0  1.0  1.0  1.0  1.0
 1.0  1.0  0.0  1.0  0.0  1.0  1.0  1.0  1.0  1.0
 1.0  1.0  0.0  1.0  0.0  1.0  1.0  1.0  1.0  1.0
 1.0  1.0  0.0  1.0  0.0  1.0  1.0  1.0  1.0  1.0
}\data

\begin{tikzpicture}[scale=0.6]

  \foreach \r in {0,...,9} {
    \foreach \c in {0,...,9} {
      \pgfplotstablegetelem{\r}{[index]\c}\of{\data}%
      \pgfmathsetmacro{\val}{\pgfplotsretval}%

      \pgfmathparse{\val>0.5}%
      \ifnum\pgfmathresult=1
        \node[
          pattern=north east lines,
          pattern color=gray,
          minimum size=6mm
        ] at ({\c + 1},{-\r -1}) {};
      \else
        \node[
          fill=white,
          minimum size=6mm
        ] at ({\c + 1},{-\r -1}) {};
      \fi
    }
  }

  \draw[thin,black] (0.5,-0.5) rectangle (10.5,-10.5);

  \draw[thin,black] (5.5,-5.5) rectangle (10.5,-10.5);
  \draw[thin,black] (5.5,-3.5) rectangle (10.5,-4.5);
  \draw[thin,black] (5.5,-0.5) rectangle (10.5,-2.5);

  \draw[thin,black] (3.5,-5.5) rectangle (4.5,-10.5);
  \draw[thin,black] (3.5,-3.5) rectangle (4.5,-4.5);
  \draw[thin,black] (3.5,-0.5) rectangle (4.5,-2.5);

  \draw[thin,black] (0.5,-5.5) rectangle (2.5,-10.5);
  \draw[thin,black] (0.5,-3.5) rectangle (2.5,-4.5);
  \draw[thin,black] (0.5,-0.5) rectangle (2.5,-2.5);
  
  \foreach \i in {1,...,10} {
    \node[font=\small] at (0,-\i) {\the\numexpr11-\i\relax};
    \node[font=\small] at (\i,0) {\the\numexpr11-\i\relax};
  }
\end{tikzpicture} & \pgfplotstableread{
 0.4444  0.1556  0.0  0.0  0.0  0.0667  0.0  0.0  0.0  0.0
 0.1556  0.0544  0.0  0.0  0.0  0.0233  0.0  0.0  0.0  0.0
 0.0     0.0     0.0  0.0  0.0  0.0     0.0  0.0  0.0  0.0
 0.0     0.0     0.0  0.0  0.0  0.0     0.0  0.0  0.0  0.0
 0.0     0.0     0.0  0.0  0.0  0.0     0.0  0.0  0.0  0.0
 0.0667  0.0233  0.0  0.0  0.0  0.01    0.0  0.0  0.0  0.0
 0.0     0.0     0.0  0.0  0.0  0.0     0.0  0.0  0.0  0.0
 0.0     0.0     0.0  0.0  0.0  0.0     0.0  0.0  0.0  0.0
 0.0     0.0     0.0  0.0  0.0  0.0     0.0  0.0  0.0  0.0
 0.0     0.0     0.0  0.0  0.0  0.0     0.0  0.0  0.0  0.0
}\data

\begin{tikzpicture}[scale=0.6]
  \def\minval{0.001}
  \def\maxval{0.25} 

  \foreach \r in {0,...,9} {
    \foreach \c in {0,...,9} {
      \pgfplotstablegetelem{\r}{[index]\c}\of{\data}%
      \pgfmathsetmacro{\val}{\pgfplotsretval}%

      \pgfmathsetmacro{\pct}{%
        ifthenelse(
          \val>\minval,
          (ln(\val + \minval)-ln(2*\minval)) / (ln(\maxval+\minval) -ln(2*\minval) ) * 100, 0
          )%
      }%

      \node[
        fill=gray!\pct!white,
        minimum size=6mm
      ] at ({\c + 1},{-\r -1}) {};
    }
  }

  \draw[thin,black] (0.5,-0.5) rectangle (10.5,-10.5);

  \foreach \i in {1,...,10} {
    \node[font=\small] at (0,-\i) {\the\numexpr11-\i\relax};
    \node[font=\small] at (\i,0) {\the\numexpr11-\i\relax};
  }
\end{tikzpicture}   \\
 (a) Satisficing equilibria & (b) Level-$\mathpzc k$ & (c) Rationalizable & (d) Symmetric Nash  \\
\pgfplotstableread{
 0.017   0.0158  0.0146  0.0136  0.0125  0.0126  0.0125  0.0115  0.0106  0.0098
 0.0158  0.0146  0.0135  0.0126  0.0116  0.0117  0.0116  0.0106  0.0099  0.0091
 0.0146  0.0135  0.0125  0.0116  0.0107  0.0108  0.0107  0.0098  0.0091  0.0084
 0.0136  0.0126  0.0116  0.0108  0.01    0.01    0.01    0.0091  0.0085  0.0078
 0.0125  0.0116  0.0107  0.01    0.0092  0.0093  0.0092  0.0084  0.0078  0.0072
 0.0126  0.0117  0.0108  0.01    0.0093  0.0093  0.0093  0.0085  0.0079  0.0072
 0.0125  0.0116  0.0107  0.01    0.0092  0.0093  0.0092  0.0084  0.0078  0.0072
 0.0115  0.0106  0.0098  0.0091  0.0084  0.0085  0.0084  0.0077  0.0072  0.0066
 0.0106  0.0099  0.0091  0.0085  0.0078  0.0079  0.0078  0.0072  0.0067  0.0061
 0.0098  0.0091  0.0084  0.0078  0.0072  0.0072  0.0072  0.0066  0.0061  0.0056
}\data

\begin{tikzpicture}[scale=0.6]
  \def\minval{0.001}
  \def\maxval{0.25} 

  \foreach \r in {0,...,9} {
    \foreach \c in {0,...,9} {
      \pgfplotstablegetelem{\r}{[index]\c}\of{\data}%
      \pgfmathsetmacro{\val}{\pgfplotsretval}%

      \pgfmathsetmacro{\pct}{%
        ifthenelse(
          \val>\minval,
          (ln(\val + \minval)-ln(2*\minval)) / (ln(\maxval+\minval) -ln(2*\minval) ) * 100, 0
          )%
      }%

      \node[
        fill=gray!\pct!white,
        minimum size=6mm
      ] at ({\c + 1},{-\r -1}) {};
    }
  }

  \draw[thin,black] (0.5,-0.5) rectangle (10.5,-10.5);

  \foreach \i in {1,...,10} {
    \node[font=\small] at (0,-\i) {\the\numexpr11-\i\relax};
    \node[font=\small] at (\i,0) {\the\numexpr11-\i\relax};
  }
\end{tikzpicture} & \pgfplotstableread{
 0.0657  0.0478  0.0334  0.0261  0.0179  0.0233  0.0208  0.0103  0.0075  0.0036
 0.0478  0.0347  0.0242  0.019   0.013   0.017   0.0151  0.0075  0.0054  0.0026
 0.0334  0.0242  0.0169  0.0133  0.0091  0.0118  0.0105  0.0052  0.0038  0.0018
 0.0261  0.019   0.0133  0.0104  0.0071  0.0093  0.0082  0.0041  0.003   0.0014
 0.0179  0.013   0.0091  0.0071  0.0049  0.0064  0.0057  0.0028  0.002   0.001
 0.0233  0.017   0.0118  0.0093  0.0064  0.0083  0.0074  0.0037  0.0026  0.0013
 0.0208  0.0151  0.0105  0.0082  0.0057  0.0074  0.0066  0.0033  0.0024  0.0011
 0.0103  0.0075  0.0052  0.0041  0.0028  0.0037  0.0033  0.0016  0.0012  0.0006
 0.0075  0.0054  0.0038  0.003   0.002   0.0026  0.0024  0.0012  0.0008  0.0004
 0.0036  0.0026  0.0018  0.0014  0.001   0.0013  0.0011  0.0006  0.0004  0.0002
}\data

\begin{tikzpicture}[scale=0.6]
  \def\minval{0.001}
  \def\maxval{0.25} 

  \foreach \r in {0,...,9} {
    \foreach \c in {0,...,9} {
      \pgfplotstablegetelem{\r}{[index]\c}\of{\data}%
      \pgfmathsetmacro{\val}{\pgfplotsretval}%

      \pgfmathsetmacro{\pct}{%
        ifthenelse(
          \val>\minval,
          (ln(\val + \minval)-ln(2*\minval)) / (ln(\maxval+\minval) -ln(2*\minval) ) * 100, 0
          )%
      }%

      \node[
        fill=gray!\pct!white,
        minimum size=6mm
      ] at ({\c + 1},{-\r -1}) {};
    }
  }

  \draw[thin,black] (0.5,-0.5) rectangle (10.5,-10.5);

  \foreach \i in {1,...,10} {
    \node[font=\small] at (0,-\i) {\the\numexpr11-\i\relax};
    \node[font=\small] at (\i,0) {\the\numexpr11-\i\relax};
  }
\end{tikzpicture}& \pgfplotstableread{
 0.1884  0.1214  0.044   0.0201  0.007   0.0381  0.0142  0.0005  0.0003  0.0
 0.1214  0.0783  0.0283  0.0129  0.0045  0.0245  0.0092  0.0003  0.0002  0.0
 0.044   0.0283  0.0103  0.0047  0.0016  0.0089  0.0033  0.0001  0.0001  0.0
 0.0201  0.0129  0.0047  0.0021  0.0007  0.0041  0.0015  0.0001  0.0     0.0
 0.007   0.0045  0.0016  0.0007  0.0003  0.0014  0.0005  0.0     0.0     0.0
 0.0381  0.0245  0.0089  0.0041  0.0014  0.0077  0.0029  0.0001  0.0001  0.0
 0.0142  0.0092  0.0033  0.0015  0.0005  0.0029  0.0011  0.0     0.0     0.0
 0.0005  0.0003  0.0001  0.0001  0.0     0.0001  0.0     0.0     0.0     0.0
 0.0003  0.0002  0.0001  0.0     0.0     0.0001  0.0     0.0     0.0     0.0
 0.0     0.0     0.0     0.0     0.0     0.0     0.0     0.0     0.0     0.0
}\data

\begin{tikzpicture}[scale=0.6]
  \def\minval{0.001}
  \def\maxval{0.25} 

  \foreach \r in {0,...,9} {
    \foreach \c in {0,...,9} {
      \pgfplotstablegetelem{\r}{[index]\c}\of{\data}%
      \pgfmathsetmacro{\val}{\pgfplotsretval}%

      \pgfmathsetmacro{\pct}{%
        ifthenelse(
          \val>\minval,
          (ln(\val + \minval)-ln(2*\minval)) / (ln(\maxval+\minval) -ln(2*\minval) ) * 100, 0
          )%
      }%

      \node[
        fill=gray!\pct!white,
        minimum size=6mm
      ] at ({\c + 1},{-\r -1}) {};
    }
  }

  \draw[thin,black] (0.5,-0.5) rectangle (10.5,-10.5);

  \foreach \i in {1,...,10} {
    \node[font=\small] at (0,-\i) {\the\numexpr11-\i\relax};
    \node[font=\small] at (\i,0) {\the\numexpr11-\i\relax};
  }
\end{tikzpicture} & \pgfplotstableread{
 0.0641  0.0247  0.0057  0.0014  0.0002  0.0472  0.0209  0.0011  0.0506  0.0372
 0.0247  0.0095  0.0022  0.0006  0.0001  0.0182  0.0081  0.0004  0.0195  0.0143
 0.0057  0.0022  0.0005  0.0001  0.0     0.0042  0.0019  0.0001  0.0045  0.0033
 0.0014  0.0006  0.0001  0.0     0.0     0.0011  0.0005  0.0     0.0011  0.0008
 0.0002  0.0001  0.0     0.0     0.0     0.0002  0.0001  0.0     0.0002  0.0001
 0.0472  0.0182  0.0042  0.0011  0.0002  0.0348  0.0154  0.0008  0.0373  0.0274
 0.0209  0.0081  0.0019  0.0005  0.0001  0.0154  0.0068  0.0004  0.0165  0.0121
 0.0011  0.0004  0.0001  0.0     0.0     0.0008  0.0004  0.0     0.0009  0.0006
 0.0506  0.0195  0.0045  0.0011  0.0002  0.0373  0.0165  0.0009  0.04    0.0294
 0.0372  0.0143  0.0033  0.0008  0.0001  0.0274  0.0121  0.0006  0.0294  0.0216
}\data

\begin{tikzpicture}[scale=0.6]
  \def\minval{0.001}
  \def\maxval{0.25} 

  \foreach \r in {0,...,9} {
    \foreach \c in {0,...,9} {
      \pgfplotstablegetelem{\r}{[index]\c}\of{\data}%
      \pgfmathsetmacro{\val}{\pgfplotsretval}%

      \pgfmathsetmacro{\pct}{%
        ifthenelse(
          \val>\minval,
          (ln(\val + \minval)-ln(2*\minval)) / (ln(\maxval+\minval) -ln(2*\minval) ) * 100, 0
          )%
      }%

      \node[
        fill=gray!\pct!white,
        minimum size=6mm
      ] at ({\c + 1},{-\r -1}) {};
    }
  }

  \draw[thin,black] (0.5,-0.5) rectangle (10.5,-10.5);

  \foreach \i in {1,...,10} {
    \node[font=\small] at (0,-\i) {\the\numexpr11-\i\relax};
    \node[font=\small] at (\i,0) {\the\numexpr11-\i\relax};
  }
\end{tikzpicture}   \\
(e) QRE ($\lambda = 0.1$) & (f) QRE ($\lambda = 0.4$)  & (g) QRE ($\lambda = 1$) & (h)  Sampling equilibrium 
\end{tabular}
}
\noindent\begin{minipage}{1\textwidth} 
\vspace{0.2cm}
\small{\textbf{Note.} 
Each panel plots agent 1's action on the vertical axis and agent 2's on the horizontal axis.  
Panel~(a) shows the $(k,k)$-satisficing equilibria, or $k$-SE, for $k\in\{1,2\}$ (note the nested structure). 
Panel~(b) shows the possible actions for level-$\mathpzc k$-thinkers with $\mathpzc k\geq \mathpzc 2$. 
Panel~(c) shows the set of rationalizable action profiles.  
Panels~(d)–(h) display predicted probabilities under the unique symmetric Nash equilibrium, three logit quantal-response equilibria (QRE) with rationality parameters $\lambda\in\{0.1,0.4,1\}$, and the unique sampling equilibrium. 
Darker shading indicates higher predicted probability; 
white denotes zero.
}
\end{minipage}
\label{fig:Halving_concepts}
\end{figure}

Turning to level-$\mathpzc k$ thinking, first consider the anchor $10$. Then, level-$\mathpzc 1$ thinking yields action $5$, level-$\mathpzc 2$ thinking  action $2$, and level-$\mathpzc 3$ thinking  action $1$. For higher levels of $\mathpzc k$ the best responses cycle through the actions $5$, $2$, and $1$. For a uniformly distributed anchor the picture is similar; however, the cycle $5,2,1$ is only reached from level-$\mathpzc 2$ thinking onward. In fact, it is easy to check that for \emph{any} level-$\mathpzc 0$ anchor, level-$\mathpzc k$ thinkers with $\mathpzc k\geq \mathpzc  2$ play action $5$ or below (see \cref{fig:Halving_concepts}(b)). In lab experiments, most subjects are often found to be level-$\mathpzc 2$ or level-$\mathpzc 3$ thinkers \citep*{arad201211,crawford2013structural}. Hence, the prediction of level-$\mathpzc k$ thinking is opposite to the prediction of $\mathbf k$-satisficing equilibrium for low $\mathbf k$.  For completeness, \cref{fig:Halving_concepts} includes the predictions of the equilibrium concepts considered in \cref{sec:otherconcepts}. A detailed explanation of the respective computations is given in Supplementary Appendix IV.1.\footnote{For quantal-response equilibrium, because the sensitivity of the predicted distributions of play vary by game, we show varying values of $\lambda$ for the different games.}

\subsection{The Traveler's Dilemma \citep{basu1994traveler}}\label{sec:TD}
The Traveler's Dilemma was introduced by \cite*{basu1994traveler} to challenge the presumption of (economic) rationality, without requiring sequential play as in previously proposed paradoxes such as the centipede game. 

\begin{quotation}
\noindent \textbf{The Traveler's Dilemma \citep[][not verbatim]{basu1994traveler}.} 
Two players each request an integer between \$2 and \$11. If both write the same number, each will receive that amount. But if they write different numbers,  both of them receive the lower number along with a bonus or penalty: the player who wrote the lower number will get \$2 more and the one who wrote the higher number will get \$2 less.
\end{quotation}

The Traveler's Dilemma was originally proposed with action sets for each agent  $\{2,3,\ldots, 100\}$ and has since been experimentally studied for various numbers of available actions. Our restriction to 10 actions per agent allows for a similar presentation to and a closer comparison with the other considered games.

\begin{figure}[ht]
\caption{The Traveler's Dilemma: predictions.}
\centering
\adjustbox{width=1.0\linewidth}{
\begin{tabular}{cccc}
  \begin{tikzpicture}[scale=0.6]
  \draw[thin,black] (0.5,-0.5) rectangle (10.5,-10.5);

 \draw[very thick,draw=none,fill=HTML0,rounded corners = 0.75mm,opacity=1, scale around={1.02:(5.5,-5.5)}]
    (0.5,-1.5)
    \foreach \i in {1,...,7}{
      -- ({\i+0.5}, {-\i-0.5})   
      -- ({\i+0.5}, {-\i-1.5})   
    }
    -- (7.5,-8.5)
    -- (8.5,-8.5)
    -- (8.5,-9.5)
    -- (7.5,-9.5)
    -- (7.5,-10.5)
    
    -- (10.5,-10.5)

    -- (10.5,-7.5)  
    -- (9.5,-7.5)
    -- (9.5,-8.5)
    -- (8.5,-8.5)
    \foreach \j in {9,...,2} {
      -- ({\j-0.5}, {-\j+0.5})
      -- ({\j-0.5}, {-\j+1.5})
    }
    -- (0.5,-0.5)
    -- cycle;

\draw[very thick,draw=none,fill=HTML6,rounded corners = 0.75mm,opacity=1, scale around={0.9:(10,-10)}] (9.5 ,-10.5) rectangle (10.5,-9.5);
  
 \node[] at (4,-9) {\Large \textcolor{HTML6}{{\textbf{1-SE}}} $\subseteq$ \textcolor{HTML0}{\textbf{2-SE}}};

  \foreach \i in {1,...,10} {
    \node[font=\small] at (0,-\i) {\the\numexpr12-\i\relax};
    \node[font=\small] at (\i,0) {\the\numexpr12-\i\relax};
  }
\end{tikzpicture} & \begin{tikzpicture}[scale=0.6]
  \draw[thin,black] (0.5,-0.5) rectangle (10.5,-10.5);
  \draw[thin,black] (9.5,-10.5) rectangle (10.5,-9.5);
    \begin{scope}
    \fill[pattern=north east lines, pattern color=gray]
      (9.5,-10.5) rectangle (10.5,-9.5) -- cycle;
  \end{scope}

  \foreach \i in {1,...,10} {
    \node[font=\small] at (0,-\i) {\the\numexpr12-\i\relax};
    \node[font=\small] at (\i,0) {\the\numexpr12-\i\relax};
  }
\end{tikzpicture} & \pgfplotstableread{
 0.0     0.0    0.0     0.0   0.0     0.0     0.0  0.0  0.0  0.0
 0.0     0.0    0.0     0.0   0.0     0.0     0.0  0.0  0.0  0.0
 0.0     0.0    0.0     0.0   0.0     0.0     0.0  0.0  0.0  0.0
 0.0     0.0    0.0     0.0   0.0     0.0     0.0  0.0  0.0  0.0
 0.0     0.0    0.0     0.0   0.0     0.0     0.0  0.0  0.0  0.0
 0.0     0.0    0.0     0.0   0.0     0.0     0.0  0.0  0.0  0.0
 0.0     0.0    0.0     0.0   0.0     0.0     0.0  0.0  0.0  0.0
 0.0     0.0    0.0     0.0   0.0     0.0     0.0  0.0  0.0  0.0
 0.0     0.0    0.0     0.0   0.0     0.0     0.0  0.0  0.0  0.0
 0.0     0.0    0.0     0.0   0.0     0.0     0.0  0.0  0.0  1.0
}\data

\begin{tikzpicture}[scale=0.6]
  \def\minval{0.001}
  \def\maxval{0.25}

  \foreach \r in {0,...,9} {
    \foreach \c in {0,...,9} {
      \pgfplotstablegetelem{\r}{[index]\c}\of{\data}%
      \pgfmathsetmacro{\val}{\pgfplotsretval}%

      \pgfmathsetmacro{\pct}{%
        ifthenelse(
          \val>\minval,
          (ln(\val + \minval)-ln(2*\minval)) / (ln(\maxval+\minval) -ln(2*\minval) ) * 100, 0
          )%
      }%

      \node[
        fill=gray!\pct!white,
        minimum size=6mm
      ] at ({\c + 1},{-\r -1}) {};
    }
  }

  \draw[thin,black] (0.5,-0.5) rectangle (10.5,-10.5);

  \foreach \i in {1,...,10} {
    \node[font=\small] at (0,-\i) {\the\numexpr12-\i\relax};
    \node[font=\small] at (\i,0) {\the\numexpr12-\i\relax};
  }
\end{tikzpicture}  &  \pgfplotstableread{
 0.0075  0.0086  0.0097  0.0106  0.011   0.0108  0.0097  0.0081  0.0062  0.0043
 0.0086  0.01    0.0112  0.0122  0.0127  0.0124  0.0112  0.0093  0.0071  0.005
 0.0097  0.0112  0.0126  0.0138  0.0143  0.014   0.0127  0.0105  0.008   0.0056
 0.0106  0.0122  0.0138  0.015   0.0156  0.0152  0.0138  0.0114  0.0087  0.0061
 0.011   0.0127  0.0143  0.0156  0.0162  0.0158  0.0143  0.0119  0.009   0.0064
 0.0108  0.0124  0.014   0.0152  0.0158  0.0155  0.014   0.0116  0.0088  0.0062
 0.0097  0.0112  0.0127  0.0138  0.0143  0.014   0.0127  0.0105  0.008   0.0056
 0.0081  0.0093  0.0105  0.0114  0.0119  0.0116  0.0105  0.0088  0.0067  0.0047
 0.0062  0.0071  0.008   0.0087  0.009   0.0088  0.008   0.0067  0.0051  0.0036
 0.0043  0.005   0.0056  0.0061  0.0064  0.0062  0.0056  0.0047  0.0036  0.0025
}\data

\begin{tikzpicture}[scale=0.6]
  \def\minval{0.001}
  \def\maxval{0.25}   

  \foreach \r in {0,...,9} {
    \foreach \c in {0,...,9} {
      \pgfplotstablegetelem{\r}{[index]\c}\of{\data}%
      \pgfmathsetmacro{\val}{\pgfplotsretval}%

      \pgfmathsetmacro{\pct}{%
        ifthenelse(
          \val>\minval,
          (ln(\val + \minval)-ln(2*\minval)) / (ln(\maxval+\minval) -ln(2*\minval) ) * 100, 0
          )%
      }%

      \node[
        fill=gray!\pct!white,
        minimum size=6mm
      ] at ({\c + 1},{-\r -1}) {};
    }
  }

  \draw[thin,black] (0.5,-0.5) rectangle (10.5,-10.5);

  \foreach \i in {1,...,10} {
    \node[font=\small] at (0,-\i) {\the\numexpr12-\i\relax};
    \node[font=\small] at (\i,0) {\the\numexpr12-\i\relax};
  }
\end{tikzpicture} \\
 (a) Satisficing equilibria & (b) Rationalizable & (c) Nash equilibrium & (d) QRE ($\lambda = 0.5$) \\
 \pgfplotstableread{
 0.0012  0.0015  0.0019  0.0025  0.0032  0.0041  0.0052  0.0064  0.006   0.0025
 0.0015  0.0019  0.0025  0.0032  0.0041  0.0052  0.0066  0.0082  0.0076  0.0032
 0.0019  0.0025  0.0032  0.0041  0.0052  0.0067  0.0085  0.0105  0.0098  0.0041
 0.0025  0.0032  0.0041  0.0052  0.0067  0.0085  0.0109  0.0135  0.0125  0.0052
 0.0032  0.0041  0.0052  0.0067  0.0085  0.0109  0.0139  0.0172  0.0161  0.0066
 0.0041  0.0052  0.0067  0.0085  0.0109  0.014   0.0178  0.022   0.0205  0.0085
 0.0052  0.0066  0.0085  0.0109  0.0139  0.0178  0.0227  0.028   0.0261  0.0108
 0.0064  0.0082  0.0105  0.0135  0.0172  0.022   0.028   0.0347  0.0323  0.0134
 0.006   0.0076  0.0098  0.0125  0.0161  0.0205  0.0261  0.0323  0.0302  0.0125
 0.0025  0.0032  0.0041  0.0052  0.0066  0.0085  0.0108  0.0134  0.0125  0.0052
}\data

\begin{tikzpicture}[scale=0.6]
  \def\minval{0.001}
  \def\maxval{0.25}   
  
  \foreach \r in {0,...,9} {
    \foreach \c in {0,...,9} {
      \pgfplotstablegetelem{\r}{[index]\c}\of{\data}%
      \pgfmathsetmacro{\val}{\pgfplotsretval}%

      \pgfmathsetmacro{\pct}{%
        ifthenelse(
          \val>\minval,
          (ln(\val + \minval)-ln(2*\minval)) / (ln(\maxval+\minval) -ln(2*\minval) ) * 100, 0
          )%
      }%

      \node[
        fill=gray!\pct!white,
        minimum size=6mm
      ] at ({\c + 1},{-\r -1}) {};
    }
  }

  \draw[thin,black] (0.5,-0.5) rectangle (10.5,-10.5);

  \foreach \i in {1,...,10} {
    \node[font=\small] at (0,-\i) {\the\numexpr12-\i\relax};
    \node[font=\small] at (\i,0) {\the\numexpr12-\i\relax};
  }
\end{tikzpicture}
 &  \pgfplotstableread{
 0.0     0.0     0.0     0.0     0.0     0.0     0.0     0.0     0.0     0.0003
 0.0     0.0     0.0     0.0     0.0     0.0     0.0     0.0     0.0     0.0003
 0.0     0.0     0.0     0.0     0.0     0.0     0.0     0.0     0.0     0.0003
 0.0     0.0     0.0     0.0     0.0     0.0     0.0     0.0     0.0     0.0003
 0.0     0.0     0.0     0.0     0.0     0.0     0.0     0.0     0.0     0.0003
 0.0     0.0     0.0     0.0     0.0     0.0     0.0     0.0     0.0     0.0003
 0.0     0.0     0.0     0.0     0.0     0.0     0.0     0.0     0.0     0.0003
 0.0     0.0     0.0     0.0     0.0     0.0     0.0     0.0     0.0     0.0003
 0.0     0.0     0.0     0.0     0.0     0.0     0.0     0.0     0.0     0.0003
 0.0003  0.0003  0.0003  0.0003  0.0003  0.0003  0.0003  0.0003  0.0003  0.9945
}\data

\begin{tikzpicture}[scale=0.6]
  \def\minval{0.001}
  \def\maxval{0.25} 

  \foreach \r in {0,...,9} {
    \foreach \c in {0,...,9} {
      \pgfplotstablegetelem{\r}{[index]\c}\of{\data}%
      \pgfmathsetmacro{\val}{\pgfplotsretval}%

      \pgfmathsetmacro{\pct}{%
        ifthenelse(
          \val>\minval,
          (ln(\val + \minval)-ln(2*\minval)) / (ln(\maxval+\minval) -ln(2*\minval) ) * 100, 0
          )%
      }%

      \node[
        fill=gray!\pct!white,
        minimum size=6mm
      ] at ({\c + 1},{-\r -1}) {};
    }
  }

  \draw[thin,black] (0.5,-0.5) rectangle (10.5,-10.5);

  \foreach \i in {1,...,10} {
    \node[font=\small] at (0,-\i) {\the\numexpr12-\i\relax};
    \node[font=\small] at (\i,0) {\the\numexpr12-\i\relax};
  }
\end{tikzpicture}&   \pgfplotstableread{
 0.0     0.0    0.0     0.0   0.0     0.0     0.0  0.0  0.0  0.0
 0.0     0.0    0.0     0.0   0.0     0.0     0.0  0.0  0.0  0.0
 0.0     0.0    0.0     0.0   0.0     0.0     0.0  0.0  0.0  0.0
 0.0     0.0    0.0     0.0   0.0     0.0     0.0  0.0  0.0  0.0
 0.0     0.0    0.0     0.0   0.0     0.0     0.0  0.0  0.0  0.0
 0.0     0.0    0.0     0.0   0.0     0.0     0.0  0.0  0.0  0.0
 0.0     0.0    0.0     0.0   0.0     0.0     0.0  0.0  0.0  0.0
 0.0     0.0    0.0     0.0   0.0     0.0     0.0  0.0  0.0  0.0
 0.0     0.0    0.0     0.0   0.0     0.0     0.0  0.0  0.0  0.0
 0.0     0.0    0.0     0.0   0.0     0.0     0.0  0.0  0.0  1.0
}\data

\begin{tikzpicture}[scale=0.6]
  \def\minval{0.001}
  \def\maxval{0.25} 

  \foreach \r in {0,...,9} {
    \foreach \c in {0,...,9} {
      \pgfplotstablegetelem{\r}{[index]\c}\of{\data}%
      \pgfmathsetmacro{\val}{\pgfplotsretval}%

      \pgfmathsetmacro{\pct}{%
        ifthenelse(
          \val>\minval,
          (ln(\val + \minval)-ln(2*\minval)) / (ln(\maxval+\minval) -ln(2*\minval) ) * 100, 0
          )%
      }%

      \node[
        fill=gray!\pct!white,
        minimum size=6mm
      ] at ({\c + 1},{-\r -1}) {};
    }
  }

  \draw[thin,black] (0.5,-0.5) rectangle (10.5,-10.5);

  \foreach \i in {1,...,10} {
    \node[font=\small] at (0,-\i) {\the\numexpr12-\i\relax};
    \node[font=\small] at (\i,0) {\the\numexpr12-\i\relax};
  }
\end{tikzpicture}&  \pgfplotstableread{
 0.0019  0.0053  0.0098  0.0131  0.0106  0.0026  0.0001  0.0  0.0  0.0
 0.0053  0.0148  0.0275  0.0368  0.0297  0.0073  0.0001  0.0  0.0  0.0
 0.0098  0.0275  0.0512  0.0685  0.0553  0.0137  0.0003  0.0  0.0  0.0
 0.0131  0.0368  0.0685  0.0917  0.0741  0.0183  0.0004  0.0  0.0  0.0
 0.0106  0.0297  0.0553  0.0741  0.0598  0.0148  0.0003  0.0  0.0  0.0
 0.0026  0.0073  0.0137  0.0183  0.0148  0.0036  0.0001  0.0  0.0  0.0
 0.0001  0.0001  0.0003  0.0004  0.0003  0.0001  0.0     0.0  0.0  0.0
 0.0     0.0     0.0     0.0     0.0     0.0     0.0     0.0  0.0  0.0
 0.0     0.0     0.0     0.0     0.0     0.0     0.0     0.0  0.0  0.0
 0.0     0.0     0.0     0.0     0.0     0.0     0.0     0.0  0.0  0.0
}\data

\begin{tikzpicture}[scale=0.6]
  \def\minval{0.001}
  \def\maxval{0.25}

  \foreach \r in {0,...,9} {
    \foreach \c in {0,...,9} {
      \pgfplotstablegetelem{\r}{[index]\c}\of{\data}%
      \pgfmathsetmacro{\val}{\pgfplotsretval}%

      \pgfmathsetmacro{\pct}{%
        ifthenelse(
          \val>\minval,
          (ln(\val + \minval)-ln(2*\minval)) / (ln(\maxval+\minval) -ln(2*\minval) ) * 100, 0
          )%
      }%

      \node[
        fill=gray!\pct!white,
        minimum size=6mm
      ] at ({\c + 1},{-\r -1}) {};
    }
  }

  \draw[thin,black] (0.5,-0.5) rectangle (10.5,-10.5);

  \foreach \i in {1,...,10} {
    \node[font=\small] at (0,-\i) {\the\numexpr12-\i\relax};
    \node[font=\small] at (\i,0) {\the\numexpr12-\i\relax};
  }
\end{tikzpicture}\\
 (e) QRE ($\lambda=2$) &(f) QRE ($\lambda = 4$) & (g) First sampling equilibrium & (h) Second sampling equilibrium
\end{tabular}
}

\noindent\begin{minipage}{1\textwidth} 
\vspace{0.2cm}
\small{\textbf{Note.} 
Each panel plots agent 1's actions on the vertical axis and agent 2's on the horizontal axis.  
Panel~(a) shows the $(k,k)$-satisficing equilibria, or $k$-SE, for $k\in\{1,2\}$ (note the nested structure). Panel~(b) shows the set of rationalizable action profiles.  
Panels~(c)–(h) display predicted probabilities under the unique symmetric Nash equilibrium, three logit quantal-response equilibria (QRE) with rationality parameters $\lambda\in\{0.5,2,4\}$, and the two sampling equilibria. 
Darker shading indicates higher predicted probability; 
white denotes zero.
}
\end{minipage}
\label{fig:TD_concepts}
\end{figure}

It is (implicitly) assumed that each agent’s preferences are such that more money for oneself is considered strictly better.
The game admits a unique $(1,1)$-satisficing equilibrium---that is, a unique pure Nash equilibrium---at action profile $(2,2)$. The $(2,2)$-satisficing equilibria consist of all the symmetric action profiles, that is, $(2,2),(3,3),\ldots,(11,11)$, as well as $(2,3)$, $(3,2)$, $(2,4)$, $(4,2)$. See \cref{fig:TD_concepts}(a). 

Experimental evidence for the Traveler's Dilemma suggests that test subjects tend to select large numbers \citep*[see, e.g.,][]{capra1999anomalous, goeree2001ten}. This points to the question of equilibrium selection: while satisficing equilibrium for $k = 2$ is consistent with empirical evidence in the Traveler's Dilemma, a theory of selection would be required to predict experimentally observed play. For completeness, \cref{fig:TD_concepts} includes the predictions of the equilibrium concepts considered in \cref{sec:otherconcepts}.

\section{Technical Appendix}\label{app:proofs}

Each section of the Technical Appendix contains the proofs and further results for its respective section in the main text.

We start by introducing several notions, which we reuse throughout the Technical Appendix. For any positive integer $n$ and any vector of positive integers $\anumvec = \V{m_1}{m_n}$, a \emph{game frame} is a tuple
\[
\gframe := \left( \;N, \, \{A_i\}_{i \in N} \;\right)
\]
consisting of a set of agents $N := \{1,\ldots,n\}$ and a set of actions $A_i:= \{1,\ldots,m_i\}$ for each agent $i \in N$. Moreover, let $\profileset:= \times_{i \in N} A_i$ and, for each $i \in N$, $\profileset_{-i} := \times_{j \in N \setminus \{i\}} A_j$. A game frame thus contains the structural elements of a game, absent preferences. For any game frame $\gframe$, $i \in N$, and $\actionvec_{-i} \in \profileset_{-i}$, define
\[
\ell_i(\actionvec_{-i}) := \{ (x,\actionvec_{-i}) \in \profileset : x \in \actionset \}
\]
as the $i$-\emph{line} with opponents' action profile $\actionvec_{-i}$. Observe that each $i \in N $ has $|\profileset_{-i}|$ $i$-lines, and each $i$-line has $m_i$ action profiles. Finally, recall that for any game $\g$ and any $\actionvec \in \profileset$, 
\[
\rank_i(\actionvec \,|\preferencerelation):= 1 + |\{x \in A_i : (x,\actionvec_{-i}) \succ_i \actionvec \}| .
\]
When the agents' preferences are clear from the context, we write $\rank_i(\actionvec)$.

\subsection{Appendix to \cref{sec:structural}}

For any positive integers $n$ and any vector of positive integers $\anumvec = (\anum_1,\dots,\anum_n)$,
\[
\rho(g) := 1- \frac{| \, \{ \actionvec \in \profileset : \actionvec   \text{ is a } \V{m_1-1}{m_n-1}\text{-satisficing equilibrium of } g\} \, |}{|\profileset|} 
\]
is the minimal precision of satisficing equilibrium in $g \in \G{n}{\anumvec}$.

\subsubsection{Proof of \cref{prop:preditiveprecision}}
For any $\G{n}{\anumvec}$ and any $g\in\G{n}{\anumvec}$, our aim is to show that 
$\rho(g) \geq \tfrac{1}{\min_i m_i}$, and to show that this bound is tight.

For any agent $i \in N$, define the set of action profiles in which $i$'s action 
is a worst response to the actions of the other agents: 
$W_i := \{ \actionvec \in \profileset : \rank_i(\actionvec) = m_i\}$. 
Since $i$ has a worst action in response to the actions of all other agents, we 
have $|W_i| = \tfrac{1}{m_i} |\profileset|$. The set of all action profiles at 
which some agent's action is a worst response to the actions of the other agents, 
$W := \cup_{i \in N} W_i$, therefore satisfies
\[
|W| \geq \max_{i \in N} |W_i| = \tfrac{1}{\min_i m_i} |\profileset|.
\]
Observing that $\rho(g) = \frac{|W|}{|\profileset|}$ completes the proof of the 
lower bound.

To show that the bound is tight, consider any game frame $\gframe$. We next 
construct a collection of preference relations $\{\succsim_i\}_{i\in N}$ with 
$\succsim_i \subseteq \profileset \times \profileset$ for each $i \in N$ such 
that, for the resulting game $\g$, we have $|W| = \tfrac{1}{\min_i m_i} 
|\profileset|$.

Take $i^\star \in \arg \min_i m_i$, fix any action $a_{i^\star}^0 \in 
A_{i^\star}$, and define
\[
S := \bigl\{ (a_{i^\star}^0,\, \actionvec_{-i^\star}) 
     : \actionvec_{-i^\star} \in \profileset_{-i^\star} \bigr\}.
\]
By construction, $S$ contains exactly one profile from each $i^\star$-line, so 
$|S| = \tfrac{1}{m_{i^\star}}|\profileset| = \tfrac{1}{\min_i m_i}|\profileset|$. 
Moreover, $S$ intersects every $j$-line for every $j \neq i^\star$: given any 
$j \neq i^\star$ and any $\actionvec_{-j} \in \profileset_{-j}$, the profile 
$(a_{i^\star}^0, a_j, \actionvec_{-\{j,i^\star\}})$ lies in $S$ for any choice 
of $a_j \in A_j$, and agrees with $\actionvec_{-j}$ on all coordinates other 
than $j$.

Construct $i^\star$'s preferences so that $W_{i^\star} = S$, giving 
$|W_{i^\star}| = \tfrac{1}{\min_i m_i}|\profileset|$. It now remains to 
construct sets $W_j$ for $j \neq i^\star$ such that $W_j \subseteq W_{i^\star}$. 
Since $S$ intersects every $j$-line, we can construct $W_j$ by selecting a 
single profile from $S \cap \ell_j(\actionvec_{-j})$ for each $\actionvec_{-j}$ 
(and construct $j$'s preferences so that the selected profile is the worst one 
for $j$ among those in $\ell_j(\actionvec_{-j})$). This ensures $W_j \subseteq 
W_{i^\star} = S$, and therefore $|W| = |W_{i^\star}| = \tfrac{1}{\min_i 
m_i}|\profileset|$.\hfill\qed

\subsubsection{Proof of \cref{thm:predictiveprecision}}
We use the probabilistic method for this proof.\footnote{\label{ftn:random}Drawing games at random has a long history in game theory; recent examples include the prevalence of dominance-solvability \citep*{alon2021dominance} and the convergence of best-response dynamics \citep*{amiet2021pure,Jon23}.} Let $G$ denote a game drawn uniformly at random from $\G{n}{\anumvec}$.

The probability that, in $G$, an action profile $\actionvec \in \profileset$ is a $\mathbf{k}$-satisficing equilibrium is given by $\prod_{i=1}^n \tfrac{k_i}{m_i}$. It follows that the expected number of $\mathbf{k}$-satisficing equilibria is
\[
|\profileset| \prod_{i=1}^n \tfrac{k_i}{m_i} = \prod_{i =1}^n  m_i \cdot \prod_{i=1}^n \tfrac{k_i}{m_i} = \prod_{i =1 }^n k_i .
\]
Therefore, setting $\mathbf{k} = \V{m_1-1}{m_n-1}$, we obtain 
\[
\mathbb{E}[1 - \rho(G)] = \prod_{i \in N} \tfrac{m_i-1}{m_i} = \prod_{i=1}^n (1 - \tfrac{1}{m_i}) .
\]
Applying Markov's inequality to the random variable $1-\rho(G)$, we have that for any $\epsilon > 0$,
\[
\Pr[\, 1- \rho(G) \geq \epsilon \,] \leq \frac{\mathbb{E}[1-\rho(G)]}{\epsilon} = \frac1\epsilon \prod_{i=1}^n (1 - \tfrac{1}{m_i}) .
\]
It follows that there exist positive constants $c$ and $\bar{n}$ such that for all $n \geq \bar{n}$,
\[
\Pr[\, \rho(G) \geq 1- \epsilon \,] \geq  1-\frac1\epsilon \prod_{i=1}^n (1 - \tfrac{1}{m_i}) \geq 1-e^{-cn} .
\]
Finally, since $G$ is drawn uniformly, observe that for any $\epsilon > 0$,
\[
\frac{| \, \{g \in \G{n}{\anumvec} : \rho(g) \geq 1-\epsilon\} \, |  }{| \, \G{n}{\anumvec} \, |} = \Pr[\, \rho(G) \geq 1-\epsilon \,] ,
\]
which completes the proof.\hfill \qed

\subsubsection{Proof of \cref{prop:nonexistence}}

We first state the following lemma; its proof is provided at the end of the section.

\begin{lemma}\label{lem:hall}
    For each $i\in N$ and $\actionvec_{-i} \in \profileset_{-i}$ there exists $\beta_i(\actionvec_{-i}) \subseteq \ell_i(\actionvec_{-i})$ with $|\beta_i(\actionvec_{-i})| \leq c_i:=\lceil \tfrac{m_i}{n} \rceil$ such that for each $\actionvec \in \profileset$ there is some $i \in N$ for which $\actionvec \in \beta_i(\actionvec_{-i})$.
\end{lemma}

\begin{proof}[Proof of \cref{prop:nonexistence}.]
For any $\gframe$ we construct a collection of preference relations $\{\succsim_i\}_{i\in N}$ with $\succsim_i \subseteq \profileset \times \profileset$ for each $i \in N$ such that the resulting game $\g$ does not have a $\mathbf k$-satisficing equilibrium for $\mathbf k\leq \V{m_1 - \lceil \tfrac{m_1}{n} \rceil}{m_n - \lceil \tfrac{m_n}{n} \rceil}$.

For each $i\in N$ and $\actionvec_{-i} \in \profileset_{-i}$ fix a set $\beta_i(\actionvec_{-i})$ as shown to exist in \cref{lem:hall}. We construct the preferences such that the actions in this set are agent $i$'s ``bad'' actions in response to $\actionvec_{-i}$. To this end, designate the action profiles in $\beta_i(\actionvec_{-i})$ to have the $|\beta_i(\actionvec_{-i})|$ highest ranks among the profiles in $\ell_i(\actionvec_{-i})$, and rank all profiles in $\ell_i(\actionvec_{-i})$ that are not in $\beta_i(\actionvec_{-i})$ as having lower ranks than any of those in $\beta_i(\actionvec_{-i})$. Any comparisons between action profiles that are in different $i$-lines can be fixed arbitrarily since only unilateral deviations matter for assessing whether a profile is a satisficing equilibrium.

A game $\g$ in which preferences are constructed as described above does not have a $\mathbf k$-satisficing equilibrium for $\mathbf k\leq \V{m_1 - \lceil \tfrac{m_1}{n} \rceil}{m_n - \lceil \tfrac{m_n}{n} \rceil}$: by construction, at every action profile $\actionvec$, there is some agent $i \in N$ for whom $\actionvec \in \beta_i(\actionvec_{-i})$ and therefore $\actionvec$ is among the $c_i$ worst responses to $\actionvec_{-i}$ for $i$.\end{proof}

\begin{proof}[Proof of \cref{lem:hall}]
    This is an application of \citeauthor{Hall35}'s Theorem (\citeyear{Hall35}). For each $i \in N$ and $\actionvec_{-i} \in \profileset_{-i}$, create $c_i$ clones of the line $\ell_i(\actionvec_{-i})$, denoted $\ell_i^j(\actionvec_{-i})$ for $j \in \{1,\ldots,c_i\}$. Define a bipartite graph $G = (L,R,E)$ where the ``left'' vertex set is the set of action profiles, $L= \profileset$, and the ``right'' vertex set $R$ consists of every clone for each line $\ell_i(\actionvec_{-i})$ for each $i \in N$ and $\actionvec_{-i} \in \profileset_{-i}$.\footnote{Observe therefore that $|L| = \prod_{i=1}^n m_i$ and $|R| = \sum_{i=1}^n c_i \prod_{j \neq i } m_j$.} The edge set is constructed as follows: connect each $\actionvec \in L$ to $\ell_i^j(\actionvec_{-i})$ for each $i \in N$ and $j \in \{1,\ldots,c_i\}$. Establishing the existence of a matching that covers $L$ is sufficient for our result because then we can set $\beta_i(\actionvec_{-i}) = \left\{ \actionvec \in \ell_i(\actionvec_{-i}) : \actionvec \text{ is matched to some clone of }\ell_i(\actionvec_{-i}) \right\}$. By \citeauthor{Hall35}'s Theorem (\citeyear{Hall35}), there is such a matching if and only if for every $U \subseteq L$, $|U| \leq |N_G(U)|$. 
    
    We now show that this condition holds in the graph $G$. For $U \subseteq L$, let $\Pi_i(U) :=\{\actionvec_{-i} \in \profileset_{-i} : \text{ there is some } \mathbf b \in U \text{ with } \mathbf{b}_{-i} = \actionvec_{-i}\}$, which extracts the set of distinct action profiles of agents other than $i$ from the action profiles in $U$. Then the neighborhood of $U$ is the disjoint union $N_G(U)=\bigsqcup_{i=1}^n \bigsqcup_{\actionvec_{-i}\in \Pi_i(U)} \{\ell_i^{j}(\actionvec_{-i}): j \in \{1,\dots,c_i\}\}$ and hence $| N_G(U) | = \sum_{i=1}^n c_i |\Pi_i(U)|$. Next, fix some $i \in N$. Since every $\actionvec \in U$ lies in exactly one $i$-line, we have that $ |U| = \sum_{\actionvec_{-i} \in \Pi_i(U)} | U \cap \ell_i(\actionvec_{-i})| $. But since $|\ell_i(\cdot)| = m_i$, we have that $\sum_{\actionvec_{-i} \in \Pi_i(U)} | U \cap \ell_i(\actionvec_{-i})| \leq m_i |\Pi_i(U)|$. Therefore
    \[
    |U| = \frac{1}{n} \sum_{i=1}^n \sum_{\actionvec_{-i} \in \Pi_i(U)} | U \cap \ell_i(\actionvec_{-i})|\leq \sum_{i=1}^n \tfrac{m_i}{n} |\Pi_i(U)| \leq \sum_{i=1}^n \lceil \tfrac{m_i}{n} \rceil |\Pi_i(U)| = | N_G(U) | ,
    \]
    as required.
\end{proof} 

\subsubsection{Tightness of \cref{prop:nonexistence}}\label{sec:nonexistencetight}

\begin{proposition}
For any set of games $\G{n}{\anumvec}$, every game $g\in \G{n}{\anumvec}$ has a $\mathbf k$-satisficing equilibrium for $\mathbf k\leq\V{m_1-\lceil \tfrac{m_1}{n}\rceil+1}{m_n-\lceil \tfrac{m_n}{n}\rceil+1}$.
\end{proposition}

\begin{proof}
Consider any game $g\in\G{n}{\anumvec}$ and, for any $i \in N$, define
\[
S_i := \{\actionvec \in \profileset : \rank_i(\actionvec) > m_i - \lceil \tfrac{m_i}{n} \rceil + 1\}.
\]
Observe that $|S_i| = (\lceil \tfrac{m_i}{n} \rceil - 1) \tfrac{1}{m_i} \prod_{j=1}^n m_j$. Since $(\lceil \tfrac{m_i}{n} \rceil-1)< \tfrac{m_i}{n}$ we have that $\sum_{i=1}^n |S_i| < \tfrac{1}{n}\sum_{i=1}^n \prod_{j=1}^n m_j = \prod_{j=1}^n m_j = |\profileset|$. Therefore, the union $\cup_{i=1 }^n S_i$ is a strict subset of $\profileset$. It follows that there exists an action profile $\actionvec^\star$ that satisfies $\rank_i(\actionvec^\star)\leq m_i-\lceil \frac{m_i}{n}\rceil+1$ for all $i$, i.e., it is a $\mathbf k$-satisficing equilibrium for $\mathbf k \leq \V{m_1-\lceil \frac{m_1}{n}\rceil+1}{m_n-\lceil \frac{m_n}{n}\rceil+1}$.
\end{proof}

\subsubsection{Proof of \cref{thm:one_is_enough}}\label{appendixB}

We start by presenting \cref{lem:general} below, from which we will derive \cref{thm:one_is_enough}. 

For any positive integer $n$, any vector of positive integers $\anumvec$, any $g \in \G{n}{\anumvec}$, any finite set $S \subseteq \{1,\ldots,n\}$, any integer $0 \leq d \leq |S|$, and any integer $k \geq 2$, let
\[
q(d,S,k,g) := \bigg| \left\{ \actionvec \in \profileset : \;
  \vcenter{\hbox{\text{\shortstack[l]{$d$ agents from $S$ are $k$-satisficed and\\
    all other agents are $1$-satisficed at $\actionvec$}}}}
    \; \right\} \bigg|,
\]
and for any positive integer $z$, let
\[
\gamma_{n,\anumvec}(d,S,k,z) := \frac{| \, \{g \in \G{n}{\anumvec} : q(d,S,k,g) \geq z \} \, |  }{| \, \G{n}{\anumvec} \, |}  .
\]
Therefore $\gamma_{n,\anumvec}(d,S,k,z)$ is the fraction of games in $\G{n}{\anumvec}$ that have at least $z$ action profiles in which $d$ agents from $S$ are $k$-satisficed and all other agents are $1$-satisficed. Finally, for any $z \geq 1$, we write $\Pr[\Poisson(\lambda) \geq z]$ for the probability that a $\Poisson(\lambda)$ random variable is no less than $z$. 

\begin{theorem}\label{lem:general}
For any $S \subseteq \{1,\dots,n\}$, $0 \leq d \leq |S|$, and integers $k\geq 2$ and $z \geq 1$, there exist positive constants $c$ and $\bar{n}$ such that for any set of games $\G{n}{\anumvec}$ with $n \geq \bar{n}$ and any $\anumvec \in \{2,3,\ldots\}^n$,
\[
\Big|\,  \gamma_{n,\anumvec}(d,S,k,z) - 
\Pr[\Poisson(\psi(d,S,k))\geq z]  \,\Big| \leq \frac{5nk^2 \max_i m_i}{2^{n}}\psi(d,S,k)^2,
\]
where
\[
\psi(d,S,k) :=  \sum_{T \subseteq S, 0\leq |T|\leq d} \, \prod_{i \in T} (\min\{k,\anum_i\}-1).
\]
\end{theorem}
If the upper bound in \cref{lem:general} is small, then $\gamma_{n,\anumvec}(d,S,k,z)$ can be approximated by the counter-cumulative distribution function $\Pr[X \geq z]$ of a Poisson random variable $X \sim \Poisson(\psi(d,S,k))$. The proof of \cref{lem:general} is in \cref{sec:bigproof}.
\begin{proof}[Proof of \cref{thm:one_is_enough}]
Consider \Cref{lem:general} for the case where $S=\{1,\dots,n\}$, $d=1$, and $k=2$. We have that $\psi(1,\{1,\dots,n\},2)=n+1$, and since $\Pr[\Poisson(n+1)\geq 1] = 1 - e^{-(n+1)}$, \cref{lem:general} gives us that 
\[
\Big|\,  \gamma_{n,\anumvec}(1,\{1,\ldots,n\},2,1) - 
\left(1 - e^{-(n+1)} \right)  \,\Big| \leq \frac{20n(n+1)^2 \max_i m_i}{2^{n}}.
\]
We can therefore find positive constants $c$ and $\bar{n}$ such that for all $n \geq \bar{n}$,
\[
 \gamma_{n,\anumvec}(1,\{1,\ldots,n\},2,1) \geq 1 - e^{-cn},
\]
which completes the proof. 
\end{proof}
In addition to \cref{thm:one_is_enough}, \cref{lem:general} also allows us to derive the corollary below regarding the distribution of satisficing equilibria, which is of independent interest. The result is a straightforward application of \cref{lem:general} so we state it without proof.
\begin{corollary}\label{cor:cor1}
 Consider a fixed set $S\subset \{1,\dots,n\}$ that does not grow with $n$. There exist  positive constants $c$ and $\bar{n}$ such that for all $n \geq \bar{n}$, any $\anumvec \in \{2,3\ldots,\}^n$, and any $z \geq 1$,
\[
\Bigg|\, \gamma_{n,\anumvec}(|S|,S,k,z) - \Pr\biggl[\Poisson \Bigl( \prod_{i \in S} \min\{k,m_i\} \Bigr)\geq z\biggr] \, \Bigg| \leq  e^{-cn}.
\]   
\end{corollary}
Observe that setting $S = \varnothing$ above implies that the fraction of games in $\G{n}{\anumvec}$ that have at least $z\geq 1$ pure Nash equilibria is asymptotically $\Pr[\Poisson(1)\geq z]$ when $n$ is large.\footnote{We adopt the standard convention that a product over an empty set is 1, i.e. $\prod_{i \in \varnothing} c_i = 1$ whatever $c_i$ may be.} This therefore gives the distribution of pure Nash equilibria when the number of agents is large, which was previously derived in \citet[Proposition 2.4]{Rin00}. 

Suppose, instead, that we set $k=2$ and $S = \{i\}$ for some fixed integer $i \in \{1,\ldots,n\}$. Then \cref{cor:cor1} implies that the fraction of games that admit an action profile in which specifically the fixed agent $i$ is $2$-satisficed is asymptotically $1 - 1/e^2 \approx 86\%$. This is in contrast to \cref{thm:one_is_enough}, where the agent that is $2$-satisficed is not fixed. 

\subsubsection{Proof of \cref{lem:general}}\label{sec:bigproof}

The proof employs the probabilistic method. Draw a game $G$ uniformly at random from $\G{n}{\anumvec}$. For any such game and for any $S \subseteq \{1,\dots,n\}$, and $0 \leq d \leq |S|$, and integer $k\geq 2$, define $\indicator{\actionvec}$ to be an indicator variable that takes the value 1 if, in $G$, $d$ agents from $S$ are $k$-satisficed at $\actionvec$ and all other agents are $1$-satisficed at $\actionvec$. It takes the value 0 otherwise. Moreover, let $N(d,S,k) := \sum_{\actionvec \in \profileset} \indicator{\actionvec}$ denote the number of action profiles in $G$ in which $d$ from $S$ agents are $k$-satisficed and all the other agents are $1$-satisficed. Our proof relies on the observation that
\[
\gamma_{n,\anumvec}(d,S,k,z) = \Pr[N(d,S,k) \geq z].
\]
For any $\actionvec \in \profileset$, we write the set of all lines going through $\actionvec$ as $\ell(\actionvec) := \cup_{i =1}^n \ell_i(\actionvec_{-i})$. Observe that $|\ell(\actionvec)|=1 + \sum_{i =1}^n (\anum_i - 1)$ and that, in our case, $\indicator{\actionvec} \independent \indicator{\actionvec'}$ if and only if $\actionvec' \not\in \ell(\actionvec)$. The Chen-Stein theorem then gives us that for any $S \subseteq \{1,\dots,n\}$, $0 \leq d \leq |S|$, and integer $k \geq 2$,
\begin{align*}
&\; \sup_{z \in \{1,2,\ldots\}}\big|\, \Pr\left[ N(d,S,k) \geq z \right] 
 - \Pr[\Poisson(\psi(d,S,k))\geq z] \,\big| \\
=&\; \sup_{z \in \{1,2,\ldots\}}\big|\, \gamma_{n,\anumvec}(d,S,k,z) - \Pr[\Poisson(\psi(d,S,k))\geq z] \,\big| \\
 \leq&\; c_1(d,S,k) + c_2(d,S,k),    
\end{align*}
where 
\begin{equation*}
c_1(d,S,k) := \sum_{\actionvec \in \profileset} \sum_{\actionvec' \in \ell(\actionvec)} \mathbb{E}[\indicator{\actionvec}] \, \mathbb{E}[\indicator{\actionvec'}], \text{ and }c_2(d,S,k) := \sum_{\actionvec \in \profileset} \sum_{\actionvec' \in \ell(\actionvec) \setminus \{\actionvec\}} \mathbb{E}[\indicator{\actionvec} \, \indicator{\actionvec'}] .
\end{equation*}
The remainder of the proof consists of deriving an explicit upper bound for $c_1(d,S,k) + c_2(d,S,k)$, which we do in several steps (\Cref{lem:expectation} to \Cref{lem:c2}).

\begin{remark}[The Chen-Stein method in games]
The Chen-Stein theorem provides a method for approximating the distribution of the sum of Bernoulli random variables with a Poisson distribution, provided the Bernoulli variables have limited dependence. See \citet*{arratia1989two} for a discussion of the Chen-Stein method and a proof. The method has previously been employed in the study of games by, for example, \cite*{powers1990limiting} who studies the distribution of pure Nash equilibria in random games in the presence of ties; \cite*{Rin00} who studies the distribution of pure Nash equilibria in random games with correlated payoffs; and \cite*{stanford2004individually} who proves the non-emptiness of the set of feasible and individually rational payoffs in games with many actions per agent. In fact, to prove the latter, \cite*{stanford2004individually} derives as an intermediate step the prevalence of action profiles in which each agent plays a top $k$ action in response to the actions of the other agents, i.e., $\V{k}{k}$-satisficing equilibria,  as the numbers of actions diverge. The latter can be shown to follow as a corollary of our \cref{lem:general}.
\end{remark}

In what follows for any non-empty set $U \subseteq \{1,\dots,n\}$, any (possibly empty) $S \subseteq U$, any integer $0 \leq d \leq |S|$ and integer $k \geq 2$, it will be useful to define
\[
\zeta(U,d,S,k):= \left( \prod_{i \in U} \frac{1}{\anum_i} \right) \psi(d,S,k) =\left( \prod_{i \in U} \frac{1}{\anum_i} \right) \left( \sum_{T \subseteq S, 0 \leq |T|\leq d} \, \prod_{i \in T} (\min\{k,\anum_i\}-1) \right) .
\]
\begin{remark}
The quantity $\zeta(U,d,S,k)$ is the probability, at any given action profile, that $d$ agents from $S$ are $k$-satisficed and all  agents in $U$ who are not these $d$ agents are $1$-satisficed. The reason is the following: for each subset $T \subseteq S$ such that $0 \leq |T|\leq d$, each agent $i \in T$ is $k$-satisficed (but \emph{not} $1$-satisficed) with probability $\frac{\min\{k,m_i\}-1}{m_i}$; all agents in $U \setminus T$ are $1$-satisficed with probability $\prod_{i \in U \setminus T} \frac{1}{m_i}$. Putting this together yields the expression above.
\end{remark}

\begin{lemma}\label{lem:expectation} For any integer $S \subseteq \{1,\dots,n\}$, $0 \leq d \leq |S|$, and integer $k \geq 2$,
\[
\mathbb{E}[\indicator{\actionvec}] =\zeta(\{1,\dots,n\},d,S,k) ,
\]
and moreover,
\[
 \mathbb{E}[N(d,S,k)]=\sum_{\actionvec \in \profileset} \mathbb{E}[\indicator{\actionvec}] = \psi(d,S,k).
\]
\end{lemma}
\begin{proof}
Clearly, $\mathbb{E}[\indicator{\actionvec}] = \Pr[\indicator{\actionvec} = 1]$. But the latter is equal to $\zeta(\{1,\dots,n\},d,S,k)$ by definition of $\indicator{\actionvec}$ and $\zeta$. The expectation of $N(d,S,k)$ then follows from the fact that $|\profileset| = \prod_{i=1}^n m_i$.
\end{proof}

\begin{lemma}\label{lem:c1} For any $S \subseteq \{1,\dots,n\}$, $0 \leq d \leq |S|$, and integer $k \geq 2$,
    \[
    c_1(d,S,k) = \frac{1+\sum_{i =1}^n (\anum_i-1)}{\prod_{i=1}^n \anum_i} \cdot \psi(d,S,k)^2 \leq   \frac{n \max_{i \in \{1,\dots,n\}} \anum_i}{\prod_{i = 1}^n \anum_i} \psi(d,S,k)^2 .
    \]
\end{lemma}
\begin{proof}
\begin{align*}
c_1(d,S,k) &= \sum_{\actionvec \in \profileset} \sum_{\actionvec' \in \ell(\actionvec)} \mathbb{E}[\indicator{\actionvec}] \, \mathbb{E}[\indicator{\actionvec'}] = |\profileset| \cdot |\ell(\actionvec)| \cdot \mathbb{E}[\indicator{\actionvec}]^2 
\end{align*}
Substituting for these values (using \Cref{lem:expectation}) completes the proof.
\end{proof}

\begin{lemma}\label{lem:cov}
Consider any $S \subseteq \{1,\dots,n\}$, $0 \leq d \leq |S|$, and integer $k \geq 2$. For any distinct $\actionvec$ and $\actionvec'$ that differ only in the $i\text{th}$ index,
\[
\mathbb{E}[\indicator{\actionvec} \indicator{\actionvec'}] \leq 4 \frac{(\min\{k,m_i\}-1)^2}{\prod_{i=1}^n m_i^2} \psi(d,S,k)^2
\]
\end{lemma}
\begin{proof}
Consider distinct $\actionvec$ and $\actionvec'$ that differ only in the $i\text{th}$ index. For any $S \subseteq \{1,\dots,n\}$, write $S_{-i}:= S \setminus \{i\}$. 

If $i \not\in S$ or if $d=0$, $\indicator{\actionvec} \indicator{\actionvec'}=0$ because $i$ would need to be $1$-satisficed at both $\actionvec$ and $\actionvec'$ which is impossible. Assume therefore that $i \in S$ and $d \geq 1$. The event $\indicator{\actionvec} \indicator{\actionvec'}=1$ occurs precisely when: 
\begin{enumerate}[(i),leftmargin=0cm]
    \item $i$ plays a best response at $\actionvec$ and plays a top $k'$ action at $\actionvec'$ for $k' \in \{2,\dots,k\}$, which occurs with probability $\frac{1}{m_i} \frac{\min\{k,m_i\}-1}{m_i-1}$, \emph{and} $d$ agents from $S_{-i}$ are $k$-satisficed at $\actionvec$ and every agent in $\{1,\dots,n\}_{-i}\setminus S_{-i}$ is $1$-satisficed at $\actionvec$, which occurs with probability $\zeta(\{1,\dots,n\}_{-i},d,S_{-i},k)$, \emph{and} $d-1$ agents from $S_{-i}$ are $k$-satisficed at $\actionvec'$ and every agent in $\{1,\dots,n\}_{-i}\setminus S_{-i}$ is $1$-satisficed at $\actionvec'$, which occurs with probability $\zeta(\{1,\dots,n\}_{-i},d-1,S_{-i},k)$;
    \item or, point (i) occurs with $\actionvec$ and $\actionvec'$ swapped;
    \item or, $i$ plays a top $k'$ action at $\actionvec$ for $k' \in \{2,\dots,k\}$ and plays a top $k''$ action at $\actionvec'$ for $k'' \in \{2,\dots,k\}$, which occurs with probability $2 \frac{\min\{k,m_i\} - 1}{m_i} \frac{\min\{k,m_i\} - 2}{m_i-1}$, \emph{and} $d-1$ agents from $S_{-i}$ are $k$-satisficed at both $\actionvec$ and $\actionvec'$ and every agent in $\{1,\dots,n\}_{-i}\setminus S_{-i}$ is $1$-satisficed at both $\actionvec$ and $\actionvec'$, which occurs with probability $\zeta(\{1,\dots,n\}_{-i},d-1,S_{-i},k)^2$.
\end{enumerate}
Putting the above together, we have that
\begingroup
\allowdisplaybreaks
\begin{align*}
    \mathbb{E}[\indicator{\actionvec} \indicator{\actionvec'}] =& 2 \frac{1}{m_i} \frac{\min\{k,m_i\}-1}{m_i-1} \zeta(\{1,\dots,n\}_{-i},d,S_{-i},k) \zeta(\{1,\dots,n\}_{-i},d-1,S_{-i},k)\\
    &+ 2 \frac{\min\{k,m_i\} - 1}{m_i} \frac{\min\{k,m_i\} - 2}{m_i-1} \zeta(\{1,\dots,n\}_{-i},d-1,S_{-i},k)^2 \\
    \leq&  2 \frac{(\min\{k,m_i\}-1)^2}{m_i(m_i-1)} \zeta(\{1,\dots,n\}_{-i},d,S,k)^2 \\
    =&  2 \frac{m_i}{m_i-1}(\min\{k,m_i\}-1)^2 \left( \prod_{i=1}^n \frac{1}{m_i}\right)^2 \psi(d,S,k)^2  \\
    \leq &  4 \frac{(\min\{k,m_i\}-1)^2}{\prod_{i=1}^n m_i^2} \psi(d,S,k)^2
\end{align*}
\endgroup
This completes the proof.
\end{proof}

\begin{lemma}\label{lem:c2} For any $S \subseteq \{1,\dots,n\}$, $0 \leq d \leq |S|$, and integer $k \geq 2$,
\[
c_2(d,S,k) \leq 4 \psi(d,S,k)^2 \frac{\max_{i \in \{1,\dots,n\}}m_i}{\prod_{i=1}^n m_i} \sum_{i\in S} (\min\{k,m_i\}-1)^2 .
\] 
\end{lemma}
\begin{proof}
We rely on \Cref{lem:cov} in the derivation below.
\begin{align*}
c_2(d,S,k) &= \sum_{\actionvec \in \profileset} \sum_{\actionvec' \in \ell(\actionvec) \setminus \{\actionvec\}} \mathbb{E}[\indicator{\actionvec} \, \indicator{\actionvec'}] =\sum_{\actionvec \in \profileset} \sum_{i\in S} \sum_{\actionvec' \in \ell_i(\actionvec_{-i}) \setminus \{\actionvec\}} \mathbb{E}[\indicator{\actionvec} \, \indicator{\actionvec'}] \\
&\leq 4 \sum_{\actionvec \in \profileset} \sum_{i\in S} (m_i-1) \frac{(\min\{k,m_i\}-1)^2}{\prod_{i=1}^n m_i^2} \psi(d,S,k)^2  \\
&\leq 4 \psi(d,S,k)^2 \frac{\max_{i \in \{1,\dots,n\}}m_i}{\prod_{i=1}^n m_i} \sum_{i\in S} (\min\{k,m_i\}-1)^2.
\end{align*}
This completes the proof.
\end{proof}

Finally, we can collect the bounds and complete the proof of \cref{lem:general}. From \Cref{lem:c1} and \Cref{lem:c2} we have that
    \[
    c_1(d,S,k) + c_2(d,S,k) \leq \left(n + 4 \sum_{i\in S} (\min\{k,m_i\}-1)^2 \right) \frac{ \max_{i \in \{1,\dots,n\}} \anum_i }{ \prod_{i =1}^n \anum_i} \psi(d,S,k)^2 .
    \]
But $\prod_{i=1}^n m_i \geq 2^n$ and $n + 4 \sum_{i\in S} (\min\{k,m_i\}-1)^2 \leq n + 4nk^2 \leq 5nk^2$. This completes the proof. \hfill \qed

\subsubsection{Proof of \cref{thm:approximatepotential}}

Suppose that $\g$ is a $\mathbf k$-approximate ordinal potential game. There is a weak order $\potpref$ on $\profileset$ such that for any $i \in \agentset$ and $\actionvec \in \profileset$,
\begin{equation*}
     \rank_i(\actionvec \,|\succsim_i)-\rank_i(\actionvec \,|\,\potpref) \,  < k_i. 
\end{equation*}
Now consider a game $g_{n,\anumvec,\potpref}$ in which each agents' preferences are $\potpref$. Since all agents have the same preferences, $g_{n,\anumvec,\potpref}$ is an ordinal potential game and it possesses a pure Nash equilibrium \citep[Corollary 2]{monderer1996potential}. Call such an equilibrium $\actionvec^\star$ and observe that $\rank_i(\actionvec^\star \,|\,\potpref)=1$ for all $i\in N$. Hence, evaluating the inequality above at $\actionvec^\star$ gives us that $\rank_i(\actionvec^\star \,|\succsim_i)< k_i + 1$, which  is the condition for $\actionvec^\star$ to be a $\mathbf k$-satisficing equilibrium of $\g$.\hfill \qed

\subsection{Appendix to \cref{sec:emergence}}

\subsubsection{Proof of \cref{prop:dynamic}}
Our proof relies on the probabilistic method. We draw a game $\Ggame$ uniformly at random from $\G{n}{\anumvec}$ and consider a $\mathbf k$-satisficing equilibrium in $G$. Next, we construct a coupled game $G^\star$ for which the same action profile is a pure Nash equilibrium. We then show that, from any initial profile, there is a sequence of top-$k_i$ response moves to the $\mathbf k$-satisficing equilibrium in the $G$ if there exists a best-response path to the constructed pure Nash equilibrium in the coupled game $G^\star$. \citet*{Jon23} show that existence of such a path in the coupled game holds with high probability provided the number of agents is sufficiently large.

Let $A_{n,\anumvec}(\mathbf k)$ denote the event that $\Ggame \in \G{n}{\anumvec}$ has a ${\mathbf k}$-satisficing equilibrium and let $B_{n,\anumvec}(\mathbf k)$ denote the event that, for any $\actionvec(0)$, the top-${\mathbf k}$ response dynamic converges almost surely to a ${\mathbf k}$-satisficing equilibrium of $\Ggame$. Fixing some integer $M \geq 2$, our aim is to show that there are constants $c$ and $\bar{n}$ such that for any $n \geq \bar{n}$, $\anumvec \in \{2,\ldots,M\}^n$, and ${\mathbf k} \in \{2,3,\ldots\}^n$,
\[
\Pr[ A_{n,\anumvec}(\mathbf k) \text{ and } B_{n,\anumvec}(\mathbf k) ] \geq 1 - e^{-cn} .
\]
Our argument below relies on a coupling of $\Ggame$ with another game that possesses a pure Nash equilibrium if and only if $\Ggame$ has a $\V{2}{2}$-satisficing equilibrium. From $\Ggame$, construct a game $\Ggame^\star \in \G{n}{\anumvec}$ as follows:
\begin{itemize}[leftmargin=0cm]
    \item If $\Ggame$ has a $\V{2}{2}$-satisficing equilibrium, select a profile  uniformly at random from among those profiles that are $\V{2}{2}$-satisficing equilibria in $\Ggame$. Call the chosen profile $\actionvec^\star$; this will be the pivotal profile from which $\Ggame^\star$ is constructed. We can partition the set of agents into those who are playing a first-best response at $\actionvec^\star$, call this set $I_1$, and those who are playing a second-best response at $\actionvec^\star$, call this set $I_2$. Set the game $\Ggame^\star$ to be identical to $\Ggame$ in every respect \emph{except} that the preference ordering of each $i \in I_2$ over $\ell_i(\actionvec^\star_{-i})$ is altered as follows: shift $i$'s first-best action at $\actionvec^\star$ in the game $\Ggame$, uniformly at random, to any position in the preference ordering over $\ell_i(\actionvec^\star_{-i})$ other than the first. Observe that $\actionvec^\star$ is a Nash equilibrium in $\Ggame^\star$.
    \item If $\Ggame$ does not have a $\V{2}{2}$-satisficing equilibrium, set $\Ggame^\star = \Ggame$.
\end{itemize}
\begin{remark}\label{rem:AtoD}
By construction, $\Ggame$ has a $\V{2}{2}$-satisficing equilibrium if and only if $\Ggame^\star$ has a pure Nash equilibrium. 
\end{remark}
For $1 \leq T < \infty$, a sequence of distinct action profiles $\actionvec(0),\dots,\actionvec(T)$ is a (finite) \emph{best-response path} if for each $t \in \{0,\dots,T-1\}$, there is an $i \in \agentset$ such that $\actionvec_{-i}(t) = \actionvec_{-i}(t+1)$ and $\action_i(t+1)$ is a (1st)-best-response to $\actionvec_{-i}(t+1)$. Using the terminology of \citet*{Jon23} we say that a game is \emph{connected} if at least one action profile is a pure Nash equilibrium and, from any non-Nash profile there is a finite best-response path to every Nash equilibrium. Let $C_{n,\anumvec}$ denote the event that $\Ggame^\star$ is connected. \cref{lem:BtoC} below shows that if $\Ggame$ has a $\V{2}{2}$-satisficing equilibrium then $\Ggame^\star$ being connected implies that, for any $\actionvec(0)$, the top-$\mathbf k$ response dynamic converges almost surely to a $\mathbf k$-satisficing equilibrium of $\Ggame$. That is,
\begin{lemma}\label{lem:BtoC}
Conditional on $A_{n,\anumvec}\V{2}{2}$, the event $C_{n,\anumvec}$ implies the event $B_{n,\anumvec}(\mathbf k)$.
\end{lemma}
Let $D_{n,\anumvec}$ denote the events that $\Ggame^\star$ has a pure Nash equilibrium, and let $A_{n,\anumvec}^\star$ denote the event that $G$ has an action profile at which one agent is 2-satisficed and all other agents are 1-satisficed. \cref{lem:BtoC} allows us to make the following argument:
\begin{align*}
    \Pr[ A_{n,\anumvec}(\mathbf k) \text{ and } B_{n,\anumvec}(\mathbf k) ] \geq & \Pr[ A_{n,\anumvec}\V{2}{2} \text{ and }  B_{n,\anumvec}(\mathbf k) ] \\
    =& \Pr[  B_{n,\anumvec}(\mathbf k) \,|\, A_{n,\anumvec}\V{2}{2} ] \cdot \Pr[A_{n,\anumvec}\V{2}{2}] \\
    \geq& \Pr[C_{n,\anumvec} \,|\, A_{n,\anumvec}\V{2}{2}] \cdot \Pr[A_{n,\anumvec}\V{2}{2}] \\
    =& \Pr[C_{n,\anumvec} \,|\, D_{n,\anumvec}] \cdot \Pr[A_{n,\anumvec}\V{2}{2}] \\
    \geq & \Pr[C_{n,\anumvec} \,|\, D_{n,\anumvec}] \cdot \Pr[A_{n,\anumvec}^\star]
\end{align*}
The first inequality follows because $\Ggame$ having a $\V{2}{2}$-satisficing equilibrium implies that it has a ${\mathbf k}$-satisficing equilibrium for any ${\mathbf k} \geq \V{2}{2}$. The second inequality follows from \cref{lem:BtoC}. The penultimate step follows from \cref{rem:AtoD}. The final inequality follows from the fact that $A_{n,\anumvec}^\star$ implies $A_{n,\anumvec}\V{2}{2}$.

 \citet*{Jon23} show that there is a positive constant $c_1$ such that for $n$ sufficiently large relative to $\max_i m_i \leq M$, $\Pr[C_{n,\anumvec} \,|\, D_{n,\anumvec}] \geq 1-e^{-c_1n}$. Moreover, \cref{thm:one_is_enough} shows that there exist positive constants $c_0$ and $\bar{n}$ such that for all $n \geq \bar{n}_0$, $\Pr[A_{n,\anumvec}^\star] \geq 1-e^{-c_0 n}$. It follows that there exist positive constants $c$ and $\bar{n}$ such that for all $n\geq \bar{n}$, $\Pr[C_{n,\anumvec} \,|\, D_{n,\anumvec}] \cdot \Pr[A_{n,\anumvec}^\star] \geq 1- e^{-cn}$.\hfill\qed

It remains for us to provide the proof of \cref{lem:BtoC}.

\begin{proof}[Proof of \cref{lem:BtoC}]
 Suppose $\Ggame$ has at least one $\V{2}{2}$-satisficing equilibrium. Choose one such equilibrium, call it $\actionvec^\star$ and use it as the pivotal profile from which $\Ggame^\star$ is constructed.

Now assume that $\Ggame^\star$ is connected. We must show that for any $\actionvec(0)$, the top-$\mathbf k$ response dynamic converges almost surely to a $\mathbf k$-satisficing equilibrium of $\Ggame$. To prove the latter, it suffices for us to show the following:
\begin{enumerate}[(a),leftmargin=0cm]
\item At any action profile that is not a $\mathbf k$-satisficing equilibrium of $\Ggame$, the top-${\mathbf k}$ response dynamic has a positive probability of leaving that profile. 
\item From any action profile that is not a ${\mathbf k}$-satisficing equilibrium of $\Ggame$ there is, with positive probability, a sequence of top-$k_i$ response moves that leads to a ${\mathbf k}$-satisficing equilibrium of $\Ggame$.
\item At any action profile that is a ${\mathbf k}$-satisficing equilibrium of $\Ggame$, with positive probability, the top-${\mathbf k}$ response dynamic stays at that equilibrium forever.
\end{enumerate}
Point (a) is straightforward: at any action profile that is not a ${\mathbf k}$-satisficing equilibrium, there is a positive probability that an agent who is not satisficed at that profile becomes active and therefore changes the profile.

Establishing point (b) requires thinking about our coupled game $\Ggame^\star$. Consider an initial profile $\actionvec(0)$ that is not a ${\mathbf k}$-satisficing equilibrium of $\Ggame$. Then, there is a positive probability that exactly one of the agents who is not $k_{i_0}$-satisficed at that profile, say agent $i_0$, to become active and move the dynamic to some new action profile $\actionvec(1)$. If this new profile is a ${\mathbf k}$-satisficing equilibrium of $\Ggame$, we are done.
\begin{center}
\begin{tikzpicture}
    \node at (0, 0) (a0) {$\mathbf{a}(0)$};
    \node at (2, 0) (a1) {$\mathbf{a}(1)$};
    \node at (4, 0) (a2) {$\mathbf{a}(2)$};
    \node at (6, 0) (a3) {$\cdots$};
    \node at (8, 0) (a4) {$\mathbf{a}(T)$};
    \node at (10.75, 0) (a5) {$\mathbf{a}(T+1)=\mathbf{a}^\star$};
    \draw[->]  (a0) -- (a1) node[midway, above] {$i_0$};
    \draw[->]  (a1) -- (a2) node[midway, above] {$i_1$};
    \draw[->]  (a2) -- (a3) node[midway, above] {$i_2$};
    \draw[->]  (a3) -- (a4) node[midway, above] {$i_{T-1}$};
    \draw[->]  (a4) -- (a5) node[midway, above] {$i_{T}$};
\end{tikzpicture}
\end{center}
Suppose therefore that $\actionvec(1)$ is not a $\mathbf k$-satisficing equilibrium of $\Ggame$. Since the coupled game $\Ggame^\star$ is connected, there is a (finite) best-response path $\actionvec(1),\dots,\actionvec(T+1) = \actionvec^\star$ ending with $\actionvec^\star$ which, by construction, is a pure Nash equilibrium in $\Ggame^\star$. The path is illustrated above. Observe that the subsequence of action profiles $\actionvec(1),\dots,\actionvec(T)$ is a best-response path in $\Ggame$. There is a positive probability that only agent $i_1$ becomes active, and then only agent $i_2$ becomes active, and so on. In every case, whenever these agents become active, the sequence of action profiles since they were last active must contain at least two distinct profiles. So each of these agents chooses one of her top $k_i$ actions in response to the others' actions with positive probability. There is therefore a positive probability of reaching $\actionvec(T)$. Finally, observe that $\action_{i_{T}}(T+1)$ is at least a $k_i$-th best response to $\actionvec_{-i_{T}}(T)$ for agent $i_{T}$ in $\Ggame$. It follows that there is also a positive probability of reaching $\actionvec^\star$, as required.

Now, for point (c), suppose that the dynamic is at a $\mathbf k$-satisficing equilibrium $\actionvec^\star$ of $\Ggame$. There is a positive probability that only agent 1 becomes active and re-selects her current action, and then only agent 2 becomes active and re-selects her current action, and so on until agent $n$. This results in a situation in which each agent is $k_i$-satisficed and the sequence of action profiles since they were last active is $\actionvec^\star,\dots,\actionvec^\star$. The top-$\mathbf k$ response dynamic therefore remains at $\actionvec^\star$ forever.
\end{proof}

\subsection{Appendix to \cref{sec:characterization}}\label{app:characterization}

\subsubsection{Proof of \cref{thm:testing}}\label{app:algoproof}

Consider a dataset $\{(\gf^\z,\mathbf{b}^\z)\}_{\z=1}^\Z$ and fix some $\mathbf{k}=(k_i)_{i\in N}$. 

For each $i \in N$, define $\mathbf{B}_{-i} = \{\mathbf{b}_{-i} : \mathbf{b}_{-i} = \mathbf{b}_{-i}^\z \text{ for some }\z \in \{1,\dots,\Z\}\}$ to be the set of all distinct environments faced by agent $i$ in the dataset. For each $\mathbf{b}_{-i} \in \mathbf{B}_{-i}$, the indices of the observations at which $\mathbf{b}_{-i}$ was encountered by agent $i$ is given by $\Z_i(\mathbf{b}_{-i}) = \{\z : \mathbf{b}_{-i}^\z =\mathbf{b}_{-i}\}$. We say that agent $i$'s choices $b_i^1,\dots,b_i^\Z$ are $k_i$-rationalizable if, for every $\mathbf{b}_{-i} \in \mathbf{B}_{-i}$, there exists a preference $\succsim_i^{\mathbf{b}_{-i}}$ over $A_i$ such that for each $\z \in \Z_i(\mathbf{b}_{-i})$, $b_i^\z$ is among the top $k_i$ alternatives of $B_i^\z$ according to $\succsim_i^{\mathbf{b}_{-i}}$.

The key observation is that the original dataset $\{(\gf^\z,\mathbf{b}^\z)\}_{\z=1}^\Z$ is $\mathbf{k}$-rationalizable if and only if each agent $i$'s choices are $k_i$-rationalizable. Hence, it suffices to test $k_i$-rationalizability separately for each $i \in N $. We do this in \cref{alg:main} where, for each $i \in N$ and $\mathbf{b}_{-i} \in \mathbf{B}_{-i}$, we employ the algorithm proposed by \citet*{Bar22} to carry out the decision-theoretic test of $k_i$-rationalizability of an agent $i$'s choices when faced with environment $\mathbf{b}_{-i}$. 

\vspace{0.2cm}

\SetKw{Let}{Let}        
\DontPrintSemicolon     
\begin{algorithm}[H]
\KwIn{data set: $\{(\gf^\z,\mathbf{b}^\z)\}_{\z=1}^\Z$; parameters: $\mathbf k=(k_i)_{i\in\agentset}$}
\KwOut{\emph{True} if the data set is $\mathbf k$-rationalizable; \emph{False} otherwise}

\For{$i \in N$}{
  \For{$\mathbf b_{-i}\in\mathbf B_{-i}$}{
    \Let{$X_0 \gets \bigcup_{\z \in \Z_i(\mathbf b_{-i})} B_i^\z$}\;
    \tcp{$X_0$: actions of $i$ feasible when others play $\mathbf b_{-i}$.}
    \Let{$X_1 \gets \bigcup_{\z \in \Z_i(\mathbf b_{-i})} \{b_i^\z\}$}\;
    \tcp{$X_1$: $i$'s observed actions in environment $\mathbf b_{-i}$.}
    \Let{$Y_1 \gets X_0 \setminus X_1$}\;
    \tcp{$Y_1$: actions never chosen by $i$ under $\mathbf b_{-i}$.}
    $\ell \gets 1$\;
    
    \While{$X_\ell \neq X_{\ell-1}$}{
      \Let{$Y_{\ell+1} \gets \{\, x\in X_\ell : \forall \z \in \Z_i(\mathbf b_{-i}) 
             \textnormal{ with } b_i^\z = x, z^\z \ge |B_i^\z|-k_i \,\}$}\;
        where $z^\z := |B_i^\z \cap (Y_1\cup\cdots\cup Y_\ell)|$\;
      \Let{$X_{\ell+1} \gets X_\ell \setminus Y_{\ell+1}$}\;
      $\ell \gets \ell+1$
    }

    \If{$X_\ell = \varnothing$}{continue}
    \Else{\Return{False}}
  }
}
\Return{True}
\caption{Testing for $\mathbf k$-rationalizability in games}\label{alg:main}
\end{algorithm}

\vspace{0.2cm}

The ``while'' loop is the central part of \cref{alg:main}. Intuitively, $Y_1$ contains the worst elements. $Y_2$ is a set containing the elements that can be consistently assigned as being ``next'' worst. More generally, $Y_\ell$ contains the elements that can be consistently assigned as $\ell$th worst. Observe, moreover, that the loop terminates with $X^\star_i(\mathbf{b}_{-i}):=X_{\ell + 1} = X_{\ell}$ for some finite $\ell \leq \Z_i(\mathbf{b}_{-i})$. \citet*{Bar22} show that $i$'s choices when faced with the environment $\mathbf{b}_{-i}$ are $k_i$-rationalizable if and only if $X^\star_i(\mathbf{b}_{-i}) = \varnothing$. Indeed, if $X^\star_i(\mathbf{b}_{-i}) = \varnothing$, then all observed choices have been consistently assigned to $Y_\ell$ for some $\ell$. In this case, one can reconstruct a complete, transitive preference $\succsim_i^{\mathbf b_{-i}}$ over $A_i$ that rationalizes the observed choices as $k_i$-satisficing. However, if $X^\star_i(\mathbf{b}_{-i}) \neq \varnothing$, then some observed choices are ``stuck'': for at least one feasible set, they cannot be placed among the top $k_i$ elements without contradiction. Hence no such rationalizing preference exists. Further explanation and details regarding this algorithm can be found in \citet*{Bar22}.

Since, for a given $i$ and $\mathbf{b}_{-i}$, the algorithm runs in at most $|X| \leq |A_i|$ iterations, it is polynomial in $|A_i|$. Performing the test independently for each $i\in N$ and each environment $\mathbf b_{-i}\in\mathbf B_{-i}$, the dataset $\{(\gf^\z,\mathbf b^\z)\}_{\z=1}^\Z$ is $\mathbf k$-rationalizable if and only if every such test returns $X^\star_i(\mathbf{b}_{-i})=\varnothing$.  
Therefore, there exists a polynomial-time algorithm (in $|N|$, $\Z$, and $\max_i |A_i|$) that decides whether the dataset is $\mathbf k$-rationalizable. \hfill \qed

\subsubsection{Proof of \cref{cor:testing}}
In the proof of \cref{thm:testing}, we reduced testing $\mathbf{k}$-rationalizability to testing 
$k_i$-rationalizability separately for each agent $i \in N$.   
To find the minimal $k_i$, we can perform a binary search over the interval $\{1,\dots,|A_i|\}$. This is justified by the monotonicity of $k_i$-rationalizability (namely that $k_i$-rationalizability implies $k_i'$-rationalizability for $k_i \leq k_i'$). At each step of the search, we test $k_i$-rationalizability using the polynomial-time \cref{alg:main}. Since binary search requires on the order of $\log |A_i|$ such tests, the overall cost of identifying the minimal $k_i$ remains polynomial in the number of actions. Repeating this procedure independently for each $i\in N$ yields the minimal vector $\mathbf{k}=(k_i)_{i\in N}$. Thus the smallest $\mathbf{k}$ for which the dataset is $\mathbf{k}$-rationalizable can be computed in time polynomial in $|N|$, $\Z$, and $\max_i |A_i|$.\hfill \qed

\singlespacing

\end{document}


\maketitle

\onehalfspacing

\tableofcontents

\newpage
\section{Axiomatization}\label{sec:axioms}

We now provide an axiomatization of satisficing equilibrium. For simplicity, we only axiomatize $\V{k}{k}$-satisficing equilibrium, i.e., where $k_i=k$ for all agents $i$. However, similar arguments to those given below allow an axiomatization of the heterogeneous case.

We characterize satisficing equilibrium with three axioms. Our first axiom states that agents exhibit $k$-satisficing behavior when in single-agent settings. Our other two axioms are consistency axioms that relate solutions of larger games to solutions of smaller games and vice versa. We show that these three axioms uniquely characterize the set of $\V{k}{k}$-satisficing equilibria.

\cite*{peleg1996consistency} give an axiomatization of Nash equilibrium based on the two consistency axioms and on optimizing behavior in single-agent settings. Our result builds on \cite*{peleg1996consistency} by relaxing only optimizing behavior in single-agent settings to satisficing behavior, but we keep the other two axioms.

To streamline the exposition, we slightly alter the definition of a game in this section to allow for arbitrary labeling of agents and actions. Here, an (arbitrarily labeled) ordinal game is a tuple
\[
\left( \;N, \, \{A_i\}_{i \in N}, \, \{\succsim_i\}_{i \in N} \; \right)
\]
where $N$ is a non-empty finite set of agents, each agent $i\in N$ has a finite set of actions $A_i$, and $\succsim_i$ is a preference relation on the set of action profiles $\profileset$.\footnote{Note that this differs from an ordinal game $\g$ as defined in Section 2 in the main text because the set of agent labels would then necessarily be $\{1,\ldots,n\}$ and the set of action labels for each agent $i$ would necessarily be $\{1,\ldots,m_i\}$, whereas no such labeling is imposed here.}

For any integers $\barN \geq 1$ and $\barM\geq 2$, let $\Gamma_{\barN,\barM}$ denote the set of all unlabeled ordinal games with agents from the set $\{1,...,\barN\}$ and at most $\barM$ actions per agent:
\[
\Gamma_{\barN,\barM} := \{ (N,\{A_i\}_{i \in N}, \{\succsim_i\}_{i \in N}) : \varnothing \neq N \subseteq \{1,...,\barN\} \text{ and } |A_i| \leq \barM \text{ for all } i \in N \} .
\]

A \emph{solution} on $\Gamma_{\barN,\barM}$ is a function $\phi$ that assigns to each game $g \in \Gamma_{\barN,\barM}$ a set of action profiles of $g$. 

\emph{One-person $k$-satisficing behavior} states that agents exhibit $k$-satisficing behavior when in a single-agent decision setting. Formally, $\phi$ satisfies one-person $k$-satisficing behavior if for any single-agent game $g = (\{i\}, \actionset, \succsim_i)\in \Gamma_{\barN,\barM}$ where $i\in\{1,\ldots,\barN\}$,
\begin{flalign*}
   & \phi(g)=\{x\in \actionset \, : \, \text{$i$ is $k$-satisficed at $x$} \} &
\tag*{\text{(\emph{one-person $k$-satisficing behavior})}} 
\end{flalign*}
In other words, one-person $k$-satisficing behavior requires each agent to choose one of her top $k$ actions when faced with a single-agent decision problem. One-person $1$-satisficing behavior in equivalent to optimizing behavior in single-agent decision problems.

 The next two consistency axioms are exactly those of \cite*{peleg1996consistency}. For any game $g =(N,\{A_i\}_{i \in N}, \{\succsim_i\}_{i \in N}) \in \Gamma_{\bar{n},\bar{m}} $, any non-empty $S \subset \agentset$, and any $\actionvec \in \profileset$, define $g^{S,\actionvec} \in \Gamma_{\bar{n},\bar{m}}$ as the game in which each agent $i \in \agentset \setminus S$ is fixed to selecting $a_i$, and the remaining agents in $S$ play the residual game. Denote by $\actionvec_S$ the profile of actions in $\actionvec$ that are taken by agents in $S$.

 \emph{Consistency} states that for any action profile $\actionvec$ that is contained in the solution of a game $\game$, and for any subset of agents $S$, $\actionvec_S$ is  contained in the solution of the restricted game where all agents outside of $S$ play according to $\actionvec$.
 Formally, $\phi$ satisfies consistency if for any game $g =(N,\{A_i\}_{i \in N}, \{\succsim_i\}_{i \in N}) \in \Gamma_{\barN,\barM}$,
\begin{flalign*}
  &\actionvec \in \phi(\game) \Rightarrow \actionvec_S\in \phi(\game^{S,\actionvec}) \text{ for all } \varnothing  \neq S\subset \agentset &
\tag*{\text{(\emph{consistency})}}  
\end{flalign*}

\emph{Converse consistency}, as the name suggests, 
is the counterpart to consistency and requires that if an action profile $\actionvec$ is a solution for all proper sub-games, then it is also a solution of the game itself.
Formally, a solution $\phi$ satisfies converse consistency if for any game $g =(N,\{A_i\}_{i \in N}, \{\succsim_i\}_{i \in N}) \in \Gamma_{\barN,\barM}$, 
\begin{flalign*}
  & \actionvec\in \phi(\game) \Leftarrow \actionvec_S\in \phi(\game^{S,\actionvec}) \text{ for all } \varnothing  \neq S\subset \agentset &
\tag*{\text{(\emph{converse consistency})}} 
\end{flalign*}

The three axioms above fully characterize satisficing equilibrium.

\begin{theorem*}
A solution $\phi$ on $\Gamma_{\barN,\barM}$ satisfies one-person $k$-satisficing behavior, consistency, and converse consistency if and only if it is identical to $\V{k}{k}$-satisficing equilibrium.
\end{theorem*}

This result provides strong support to our approach of building a theory for games based on decision-theoretic findings. Together with natural consistency axioms, a single-agent decision-theoretic axiom uniquely determines the solution of a game. Observe, moreover, that our result implies that any refinement of satisficing equilibrium does not satisfy the axioms. 

\begin{proof}
The arguments below adapt the proof of \cite*{peleg1996consistency}. Denote the $\V{k}{k}$-satisficing equilibrium solution concept by $k\text{-SAT}$. The theorem follows from points (i)-(iii) below. 

\begin{enumerate}[(i),leftmargin=0cm]
    \item Clearly, $k\text{-SAT}$ satisfies one-person $k$-satisficing behavior, consistency, and converse consistency.
    \item We prove that if $\phi$ satisfies one-person $k$-satisficing behavior and consistency then for any game $\game$, and any $\actionvec \in \phi(\game)$, we have $\actionvec \in k\text{-SAT}(\game)$. We do this by induction on the number of agents. Suppose $g$ is a single-agent game. By one-person $k$-satisficing behavior of $\phi$, if  $\actionvec \in \phi(\game)$ then $\actionvec \in k\text{-SAT}(\game)$. For the inductive step, suppose that $\actionvec \in \phi(h)$ implies $\actionvec \in k\text{-SAT}(h)$ for every game $h$ with fewer than $n$ agents. Now consider a game $g$ with agent set $\agentset$ such that $|\agentset|=n$, and suppose $\actionvec \in \phi(\game)$. By consistency of $\phi$, we have that $\actionvec_S \in \phi(g^{S,\actionvec})$ for each $\varnothing \neq S \subset \agentset$. By the inductive step, we know that $\actionvec_S \in k\text{-SAT}(\game^{S,\actionvec})$ for each $\varnothing \neq S \subset \agentset$. Finally, since $k\text{-SAT}$ satisfies converse consistency, $\actionvec \in k\text{-SAT}(\game)$.
    \item We prove that if $\phi$ satisfies one-person $k$-satisficing behavior and converse consistency then for any game $\game$, and any $\actionvec \in k\text{-SAT}(\game)$, we have $\actionvec \in \phi(\game)$. The argument is exactly the same as for point (ii) with $\phi$ and $k\text{-SAT}$ switched.
\end{enumerate}
\noindent Points (ii) and (iii) together imply that if $\phi$ that satisfies one-person $k$-satisficing behavior, consistency, and converse consistency, then $\phi = k\text{-SAT}$. Together with point (i), we obtain the desired conclusion.
\end{proof}

\section{Oligopoly production example}\label{app:congestiongame}

Let $F$ denote a finite set of primary factors (inputs). The oligopoly production environment of \citet[Example 2]{rosenthal1973class} is a game
\[
(\, N ,\, \{A_i\}_{i \in N}, \, \{\succsim_i\}_{i \in N}\,)
\]
where $N$ is a finite set of firms, and for each $i\in N$, $\varnothing \neq A_i \subseteq \{X : X \subseteq F\}$ and $\succsim_i$ is a preference relation over the set of action profiles $\profileset = \times_{i \in N}A_i$. An action $a_i \in A_i$ is therefore interpreted as the choice of a production process consisting of possibly many factor inputs. We define the \emph{load} $L_f(\actionvec)$ of a factor $f \in F$ at an action profile $\actionvec \in \profileset $ to be the number of firms employing factor $f$ at the action profile $\actionvec$. The \emph{cost} to agent $i$ of using factor $f$ is given by a function $c_f : \{1,2,\dots\} \to (0,\infty)$ and is identical for each agent.  

We next define the firms' preferences. \cite{rosenthal1973class} assumes that each firm prefers smaller over larger sums of the costs of the factors it employs (and is indifferent at equality). We uphold the assumption when preferences are strict. In addition, departing from \citeauthor{rosenthal1973class}, suppose that each firm $i\in\agentset$ has one specific competitor, $\gamma(i)$, with whom firm $i$ would prefer to share their process (provided that the sum of factor costs is minimized). This could be due to reputational benefits or standardization of tooling and workforce skills among downstream customers. Formally, each firm $i$'s preference relation $\succsim_i$ satisfies:
\begin{itemize}[leftmargin=1.5em]
\item If $\sum_{f \in a_i} c_f(L_f(a_i,\actionvec_{-i}))< \sum_{f \in a_i'}  c_f(L_f(a_i',\actionvec'_{-i}))$, then $\actionvec \succ_i \actionvec'$.
\item If $\sum_{f \in a_i}  c_f(L_f(a_i,\actionvec_{-i}))= \sum_{f \in a_i'}  c_f(L_f(a_i',\actionvec'_{-i}))$ and
\begin{itemize}
\item[$\circ$] $a_i = a_{\gamma(i)}$ and $a_i' \neq a_{\gamma(i)}'$, then $\actionvec\succ_i \actionvec' $;
\item[$\circ$]  $a_i = a_{\gamma(i)}$ and $a_i' = a_{\gamma(i)}'$, then $\actionvec\sim_i \actionvec' $;
\item[$\circ$]  $a_i \neq a_{\gamma(i)}$ and $a_i' \neq a_{\gamma(i)}'$, then $\actionvec\sim_i \actionvec' $.
\end{itemize}
\end{itemize}
Thus, a firm prefers processes with less congested factors, but, when congestion is identical, a firm prefers to operate the same process as its designated competitor $\gamma(i)$.

Because of the occasional strict preference for aligning with a specific competitor, this game is not an ordinal potential game, and it may not even admit a pure Nash equilibrium. For example, consider a simple instance with $F=\{p,q\}$, $N=\{1,2,3\}$ and for each $i \in N$, $A_i = \{\{p\},\{q\}\}$, that is, each process has a single factor. Suppose that both $c_p$ and $c_q$ are the identity function and let $\gamma(1)=3, \gamma(2)=1, \gamma(3)=2$. The action profiles 
\[
(p,q,p) \overset{3}{\to} (p,q,q)\overset{2}{\to} (p,p,q)\overset{1}{\to}(q,p,q) \overset{3}{\to}(q,p,p)\overset{2}{\to}(q,q,p)\overset{1}{\to}(p,q,p)
\] 
constitute a best-response cycle, and the resulting game is therefore not an ordinal potential game. The action profiles that are not visited in the above cycle are $(p,p,p)$ and $(q,q,q)$; but these profiles are not pure Nash equilibria and hence the game has no pure Nash equilibrium. 

That said, our oligopoly production game is a $\V{2}{2}$-approximate ordinal potential game and, by Theorem 3, therefore admits a $\V{2}{2}$-satisficing equilibrium. To see this, consider \cite{rosenthal1973class}'s potential function
\[
\Phi(\actionvec):=\sum_{f\in F}\sum_{\ell=1}^{L_f(\actionvec)} c_f(\ell),
\]
and let $\potpref$ be the weak order on action profiles induced by $\Phi$, that is, $\actionvec \potpref \actionvec'$ if and only if $\Phi(\actionvec)\le \Phi(\actionvec')$. Then
\[
\rank_i(\actionvec\mid \, \succsim_i)< \rank_i(\actionvec\mid\potpref)+2,
\]
where we recall that $\rank_i$ counts rank only among action profiles attainable by unilateral deviations of firm $i$. Indeed, for any unilateral deviation $(a_i,\actionvec_{-i})$, if firm $i$'s total factor cost is strictly lower at $(a_i,\actionvec_{-i})$ than at $\actionvec$, then $\Phi(a_i,\actionvec_{-i})<\Phi(\actionvec)$. Hence, the only unilateral deviations by agent $i$ that can be ranked strictly above $\actionvec$ under $\succsim_i$ but yield indifference under $\potpref$ are those for which firm $i$'s total factor cost is unchanged and firm $i$ employs the same process as competitor $\gamma(i)$ at $(a_i,\actionvec_{-i})$, while it does not at $\actionvec$.
Equivalently, any unilateral deviation by agent $i$ that is strictly preferred to $\actionvec$ under $\succsim_i$ but yields indifference under $\potpref$ must be the deviation to $a_{\gamma(i)}$, if this action is feasible for firm $i$; hence there is at most one such deviation. Therefore the rank discrepancy is at most one ($\rank_i(\actionvec\mid\succsim_i)\leq \rank_i(\actionvec\mid\potpref)+1$), and the displayed inequality follows.

\section{Bimatrix representations}\label{app:games}

The the 11--20 game and the Double or Dozen game are discussed in Section 1 of the main text and the Halving game and the Traveler's Dilemma are discussed in Appendix A of the main text. For completeness, we show the bimatrix representation of their cardinal versions.  In each cell, under each payoff pair, we also show the rank of the outcome within the set of each agent's unilateral deviations. For example, in the 11--20 game, at the action profile $(18,17)$, agent 1 gets her  fourth-best outcome given agent 2's choice, while agent 2 gets her  best outcome given agent 1's choice; therefore, we write ``$4,1$'' below the payoffs ``$18,37$''.

\begin{table}[!htbp]
\pgfplotstableread[col sep=space, header=false]{
20  20  20  20  20  20  20  20  20  20
39  19  19  19  19  19  19  19  19  19
18  38  18  18  18  18  18  18  18  18
17  17  37  17  17  17  17  17  17  17
16  16  16  36  16  16  16  16  16  16
15  15  15  15  35  15  15  15  15  15
14  14  14  14  14  34  14  14  14  14
13  13  13  13  13  13  33  13  13  13
12  12  12  12  12  12  12  32  12  12
11  11  11  11  11  11  11  11  31  11
}\MElevenTwenty

\pgfplotstableread[col sep=space, header=false]{
2   2   2   2   2   2   2   2   2   1
1   3   3   3   3   3   3   3   3   2
3   1   4   4   4   4   4   4   4   3
4   4   1   5   5   5   5   5   5   4
5   5   5   1   6   6   6   6   6   5
6   6   6   6   1   7   7   7   7   6
7   7   7   7   7   1   8   8   8   7
8   8   8   8   8   8   1   9   9   8
9   9   9   9   9   9   9   1  10   9
10  10  10  10  10  10  10  10   1  10
}\AElevenTwenty
{\centering 
\caption{Bimatrix representation of the 11--20 game.}
\vspace{0.2cm}
\begin{tikzpicture}[scale=1]
  \draw[step=1cm,thin,gray!50,shift={(0.5,-0.5)}] (0,0) grid (10,-10);
  \draw[thin,black] (0.5,-0.5) rectangle (10.5,-10.5);

\tikzmath{\del = 0.05; \eps1 = 1.5; \eps2 = 0; \eps3 = 1.0;}

\draw[very thick,draw=none, fill=HTML1,rounded corners = 0.5mm,opacity=1]
(0.5 - \eps1*\del ,-4.5)
-- (3.5,-4.5)
-- (3.5,-3.5)
-- (4.5,-3.5)
-- (4.5,-0.5 + \eps1*\del)
-- (0.5 - \eps1*\del ,-0.5 + \eps1*\del)
-- cycle;  

\draw[very thick,draw=none, fill=HTML4,rounded corners = 0.8mm,opacity=1]
(0.5+ \eps2*\del,-3.5)
-- (2.5,-3.5)
-- (2.5,-2.5)
-- (3.5,-2.5) 
-- (3.5,-0.5- \eps2*\del)
-- (0.5 + \eps2*\del,-0.5- \eps2*\del)
-- cycle;

\draw[very thick,draw=none, fill=HTML0,rounded corners = 1mm,opacity=1]
(0.5+\eps3*\del,-2.5) 
-- (1.5,-2.5) 
-- (1.5,-1.5) 
-- (2.5,-1.5)
-- (2.5,-0.5 - \eps3*\del)
-- (0.5 + \eps3*\del,-0.5 - \eps3*\del)
-- cycle;

  \foreach \j in {1,...,10} {%
    \pgfmathtruncatemacro{\r}{\j-1} 
    \foreach \i in {1,...,10} {%
      \pgfmathtruncatemacro{\c}{\i-1} 

      \pgfplotstablegetelem{\r}{[index]\c}\of{\MElevenTwenty}\let\Mij\pgfplotsretval
      \pgfplotstablegetelem{\c}{[index]\r}\of{\MElevenTwenty}\let\Mji\pgfplotsretval
      \pgfplotstablegetelem{\r}{[index]\c}\of{\AElevenTwenty}\let\Aij\pgfplotsretval
      \pgfplotstablegetelem{\c}{[index]\r}\of{\AElevenTwenty}\let\Aji\pgfplotsretval

      \node[font=\footnotesize, black!70] at (\i,-\j) {$\underset{\Aij,\Aji}{\Mij,\Mji}$};
    }%
  }%

\foreach \j in {1,...,4} {%
  \pgfmathtruncatemacro{\r}{\j-1} 
  \foreach \i in {1,...,4} {%
    \pgfmathtruncatemacro{\c}{\i-1} 
    \pgfmathtruncatemacro{\skip}{(\i==4) && (\j==4)}%
    \ifnum\skip=0
      \pgfplotstablegetelem{\r}{[index]\c}\of{\MElevenTwenty}\let\Mij\pgfplotsretval
      \pgfplotstablegetelem{\c}{[index]\r}\of{\MElevenTwenty}\let\Mji\pgfplotsretval
      \pgfplotstablegetelem{\r}{[index]\c}\of{\AElevenTwenty}\let\Aij\pgfplotsretval
      \pgfplotstablegetelem{\c}{[index]\r}\of{\AElevenTwenty}\let\Aji\pgfplotsretval
      \node[font=\footnotesize,black] at (\i,-\j)
        {$\underset{\Aij,\Aji}{\Mij,\Mji}$};
    \fi
  }%
}%

  \foreach \i in {1,...,10} {
    \node[font=\small] at (0,-\i) {\the\numexpr21-\i\relax};
    \node[font=\small] at (\i,0) {\the\numexpr21-\i\relax};
  }
\end{tikzpicture}

\label{tab:11-20}
}
\vspace{0.4cm}
\noindent\begin{minipage}{1\textwidth} 
\small{\textbf{Note.} Agent 1 chooses rows and Agent 2 chooses columns. Each cell displays ``monetary outcome for Agent 1, monetary outcome for Agent 2'' in the top line, and displays ``outcome rank within the column for Agent 1, outcome rank within the row for Agent 2'' in the bottom line.
The $(2,2)$-satisficing equilibria are the action profiles in $\{(x,y) : x,y \in \{19,20\}\}\setminus\{(19,19)\}$ and are shaded in red, the $(3,3)$-satisficing equilibria are $\{(x,y) : x,y \in \{18,19,20\}\}\setminus\{(18,18)\}$ and are shaded in gold (and red), and the  $(4,4)$-satisficing equilibria are $\{(x,y) : x,y \in \{17,18,19,20\}\}\setminus\{(17,17)\}$ and are shaded in  blue (and gold and red).}
\end{minipage}
\end{table}

\begin{table}[!htbp]
\pgfplotstableread[col sep=space, header=false]{
20 10 10 10 10 10 10 10 10 10
21 18  9  9  9  9  9  9  9  9
20 20 16  8  8  8  8  8  8  8
19 19 19 14  7  7  7  7  7  7
18 18 18 18 12  6  6  6  6  6
17 17 17 17 17 10  5  5  5  5
16 16 16 16 16 16  8  4  4  4
15 15 15 15 15 15 15  6  3  3
14 14 14 14 14 14 14 14  4  2
13 13 13 13 13 13 13 13 13  2
}\MDoD

\pgfplotstableread[col sep=space, header=false]{
 2 10  9  8  7  5  4  3  2  1
 1  3 10  9  8  7  5  4  3  2
 2  1  4 10  9  8  6  5  4  3
 4  2  1  5 10  9  8  6  5  4
 5  3  2  1  6 10  9  7  6  5
 6  5  3  2  1  5 10  9  7  6
 7  6  4  3  2  1  6 10  8  7
 8  7  6  4  3  2  1  7 10  8
 9  8  7  5  4  3  2  1  8  9
10  9  8  7  5  4  3  2  1  9
}\ADoD
{\centering 
\caption{Bimatrix representation of the cardinal Double or Dozen game.}
\vspace{0.2cm}
\begin{tikzpicture}[scale=1]
  \draw[step=1cm,very thin,gray!50,shift={(0.5,-0.5)}] (0,0) grid (10,-10);
  \draw[thin,black] (0.5,-0.5) rectangle (10+0.5,-10-0.5);

\tikzmath{\del = 0.05; \eps1 = 1; \eps2 = 1; \eps3 = 1.25;}
 \draw[very thick, draw=none, fill=HTML1,rounded corners=0.8mm,opacity=1]
    (0.5 - \eps3*\del,-1.5 - \eps3*\del)
    -- (1.5 -\eps3*\del,-1.5 - \eps3*\del)
    -- (1.5 -\eps3*\del,-2.5 - \eps3*\del)
    -- (2.5 -\eps3*\del,-2.5 - \eps3*\del)
    -- (2.5 -\eps3*\del,-3.5 -\eps3*\del)
    -- (3.5 + \eps3*\del,-3.5 -\eps3*\del)
    -- (3.5 +\eps3*\del,-2.5 + \eps3*\del)
    -- (2.5 + \eps3*\del,-2.5 + \eps3*\del)
    -- (2.5+\eps3*\del,-1.5 + \eps3*\del)
    -- (1.5 +\eps3*\del,-1.5 + \eps3*\del)
    -- (1.5 + \eps3*\del,-0.5 +\eps3*\del)
    -- (0.5 - \eps3*\del,-0.5 + \eps3*\del)
    -- cycle;

\draw[very thick, draw=none, fill=HTML4, rounded corners = 0.8mm,opacity=1]
(0.5,-1.5)
-- (1.5,-1.5)
-- (1.5,-2.5)
-- (2.5,-2.5)
-- (2.5,-1.5)
-- (1.5,-1.5)
-- (1.5,-0.5)
-- (0.5,-0.5)
-- cycle;

\draw[very thick, draw=none, fill=HTML0, rounded corners = 0.8mm,opacity=1]
  (0.5 + \eps1*\del,-1.5+\eps1*\del) rectangle (1.5-\eps1*\del,-0.5-\eps1*\del);

  \foreach \i in {1,...,10} {
    \node[font=\small] at (0,-\i) {\the\numexpr11-\i\relax};
    \node[font=\small] at (\i,0) {\the\numexpr11-\i\relax};
  }

  \foreach \j in {1,...,10} {%
    \pgfmathtruncatemacro{\r}{\j-1} 
    \foreach \i in {1,...,10} {%
      \pgfmathtruncatemacro{\c}{\i-1} 

      \pgfplotstablegetelem{\r}{[index]\c}\of{\MDoD}\let\Mij\pgfplotsretval
      \pgfplotstablegetelem{\c}{[index]\r}\of{\MDoD}\let\Mji\pgfplotsretval
      \pgfplotstablegetelem{\r}{[index]\c}\of{\ADoD}\let\Aij\pgfplotsretval
      \pgfplotstablegetelem{\c}{[index]\r}\of{\ADoD}\let\Aji\pgfplotsretval

      \node[font=\footnotesize, black!70] at (\i,-\j) {$\underset{\Aij,\Aji}{\Mij,\Mji}$};
    }%
  }%

\foreach \i in {1,...,3} {%
  \pgfmathtruncatemacro{\c}{\i-1} 
  \pgfmathtruncatemacro{\r}{\i-1} 
  \pgfplotstablegetelem{\r}{[index]\c}\of{\MDoD}\let\Mij\pgfplotsretval
  \pgfplotstablegetelem{\c}{[index]\r}\of{\MDoD}\let\Mji\pgfplotsretval
  \pgfplotstablegetelem{\r}{[index]\c}\of{\ADoD}\let\Aij\pgfplotsretval
  \pgfplotstablegetelem{\c}{[index]\r}\of{\ADoD}\let\Aji\pgfplotsretval

  \node[font=\footnotesize, black] at (\i,-\i) {$\underset{\Aij,\Aji}{\Mij,\Mji}$};
}%
\end{tikzpicture}

\label{tab:doubledozen}
}
\vspace{0.4cm}
\noindent\begin{minipage}{1\textwidth} 
\small{\textbf{Note.} Chacham chooses rows and Rasha chooses columns. Each cell displays ``number of dreidels for Chacham, number of dreidels for Rasha'' in the top line, and displays ``outcome rank within the column for Chacham, outcome rank within the row for Rasha'' in the bottom line. The unique $(2,2)$-satisficing equilibrium at the action profile $(10,10)$ is shaded in  red, the $(3,3)$-satisficing equilibria are $(10,10)$ and $(9,9)$ and are shaded in gold (and  red), and the $(4,4)$-satisficing equilibria are $(10,10), (9,9)$, and $(8,8)$ and are shaded in  blue (and  gold and  red).}
\end{minipage}
\end{table}

\begin{table}[!htbp]
\pgfplotstableread[col sep=space, header=false]{
10 10  10  10  10   0  10  10  10   4
9   9   9   9   9   9   0   9   9   4
8   8   8   8   8   8   0   8   8   4
7   7   7   7   7   7   7   0   7   4
6   6   6   6   6   6   6   0   6   4
11  5   5   5   5   5   5   5   0   5
4  11  11   4   4   4   4   4   0   4
3   3   3  11  11   3   3   3   3   0
2   2   2   2   2  11  11   2   2   0
1   1   1   1   1   1   1  11  11   1
}\MOneTen

\pgfplotstableread[col sep=space, header=false]{
2   2   2   2   2   10  2   2   2   2
3   3   3   3   3   2   9   3   3   2
4   4   4   4   4   3   9   4   4   2
5   5   5   5   5   4   3   9   5   2
6   6   6   6   6   5   4   9   6   2
1   7   7   7   7   6   5   5   9   1
7   1   1   8   8   7   6   6   9   2
8   8   8   1   1   8   7   7   7   9
9   9   9   9   9   1   1   8   8   9
10  10  10  10  10  9   8   1   1   8
}\AOneTen
{\centering 
\caption{Bimatrix representation of the cardinal Halving game.}
\vspace{0.2cm}
\begin{tikzpicture}[scale=1]
  \draw[step=1cm,thin,gray!50,shift={(0.5,-0.5)}] (0,0) grid (10,-10);
  \draw[thin,black] (0.5,-0.5) rectangle (10.5,-10.5);

\tikzmath{\del = 0.05; \eps1 = 1.5; \eps2 = 0; \eps3 = 1.0;}

\draw[very thick,draw=none, fill=HTML1,rounded corners = 0.5mm,opacity=1]
  (0.5 - \eps1*\del,  -0.5 + \eps1*\del)
    rectangle
  (3.5 + \eps1*\del,  -3.5 - \eps1*\del);

\draw[very thick,draw=none, fill=HTML4,rounded corners = 0.8mm,opacity=1]
  (0.5 + \eps2*\del,  -0.5 - \eps2*\del)
    rectangle
  (2.5 - \eps2*\del,  -2.5 + \eps2*\del);

\draw[very thick,draw=none, fill=HTML0,rounded corners = 1mm,opacity=1]
  (0.5 + \eps3*\del,  -0.5 - \eps3*\del)
    rectangle
  (1.5 - \eps3*\del,  -1.5 + \eps3*\del);

  \foreach \j in {1,...,10} {%
    \pgfmathtruncatemacro{\r}{\j-1} 
    \foreach \i in {1,...,10} {%
      \pgfmathtruncatemacro{\c}{\i-1} 

      \pgfplotstablegetelem{\r}{[index]\c}\of{\MOneTen}\let\Mij\pgfplotsretval
      \pgfplotstablegetelem{\c}{[index]\r}\of{\MOneTen}\let\Mji\pgfplotsretval
      \pgfplotstablegetelem{\r}{[index]\c}\of{\AOneTen}\let\Aij\pgfplotsretval
      \pgfplotstablegetelem{\c}{[index]\r}\of{\AOneTen}\let\Aji\pgfplotsretval

      \node[font=\footnotesize, black!70] at (\i,-\j) {$\underset{\Aij,\Aji}{\Mij,\Mji}$};
    }%
  }%

\foreach \j in {1,...,3} {%
  \pgfmathtruncatemacro{\r}{\j-1} 
  \foreach \i in {1,...,3} {%
    \pgfmathtruncatemacro{\c}{\i-1} 
    \pgfmathtruncatemacro{\skip}{(\i==4) && (\j==4)}%
    \ifnum\skip=0
      \pgfplotstablegetelem{\r}{[index]\c}\of{\MOneTen}\let\Mij\pgfplotsretval
      \pgfplotstablegetelem{\c}{[index]\r}\of{\MOneTen}\let\Mji\pgfplotsretval
      \pgfplotstablegetelem{\r}{[index]\c}\of{\AOneTen}\let\Aij\pgfplotsretval
      \pgfplotstablegetelem{\c}{[index]\r}\of{\AOneTen}\let\Aji\pgfplotsretval
      \node[font=\footnotesize,black] at (\i,-\j)
        {$\underset{\Aij,\Aji}{\Mij,\Mji}$};
    \fi
  }%
}%

  \foreach \i in {1,...,10} {
    \node[font=\small] at (0,-\i) {\the\numexpr11-\i\relax};
    \node[font=\small] at (\i,0) {\the\numexpr11-\i\relax};
  }
\end{tikzpicture}

\label{tab:110}
}
\vspace{0.4cm}
\noindent\begin{minipage}{1\textwidth} 
\small{\textbf{Note.} Agent 1 chooses rows and Agent 2 chooses columns. Each cell displays ``number of marbles for Agent 1, number of marbles for Agent 2'' in the top line, and displays ``outcome rank within the column for Agent 1, outcome rank within the row for Agent 2'' in the bottom line. The unique $(2,2)$-satisficing equilibrium at the action profile $(10,10)$ is shaded in  red, the $(3,3)$-satisficing equilibria are $\{(x,y) : x,y \in \{9,10\}\}$ and are shaded in gold (and red), and the $(4,4)$-satisficing equilibria are $\{(x,y) : x,y \in \{8,9,10\}\}$ and are shaded in blue (and gold and red).}
\end{minipage}
\end{table}

\begin{table}[!htbp]
\pgfplotstableread[col sep=space, header=false]{
11   8   7  6  5  4  3  2  1  0
12  10   7  6  5  4  3  2  1  0
11  11   9  6  5  4  3  2  1  0
10  10  10  8  5  4  3  2  1  0
9   9   9  9  7  4  3  2  1  0
8   8   8  8  8  6  3  2  1  0
7   7   7  7  7  7  5  2  1  0
6   6   6  6  6  6  6  4  1  0
5   5   5  5  5  5  5  5  3  0
4   4   4  4  4  4  4  4  4  2
}\MTD

\pgfplotstableread[col sep=space, header=false]{
2   5   5   5   5  5  5  4  3  2
1   2   5   5   5  5  5  4  3  2
2   1   2   5   5  5  5  4  3  2
4   2   1   2   5  5  5  4  3  2
5   4   2   1   2  5  5  4  3  2
6   5   4   2   1  2  5  4  3  2
7   7   5   4   2  1  2  4  3  2
8   8   8   5   4  2  1  2  3  2
9   9   9   9   5  4  2  1  2  2
10  10  10  10  10  5  4  2  1  1
}\ATD
{\centering 
\caption{Bimatrix representation of the ten action Traveler's Dilemma.}
\vspace{0.2cm}
\begin{tikzpicture}[scale=1]
\draw[step=1cm,thin,gray!50,shift={(0.5,-0.5)}] (0,0) grid (10,-10);
  \draw[thin,black] (0.5,-0.5) rectangle (10.5,-10.5);

 \draw[very thick,draw=none,fill=HTML0,rounded corners = 0.75mm,opacity=1, scale around={1.00:(5.5,-5.5)}]
    (0.5,-1.5)
    \foreach \i in {1,...,7}{
      -- ({\i+0.5}, {-\i-0.5}) 
      -- ({\i+0.5}, {-\i-1.5}) 
    }
    -- (7.5,-8.5)
    -- (8.5,-8.5)
    -- (8.5,-9.5)
    -- (7.5,-9.5)
    -- (7.5,-10.5)
    
    -- (10.5,-10.5)

    -- (10.5,-7.5) 
    -- (9.5,-7.5)
    -- (9.5,-8.5)
    -- (8.5,-8.5)
    \foreach \j in {9,...,2} {
      -- ({\j-0.5}, {-\j+0.5})
      -- ({\j-0.5}, {-\j+1.5})
    }
    -- (0.5,-0.5)
    -- cycle;
    
    \draw[very thick,draw=none,fill=HTML6,rounded corners = 0.75mm,opacity=1, scale around={0.9:(10,-10)}]
    (9.5 ,-10.5) rectangle (10.5,-9.5);
    
  \foreach \j in {1,...,10} {%
    \pgfmathtruncatemacro{\r}{\j-1} 
    \foreach \i in {1,...,10} {%
      \pgfmathtruncatemacro{\c}{\i-1} 

      \pgfplotstablegetelem{\r}{[index]\c}\of{\MTD}\let\Mij\pgfplotsretval
      \pgfplotstablegetelem{\c}{[index]\r}\of{\MTD}\let\Mji\pgfplotsretval
      \pgfplotstablegetelem{\r}{[index]\c}\of{\ATD}\let\Aij\pgfplotsretval
      \pgfplotstablegetelem{\c}{[index]\r}\of{\ATD}\let\Aji\pgfplotsretval

      \node[font=\footnotesize, black!70] at (\i,-\j) {$\underset{\Aij,\Aji}{\Mij,\Mji}$};
    }%
  }%

  \foreach \j in {1,...,10} {%
  \pgfmathtruncatemacro{\r}{\j-1} 
  \foreach \i in {1,...,10} {%
    \pgfmathtruncatemacro{\c}{\i-1} 

    \pgfmathtruncatemacro{\place}{
      (\i==\j) ||
      ((\i==10) && (\j==8 || \j==9)) ||
      ((\j==10) && (\i==8 || \i==9))
    }

    \ifnum\place=1
      \pgfplotstablegetelem{\r}{[index]\c}\of{\MTD}\let\Mij\pgfplotsretval
      \pgfplotstablegetelem{\c}{[index]\r}\of{\MTD}\let\Mji\pgfplotsretval
      \pgfplotstablegetelem{\r}{[index]\c}\of{\ATD}\let\Aij\pgfplotsretval
      \pgfplotstablegetelem{\c}{[index]\r}\of{\ATD}\let\Aji\pgfplotsretval

      \node[font=\footnotesize,black] at (\i,-\j)
        {$\underset{\Aij,\Aji}{\Mij,\Mji}$};
    \fi
  }%
}%
   
  \foreach \i in {1,...,10} {
    \node[font=\small] at (0,-\i) {\the\numexpr12-\i\relax};
    \node[font=\small] at (\i,0) {\the\numexpr12-\i\relax};
  }
\end{tikzpicture}

\label{tab:travelers-dilemma}
}
\vspace{0.4cm}
\noindent\begin{minipage}{1\textwidth} 
\small{\textbf{Note}: Agent 1 chooses rows and Agent 2 chooses columns. Each cell displays ``monetary outcome for Agent 1, monetary outcome for Agent 2'' in the top line, and displays ``outcome rank within the column for Agent 1, outcome rank within the row for Agent 2'' in the bottom line.
The unique $(1,1)$-satisficing equilibrium---that is, pure Nash equilibrium---at action profile $(2,2)$ is shaded in  teal, the $(2,2)$-satisficing equilibria  $(2,2),(3,3),\ldots,(11,11)$ and $(2,3)$, $(3,2)$, $(2,4)$, $(4,2)$ are shaded in  red (and  teal).}
\end{minipage}
\end{table}

\newpage

\section{Supplementary review of the literature}

\subsection{Solution concepts}\label{app:concepts}
This section provides formal definitions for the solution concepts from the literature that were discussed in Section 1 of the main text as well as descriptions of our computations of these solutions in specific games.

Since several concepts are defined only for cardinal games, we start by providing a definition of such games. For any positive integer $n$ and any vector of positive integers $\anumvec=(m_1,\ldots,m_n)$, a \emph{cardinal game} is a tuple
\[
\gcard := \left( \;N, \, \{A_i\}_{i \in N}, \, \{v_i\}_{i \in N} \; \right)
\]
consisting of a set of agents $N:=\{1,\ldots,n\}$, a set of actions $A_i:=\{1,\ldots,m_i\}$ for each agent $i\in\agentset$, and a utility function $v_i: \profileset \to \mathbb{R}$ for each $i \in \agentset$, where $\profileset := \times_{i \in \agentset} A_i$.

For any finite set $X$, let $\Delta X$ denote the set of all of probability distributions over $X$, with $y(x)$ denoting the probability assigned by $y \in \Delta X$ to the element $x \in X$. For any $i \in \agentset$ we will reserve the term ``action'' for elements of $A_i$ and ``strategy'' for elements of $\Delta A_i$. Agent $i$'s \emph{expected} utility at the strategy profile $\mathbf s \in \Delta \profileset$ is given by $u_i(\mathbf s) := \sum_{\actionvec \in \profileset} v_i(\actionvec) \, \mathbf s(\actionvec)$. Unless specified otherwise, we assume throughout the exposition below that strategy profiles have a product structure (i.e., they belong to $\times_{i \in \agentset} \Delta A_i$, which is a strict subset of $\Delta \profileset$).

\subsubsection*{Mixed Nash equilibrium} A strategy profile $\mathbf s$ is a \emph{mixed Nash equilibrium} if $s_i \in \arg \max_{x_i \in \Delta A_i} u_i(x_i,\mathbf s_{-i})$ for each $i \in N$.

For the games considered in this article (the Double or Dozen game, the Halving game, the 11--20 game, and the Traveler's Dilemma), the reported mixed Nash equilibria were computed with the Gambit software package \citep{gambit}.

\subsubsection*{Pure $\epsilon$-equilibrium} For $\epsilon \geq 0$, an action profile $\actionvec \in \profileset$ is a \emph{pure $\epsilon$-equilibrium} if for each $i \in \agentset$ and each $x_i \in A_i$, $v_i(\actionvec) \geq v_i(x_i,\actionvec_{-i}) - \epsilon$. For $\epsilon = 0$, we recover pure Nash equilibrium in cardinal games. Whenever $\epsilon >0 $, pure $\epsilon$-equilibrium is weaker than pure Nash equilibrium.

\subsubsection*{Rationalizability} An action for agent $i$ is \emph{rationalizable} if it is a best response to some (possibly correlated) distribution over the profiles of the opponents' rationalizable actions \citep{bernheim1984rationalizable,pearce1984rationalizable}. Rationalizability can be seen as a weakening of Nash equilibrium since  the actions in the support of any Nash equilibrium strategy must be rationalizable. Unlike (mixed) Nash equilibrium or (pure) $\epsilon$-equilibrium, rationalizability applies to the actions of individual agents rather than to profiles of actions (or of strategies).

\subsubsection*{Level$-\mathpzc{k}$ thinking} The level$-\mathpzc{k}$ framework posits that there are types $\{\mathpzc 0,\ldots,\mathpzc  K\}$ and a distribution $p \in \Delta \{\mathpzc 0,\ldots,\mathpzc  K\}$. An agent $i$ who is level$-\mathpzc 0$ selects an action from $A_i$ according to an exogenous rule (for example, a specific focal action, or uniformly at random). And, for each $k \in \{\mathpzc 1,\ldots,\mathpzc K\}$, a level$-\mathpzc{k}$ agent plays a best-response to a level$-(\mathpzc{k}-\mathpzc 1)$ agent (with ties, if any, broken according to some rule).  This framework for reasoning in games was proposed by \cite*{stahl1993evolution} and \cite*{nagel1995unraveling} \citep*[see also][]{crawford2013structural}. There are also related cognitive hierarchy models involving agents of type $\mathpzc k$ best-responding to a distribution over types $\mathpzc 0,\ldots,\mathpzc k-\mathpzc 1$ \citep*{camerer2004cognitive}.

\subsubsection*{Quantal response equilibrium} Each agent is endowed with a quantal response function $q_i: \mathbb{R}^{|A_i|} \to \Delta A_i$ which maps from vectors of expected payoffs from each action to probability distributions over actions. 
A strategy profile $\mathbf s$ is a \emph{quantal response equilibrium} if for each $i \in \agentset$ and $a_i \in A_i$, we have $s_i(a_i) = q_i(a_i,\mathbf s_{-i})$.

A common choice for the quantal response function, which we adopt here, is the logit quantal response function given by
\[
q_i(a_i,\mathbf s_{-i}) \propto e^{\lambda u_i(a_i,\mathbf s_{-i})},
\]
for each $a_i \in A_i$ and each strategy profile of the opponents, $\mathbf s_{-i}$ (see, e.g.,  \citealp{mckelvey1995quantal}). In other words, if $i$'s opponents play a strategy $\mathbf s_{-i}$ then $i$ selects action $a_i$ with a probability that is proportional to the exponential of the expected payoff of playing $a_i$ against $\mathbf s_{-i}$. The ``rationality'' parameter $\lambda \geq 0$ modulates how close agents are to optimizing behavior: $\lambda \to \infty$ corresponds to full optimization while $\lambda=0$ corresponds to uniform randomization over $A_i$.

For the games considered in this article (Double or Dozen, Halving, 11--20, and the Traveler's Dilemma), the reported logit quantal response equilibria for specific values of $\lambda$ were computed with the Gambit software package \citep{gambit}. Since the sensitivity of the predicted distributions of play vary by game, the values of $\lambda$ that we report vary across the different games.

\subsubsection*{Sampling equilibrium} This solution concept is due to \cite*{osborne1998games}. A strategy profile $\mathbf s$ is a \emph{sampling equilibrium} if for each $i \in N$ and each $a_i \in A_i$, $s_i(a_i) = w_i(a_i, \mathbf s_{-i})$, where the ``winning probability'' $w_i(a_i, \mathbf s_{-i})$ is the probability that $i$ finds action $a_i$ to be the best when $i$ samples each action once and $i$'s opponents plays according to $\mathbf s_{-i}$ (and any ties are weighted by the reciprocal of the number of tied actions).

Observe that, since it relies only on sign comparisons but not on magnitudes, sampling equilibrium is a purely ordinal concept. However, computing it requires a complete ordering of the action profiles for every agent, as opposed to say pure Nash equilibrium or satisficing equilibrium, which require comparisons only along unilateral deviations.

Sampling equilibria are the fixed points of a possibly complicated system of non-linear equations. In all the games considered in this article (which are symmetric), we report only symmetric sampling equilibria. The numbers we report for the Traveler's Dilemma were computed by \citet*{berkemer2023sampling}. For the Double or Dozen , the Halving game, and for the 11--20 game, we performed our own computations (see \url{https://github.com/btarbush/computing_sampling_equilibrium} for details).

\subsection{Bounded rationality in decision theory}\label{app:literature}

Here we review the decision-theoretic foundations of satisficing behavior, which, as outlined in Sections 1 and 4 of the main text, underpin the assumption that agents choose among their top actions.

The literature on consideration sets posits that a decision maker does not optimize over the whole choice set, but only over a smaller consideration set.
This literature goes back to at least \cite{varian1980model}'s classical treatment. More recently, see, e.g., \citet*{eliaz2011consideration,masatlioglu2012revealed, manzini2007sequentially,manzini2014stochastic,lleras2017more} for theoretical treatments and methods to empirically test for consideration sets. Empirical work confirms the explanative power of the theory of consideration sets.  For example, \cite*{goeree2008limited} study high markups in the PC industry that are sustained through advertising, \cite*{abaluck2021consumers} perform a lab experiment to validate their modeling choices and apply their findings to explain default policies in Medicare, and \citet*{barseghyan2021heterogeneous}  apply their theoretical model to explain households' choices on auto collision insurance.

In a related strand, the complexity of evaluating choices is modeled as influencing choice behavior, with alternatives that are less complex or costly to evaluate chosen more frequently. 
Early theoretical work on choice depending on state complexity in repeated games includes \cite*{neyman1985bounded,rubinstein1986finite,kalai1988finite}. \cite*{salant2011procedural} theoretically shows that when taking complexity into account, satisficing behavior emerges. Experimental support can be found, for example, in \citet*{gabaix2006costly,oprea2020makes,enke2023cognitive}. \cite*{salant2023complexity}, next to a theoretical and experimental treatment, also find support for emerging satisficing behavior in field data from chess play.

Similarly, work on search finds that costly search is terminated when a satisficing action is identified.
See \cite*{rubinstein2006model} for a theoretical treatment of choice from lists.  
\citet*{Caplin2011,reutskaja2011search,caplin2015revealed} provide experimental evidence for sequential search that terminates at a satisficing level of reservation utility, and \citet*{santos2012testing} find such evidence in field data  on web browsing and purchasing.

In another line of work, the guiding theme is that agents pay more attention to differences of first-order importance, again leading to satisficing behavior. 
\cite*{koszegi2013model} propose a theory of focus in economic choice and predict non-optimizing behavior. \cite*{gabaix2014sparsity} introduces a sparsity-based model where a decision-maker evaluates a simplified model, leading to satisficing behavior.
Relatedly, \cite*{tyson2021exponential} proposes a theory where an agent perceives her true preferences through a stochastic coarsening.
See also the related literature on salience, e.g., 
\citet*{bordalo2012salience,bordalo2022salience}.

Similarly work on status-quo biases finds that humans tend to stick to their current choice as long as other alternatives do not appear to be much better. 
\cite*{thaler1980toward} more broadly termed this the endowment effect. \cite*{samuelson1988status} coined the term status-quo bias and provided experimental evidence, studying decisions regarding health and retirement plans. For a review of the early literature, see \citet*{kahneman1991anomalies}.
\cite*{madrian2001power} present evidence for status-quo bias in choice of 401(k) retirement plans. For more recent theoretical treatments of reference-dependent behavior, see, e.g., \cite*{apesteguia2009theory,masatlioglu2014canonical}.

 To summarize, there are many reasons to expect people to exhibit (heterogeneous levels of) bounded rationality in decision-theoretic settings. 
 \citet*{Bar22} show that many theories of bounded rationality are observationally indistinguishable from modeling agents as selecting one of their top $k_i$ actions; the starting point of our analysis.

\subsection{Comments on available experimental data}\label{app:datasets}

The data sets we considered are listed in \cref{tab:data}. They all have at least one of the following properties, each of which leads to the data not being sufficient to meaningfully study satisficing equilibrium.
\begin{itemize}
    \item The choice of only one agent (e.g., the row agent) is recorded. This is sufficient to evaluate theories of thought such as level$-\mathpzc k$ thinking or rationalizability. However, as satisficing equilibrium does not incorporate a model of thought, realized play (of all agents) is necessary to evaluate agents' choices.
    \item Games with small numbers of actions are studied (i.e., with no more than three actions per agent). This is not sufficient to meaningfully evaluate satisficing equilibrium as, for example, in a two-agent, three-action game three to six action profiles are $(2,2)$-satisficing equilibria.
    \item Games are played once. Therefore, the dynamic interpretation of how satisficing equilibria may be reached cannot be tested.
\end{itemize}

\begin{table}[ht]
\centering
\caption{Experimental data sets from the literature.}
{\footnotesize
\begin{tabular}{@{}ll@{}}
\toprule
\textbf{Source} & \textbf{Description} \\[2pt]
\midrule
\citet*{stahl1994experimental} & Symmetric 3×3 games; record row's action only \\[2pt]
\citet*{stahl1995players} & Symmetric 3×3 games; record row's action only \\[2pt]
\citet*{costa2001cognition} & Small games (all smaller than 4×4) \\[2pt]
\citet*{goeree2001ten} & Various small games and Traveler's Dilemma \\[2pt]
\citet*{cooper2003evidence} & 2×2 games \\[2pt]
\citet*{haruvy2001modeling} & Symmetric 3×3 games \\[2pt]
\citet*{haruvy2007equilibrium} & Symmetric 3×3 games \\[2pt]
\citet*{costa2008stated} & 3×3 games \\[2pt]
\citet*{rogers2009heterogeneous} & Small, mostly 3×3 games \\[2pt]
\citet{wright2017predicting} & Meta-study on data sets for small games \\[2pt]
\citet{fudenberg2019predicting} & Random 3x3 games, record  row's action only \\[2pt]
\citet*{zhu2025capturing} & Random 2×2 games \\[2pt]
\bottomrule
\end{tabular}
}
\label{tab:data}
\end{table}

\newpage
\singlespacing